%% file: main.tex
\documentclass[10pt,twocolumn,twoside]{IEEEtran}
\usepackage{cite}
\usepackage{amsmath,amssymb,amsfonts,amsthm}
\usepackage{algorithmic}
\usepackage{graphicx}
\usepackage{textcomp}
\def\BibTeX{{\rm B\kern-.05em{\sc i\kern-.025em b}\kern-.08em
    T\kern-.1667em\lower.7ex\hbox{E}\kern-.125emX}}
\newtheorem{lemma}{Lemma}
\newtheorem{thm}{Theorem}

\newtheorem{cor}{Corollary}
\newtheorem{remark}{Remark}

\usepackage{pgfplots, tikz}
\pgfplotsset{compat=1.16}

\usetikzlibrary{plotmarks}

\usepackage{cite}

\PassOptionsToPackage{hyphens}{url}\usepackage[hidelinks]{hyperref}

\newcommand{\mc}{\mathcal}
\newcommand{\ol}{\overline}
\newcommand{\ul}{\underline}

\newtheorem{definition}{Definition}
\newtheorem{example}{Example}

\allowdisplaybreaks

\begin{document}
          
\title{Safe Schedule Verification for Urban Air Mobility Networks with Node Closures}
\author{Qinshuang Wei, \IEEEmembership{Student Member, IEEE}, Gustav Nilsson \IEEEmembership{Member, IEEE}, and Samuel Coogan, \IEEEmembership{Member, IEEE}
\thanks{Qinshuang Wei and Samuel Coogan are with the School of Electrical and Computer Engineering, Georgia Institute of Technology, Atlanta, 30332, USA. {\tt\small \{qinshuang, sam.coogan\}@gatech.edu}. S.~ Coogan is also with the School of Civil and Environmental Engineering, Georgia Institute of Technology.  }
\thanks{Gustav  Nilsson is with the School of Architecture, Civil and Environmental Engineering, École Polytechnique Fédérale de Lausanne (EPFL), 1015 Lausanne, Switzerland. {\tt\small gustav.nilsson@epfl.ch}}
\thanks{This work was partially supported by the NASA University Leadership Initiative (ULI) under grant number 80NSSC20M0161 and by the National Science Foundation under grant number 1749357.}
\thanks{Some preliminary results presented in this paper have been submitted to NecSys 22~\cite{Wei2022}.}
}

\maketitle

\begin{abstract}
In Urban Air Mobility (UAM) networks, takeoff and landing sites, called vertiports, are likely to experience intermittent closures due to, e.g., adverse weather. To ensure safety, all in-flight Urban Air Vehicles (UAVs) in a UAM network must therefore have alternative landing sites with sufficient landing capacity in the event of a vertiport closure. In this paper, we study the problem of safety verification of UAM schedules in the face of vertiport closures. We first provide necessary and sufficient conditions for a given UAM schedule to be safe in the sense that, if a vertiport closure occurs, then all UAVs will be able to safely land at a backup landing site. Next, we convert these conditions to an efficient algorithm for verifying safety of a UAM schedule via a linear program by using properties of totally unimodular matrices. Our algorithm allows for uncertain travel time between UAM vertiports and scales quadratically with the number of scheduled UAVs. We demonstrate our algorithm on a UAM network with up to 1,000 UAVs.
\end{abstract}


\section{Introduction}
\label{sec:intro}
Urban airspace is promising for transportating people and goods in cities and surrounding regions to avoid ground transportation congestion. 
Both commercial mobility-on-demand operators~\cite{2016uber} and government-sponsored research institutes such as NASA~\cite{thipphavong2018urban} are actively involved in developing such urban air mobility (UAM) solutions. Safety and efficiency of the urban air vehicles (UAVs) are major concerns in all UAM solutions~\cite{balakrishnan2018blueprint, 2014NextGen,2018landscape, inrix, al2020factors}. The work~\cite{al2020factors} observes that safety is one of the key factors affecting the adoption of UAM, while~\cite{balakrishnan2018blueprint,  2014NextGen,2018landscape} provide guidelines for safely integrating the UAVs into the existing airspace. 
The paper~\cite{inrix} provides insight into the improvement of commute efficiency with usage of urban airspace compared to ground transportation.
Proposed UAM solutions cover a wide range of possibilities such as allowing UAVs to land at \emph{vertistops} or \emph{vertiports} installed on roofs of existing buildings or within cloverleaf exchanges on freeways. In addition, a growing number of simulation tools have been developed to study large-scale interactions of UAVs~\cite{bosson2018simulation, xue2018fe3, 2019spark}. 

Unforeseen disruptions such as intermittent closure of landing sites due to, e.g., extreme weather conditions must be considered for any UAM solution~\cite{balakrishnan2018blueprint}. In particular, a key safety constraint is to ensure that a backup landing spot is available for all in-flight UAVs. In this paper, we model a UAM network as a graph with nodes that are finite-capacity vertiports and links that are transportation links between vertiports. A key feature of our model is the allowance of uncertain travel time between vertiports represented as an interval of possible travel times. Flights depart from origin nodes at a scheduled departure time and visit one or more vertiports along a route through the UAM graph. When a vehicle arrives at a vertiport, it occupies one of a finite number of landing spots for a fixed ground service time to, e.g, offload and load passengers. In this framework, the defining feature of safety is that a landing spot must always be available when the UAV arrives at the vertiport. The fact that travel times are uncertain adds to the complexity of the safety problem. In~\cite{wei2021scheduling}, we considered the problem of scheduling flight departures to ensure arrival at final destinations before prescribed deadlines while ensuring safety with respect to landing capacity throughout the network, but did not consider any vertiport closures which is the focus here.

In this paper, we assume given a schedule that is \emph{a priori} nominally safe obtained via, e.g., the methodology proposed in~\cite{wei2021scheduling}. Given such a schedule, the goal is to ensure that it remains safe even if a vertiport closes and in-flight UAVs must be rerouted. We assume that each link in the UAM network posseses a set of backup nodes such that any flight on that link that is inbound for a closed vertiport must be safely rerouted to one of those nodes at the moment of closure with the restriction that landing capacity is not exceeded for any node within the network. 

Our main contributions are as follows. First, we present necessary and sufficient conditions for ensuring safety in the event of a vertiport closure, i.e., for ensuring that all in-flight UAVs are able to land at a backup vertiport without exceeding landing spot capacity constraints. These conditions ensure safety for any realization of the link travel times, which are uncertain and only assumed to lie between known lower and upper bounds. We therefore refer to these conditions as worst-case safety guarantees. Second, we present an efficient algorithm for checking whether a schedule satisfies the theoretical necessary and sufficient conditions for worst-case safety. This algorithm leverages the theory of totally unimodular matrices to losslessly convert a mixed integer program into a linear program, enabling scalability to schedules with large numbers of UAVs. In particular, the proposed algorithm scales quadratically with the number of scheduled flights. Third, we present necessary and sufficient conditions for safety under some realization of the travel times. We refer to this as best-case safety, in contrast to worst-case safety which must be safe for all travel time realizations. These conditions, for example, could help a UAM operator determine if a schedule could be rendered safe by reducing travel time uncertainty. We demonstrate our results on several examples. This paper extends our prior work in~\cite{Wei2022} which only allowed for one backup node for each link in the network. Extending to multiple backup nodes is a significant generalization requiring the theory of totally unimodular matrices for an efficient algorithm that allows for checking a much larger class of safe schedules.

Safety of UAM scheduling has been explored in prior work such as~\cite{ancel2017real}, which presents a risk assessment framework to provide real-time safety evaluation where the risk of off-nominal conditions in a UAV is assessed by calculating the potential impact area and the effects of the impact to people on the ground. 

In ground transportation settings, most of the disruptions in the network can be modeled as capacity reductions, where totally disabled roads have zero capacity. The challenge is then to reroute the vehicle flows to ensure resilient operation of the network, where the flows are often assumed to be continuous quantities in the network~\cite{comopartI, comopartII}.

In this regard, our analysis is closer to classical airspace operation, where disruptions have previously been modeled and investigated to enable efficient recovery plans after the perturbations. Much of the existing literature focuses on generating a new recovery schedule~\cite{love2002disruption,zhu2005disruption,andersson2004flight,bisaillon2011large,wu2017rapid,thengvall2001multiple,yan1997airline}, rerouting aircrafts~\cite{arguello1997grasp,arguello1997framework,rosenberger2003rerouting,wu2017solving,wu2017solving2,li2013solving}, or  are integrated with recovering crew schedules~\cite{aguiar2011operational,zhang2015two,zhu2014constraint,nissen2006duty,vink2020dynamic,chen2020multiobjective} while minimizing a cost related to deviation to original schedules, available resources, and other system constraints. Other literature considers airport 
closures as disruptions~\cite{li2013solving,yan1997airline,thengvall2001multiple}. However, these works do not consider the capacity constraints of the airports, as needed here for the vertiports. Moreover, the present paper views the scheduling problem as a hard safety constraint rather than from the perspective of efficient operation.

The remainder of the paper is organized as follows: In Section~\ref{sec:problem}, we first define the UAM network model followed by the disruption model that reduces capacity of the network. 
We then establish safety criteria and develop necessary and sufficient conditions for a schedule to be safe under disruptions in Section~\ref{sec:thms}. We then develop an efficient algorithm to check that a schedule satisfies these conditions using the theory of totally unimodular matrices. In Section~\ref{sec:case_study}, we demonstrate our safety verification algorithm on a UAM network. The paper is concluded with some ideas for future work.

\section{Problem Formulation}
\label{sec:problem}
\subsection{Network Model and Nominal Scheduling}
We model an urban air mobility (UAM) network with a directed graph $\mc{G} = (\mc{V},\mc E)$, where $\mc V$ is the set of nodes and $\mc E$ is the set of links for the network. Nodes are physical landing sites for the UAVs, sometimes called \emph{vertistops} or \emph{vertiports}. Links are corridors of airspace connecting nodes. Each node $v \in \mc V$ has capacity $C_v \in \mathbb{N}_0$, that is, there are $C_v$ \emph{landing spots} at node $v$ where each landing spot allows at most one UAV to stay at any time. We denote the vector of capacities $C=\{C_v\}_{v\in\mathcal{V}}$.
 
We define $\tau: \mc E \rightarrow \mc V$ and $\sigma: \mc E \rightarrow \mc V$ so that for all $e=(v_1, v_2) \in \mc E$ where $v_1,v_2 \in \mc V$, $\tau(e) = v_1$ is the tail of link $e$ and $\sigma(e) = v_2$ is the head of link $e$. Let $S\subseteq \mathcal{V}$ (resp., $T\subseteq \mathcal{V}$) be the set of nodes that are not the head (resp., tail) of any link, $S=\{v\in\mathcal{V}\mid \sigma(e)\neq v\ \forall e\in\mathcal{E}\}$ and $T=\{v\in\mathcal{V}\mid \tau(e)\neq v\ \forall e\in\mathcal{E}\}$. We assume $S\cap T=\emptyset$.

A \emph{route} $R$ is a sequence of connected links. Denote the number of links in route $R$ by $k_R$ and enumerate the links in the route  $1^R,2^R,\ldots,k_R^R$ and the nodes in the route $0^R,1^R,\ldots,k^R_R$. To avoid cumbersome notation, we use $\ell^R$ to denote both a link and its head node along a route, i.e., $\ell^R=\sigma(\ell^R)$ for all $\ell\in\{1,\dots,k_R\}$; the intended meaning will always be clear from context. Thus the route links and nodes are enumerated so that $0^R=\tau(1^R)$ is the origin node, $k_R^R$ is the destination node, and $\sigma(\ell^R) = \tau((\ell+1)^R)$ for all $\ell\in\{1,\dots,k_R\}$ ensures the sequence is connected. Further, when the route $R$ is clear from context, we drop the superscript-$R$ notation. We denote the set of nodes that $R$ travels through as $V(R)$.
We assume that, due to operational reasons, the UAVs are only allowed to travel along a set of routes $\mc R$, and $0^R \in S$ and $k^R_R \in T$ for all $R \in \mc R$.

Since, in reality, the travel time depends on external factors such as weather conditions or a vehicle's operational capability, we assume that the travel time for each link is not exact, but rather bounded by a time interval. 
For each link $e \in \mc E$, let $\ol{x}_e$ and $\ul{x}_e$ with $\ol{x}_e \geq \ul{x}_e>0$ denote the maximum travel time and minimum travel time, respectively, for the link, and let $\ul{x}\in \mathbb{R}_+^\mathcal{E}$ and $\ol{x}\in \mathbb{R}_+^\mathcal{E}$ be the corresponding aggregated vectors. Once a UAV has landed at any node, it is assumed to block a landing spot for a fixed ground service time $w \in \mathbb{R}_+$. For ease of notation, we assume the ground service time is uniform at all nodes, but this assumption is straightforward to relax. 

\smallskip

\begin{definition}[UAM Network]
A \emph{UAM network} $\mc N$ is a tuple $\mc N = (\mc G, C, \mc R, \ul{x}, \ol{x},w)$ where $\mc G, C,\mc R, \ul{x}, \ol{x}, w$ are the network graph, node capacities, routes, minimum and maximum link travel times, and ground service time as defined above.
\end{definition}

To model the schedule of UAV flights in a UAM network $\mc N = (\mc G, C, \mc R, \ul{x}, \ol{x},w)$, we assume that every flight is associated to a route $R\in \mc R$ and stops at intermediate nodes along the route. Therefore, a \emph{schedule} is a pair $(R,\delta)$ where $R \in \mc R$ and $\delta \in \mathbb{R}_+$ is the appointed departure time from the first node along the route. A schedule profile for a UAM network is a set $\mc S = \{(R_j,\delta_j)\}_{j\in {\mc J}}$ where ${\mc J}$ is a finite index set of \emph{flights}.

For safety reasons, it is assumed that a UAV must be able to land immediately upon arrival at any node along its route. 
For flight $j$ with schedule $(R_j,\delta_j)$, for any link $\ell$ along route $R_j$, the latest arrival time at node $\sigma(\ell)$ along the route is denoted $a^{j}_{\ell}$ and given by
\begin{equation}
\label{eq:a_aggregate}
    a^{j}_{\ell} = \delta_j + \sum_{k=1}^{\ell} \ol{x}_{k}  + (\ell-1)w\,,
\end{equation}
i.e., $a^{j}_{\ell}$ is the departure time from node $0$ plus the upper bound of the time interval it takes to travel through the links $\{1,2,\dots,\ell\}$ with the time spent at each intermediate node.
Further, the time interval that the flight will potentially block a landing spot at node $\ell$ is given by
\begin{equation} \textstyle
\label{eq:M_span}
\mc M_{\ell}^j = \left[ \delta_{j} + \sum_{k=1}^{\ell} \ul{x}_{k}+ (\ell-1)w, \, a^{j}_{\ell} + w \right] .
\end{equation}
We let $\mc M_{v}^j = \mc M_{\ell}^j$ and $a^j_v = a^{j}_{\ell}$ if $v \in V(R_j)$ and $v = \ell^{R_j}$. 

\medskip
\begin{definition}[Feasible Schedule]
\label{def:schedule}
A set of schedules $\mc {S}=\{(R_j,\delta_j)\}_{j\in \mc{J}}$ where $\delta_j\in\mathbb{R}_+$ for all $j \in \mc J$ is a \emph{feasible schedule} if the number of vehicles at a node never exceeds capacity, i.e., for all $v\in{\mc V}$ and all $t \geq 0$,
\begin{equation}
    \sum_{j:v \in V(R_j)}\mathbf{1}\left(t;\mc M_{v}^j\right)\leq C_v
\end{equation}
where the notation $\mathbf{1}(\cdot;\cdot)$ is an indicator such that $\mathbf{1}(t;[a,b])=1$ if $t\in [a,b]$ and $\mathbf{1}(t;[a,b])=0$ otherwise.
\end{definition}
Since the time intervals defined by \eqref{eq:a_aggregate} and \eqref{eq:M_span} consider lower and upper bounds on the uncertain travel time, the definition of feasibility accommodates all possible travel times satisfying these lower and upper bounds, motivating the next definition.
\begin{definition}[Realization]
A \emph{realization} of a scheduled flight is a realization of the travel times such that the flight departs at the given departure time and has a fixed travel time along each link that falls within the given time interval for the link. While each realization of the same flight has the same departure time, different realizations generally have different travel times on at least one link due to uncertain travel times.
\end{definition}
A feasible schedule ensures that node capacity is not exceeded for any realization of scheduled flights. Every feasible schedule will by definition ensure proper operation of the UAM network under normal circumstances. Our goal in this paper is to check whether the schedule is further resilient to interruptions in the network.

\subsection{Disruption Model} 

In actual operation, it is expected that unforeseen disruptions that disable a node, such as adverse weather conditions, will be common. Flights affected by the disabled node 
must have a rerouting plan that ensures availability of a landing spot. In this paper, we postulate the existence of a set of \emph{backup nodes} for the network so that when any node is disabled, the flights can be redirected to a backup node depending on the link they are traveling through.

In this subsection, we introduce the assignment of the backup nodes and the operating mechanism once a node is disabled. We consider that only one node may be disabled at a time. In order to guarantee that each disrupted flight will be able to be assigned to a node after the disruption, we assign a set of backup nodes $\mc B_e$ to each link $e$ in the network. The assignment of backup nodes can be based on some rules, e.g., distances between nodes. 
We make a natural assumption that the set of backup nodes for any link  includes its tail node and head node, i.e., $\tau(e), \sigma(e) \in \mc B_e$ for any $e \in \mc E$. 
Then, a flight traveling on some link $e$ whose route is potentially blocked by a node closure will continue to the head node $\sigma(e)$ on its route if that node is functioning, or reroute to one of its backup nodes if the head node is disabled.

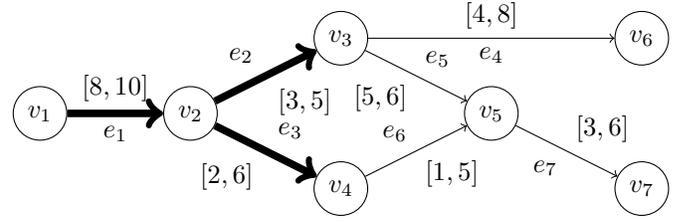
\begin{figure}
    \centering    
    \begin{tikzpicture}
    
    \node[draw, circle] (1) at (-1, 2) {$v_1$};
    \node[draw, circle] (2) at (1, 2) {$v_2$};
    
    \node[draw, circle] (3) at (3, 3) {$v_3$};
    \node[draw, circle] (4) at (3, 1) {$v_4$};
    \node[draw, circle] (5) at (5, 2) {$v_5$};
    
    \node[draw, circle] (6) at (7, 3) {$v_6$};
    \node[draw, circle] (7) at (7, 1) {$v_7$};

    \draw[->, line width=1mm] (1) -- node[above] {$[8,10]$} node[below]{$e_1$} (2);
    \draw[->, line width=1mm] (2) -- node[below right] {$[3,5]$} node[above left]{$e_2$} (3);
    \draw[->, line width=1mm] (2) -- node[below left] {$[2,6]$} node[above right]{$e_3$} (4);
    \draw[->] (3) -- node[above] {$[4,8]$} node[below]{$e_4$} (6);
    \draw[->] (3) -- node[below left] {$[5,6]$} node[above right]{$e_5$} (5);
    \draw[->] (4) -- node[below right] {$[1,5]$} node[above left]{$e_6$} (5);
    \draw[->] (5) -- node[above right] {$[3,6]$} node[below left]{$e_7$} (7);

    \end{tikzpicture}
    \caption{(Bold partial graph) The sub-graph consisting of the bold lined $4$ nodes and $3$ links is used to illustrate the simple network in Example~\ref{ex:safety}.\\
    (Entire graph) The entire graph is used to illustrate the network with $7$ nodes and $7$ links in Example~\ref{ex:case} and the case study.}
    \label{fig:case_study_net}
\end{figure}

\begin{example}
\label{ex:case}
Consider Fig.~\ref{fig:case_study_net} (the entire graph) with 7 nodes and 7 links,  $\mc V = \{v_1,v_2, \dots, v_7\}$ and $\mc E = \{e_1 = (v_1,v_2), e_2 = (v_2, v_3),e_3 = (v_2, v_4), e_4 = (v_3,v_6), e_5 = (v_3,v_5), e_6 = (v_4,v_5),e_7 = (v_5, v_7) \}$. The set of all possible origins (resp., destinations) is $S = \{v_1\}$ (resp., $T = \{ v_6, v_7\}$). We assume the origin $v_1$ does not have a capacity constraint, while $C_{v_2} = 8$, $C_{v_3} = 6$, $C_{v_4} = 4$, $C_{v_5} = 5$, $C_{v_6} = 3$ and $C_{v_7} = 5$. The links are indicated in the figure and the corresponding travel time intervals are labeled beside the links, e.g., the interval $[8,10]$ above the link $e_1$ means that the shortest (resp., longest) possible time for traveling through the link is $8$ (resp., $10$) time units. We consider three routes $\mc R = \{R^1, R^2, R^3\}$ with $R^1 = \{e_1, e_2, e_4\}$, $R^2 = \{e_1, e_2, e_5, e_7\}$ and $R^3 = \{e_1,e_3, e_6, e_7\}$. 
Each UAV remains at the verti-stops along its path for $w = 1$ time unit after landing.

Table~\ref{tb:backup} shows a possible assignment of backup nodes for the network. The first column represents the link $e \in \mc E$, the second column is the set of backup nodes assigned to the corresponding link, while the third column shows the node (or nodes) that the UAV on the link can be rerouted to if node $v_5$ is disabled. For example, although  link $e_1$ is not directly affected by the closure of node $v_5$, some flights using this link have a route that passes through $v_5$ and will therefore land at the head node $v_2$ and remain there due to the closure of $v_5$. A similar explanation holds for rows 2 and 3 of the table. 
Flights traveling on links $e_4$ and $(v_5,v_7)$ are not affected if $v_5$ fails, hence the corresponding entries in the third column are empty. Lastly, flights traveling on links $e_5$ and $e_6$ must instead route to one of the backup nodes as indicated.
\end{example}

\begin{table}
\centering
\caption{Backup nodes for each link in Example~\ref{ex:case}}
\label{tb:backup}
\begin{tabular}{ccc}
Link ($e\in \mc E$)  & $\mc B_e$ & Possible backup nodes\\
    &   & when $v_5$ is disabled \\\hline
$e_1 = (v_1,v_2)$ & $\{v_1,v_2\}$    &   $v_2$\\
$e_2 = (v_2, v_3)$& $\{v_2,v_3,v_4\}$    &   $v_3$\\
$e_3 = (v_2, v_4)$& $\{v_2,v_3,v_4\}$     &   $v_4$\\
$e_4 = (v_3,v_6)$& $\{v_3,v_5,v_6\}$     &   Not affected \\
$e_5 = (v_3,v_5)$& $\{v_3,v_4,v_5\}$    &   $v_3$, $v_4$\\
$e_6 = (v_4,v_5)$& $\{v_3,v_4,v_5,v_7\}$    &   $v_3$, $v_4$, $v_7$\\
$e_7 = (v_5,v_7)$ & $\{v_4,v_5,v_7\}$     &   Not affected\\ \hline
\end{tabular}
\end{table}



A realization of the $j$'th flight is \emph{affected} by some disabled node $v_c \in \mc V$ at time $t_c$ if $v_c \in V(R_j)$, i.e., the route of the flight travels through node $v_c$, and the flight has not yet reached $v_c$ by time $t_c$. The realization of the $j$-th flight is \emph{not affected} when node $v_c$ is disabled at time $t_c$ otherwise. The $j$'th flight is \emph{possibly affected} by disabling node $v_c$ at time $t_c$ if $v_c \in V(R_j)$ and $t_c < \sup{\mc M^j_{v_c}}$, i.e., the flight may have to travel through the disabled node later than $t_c$ and hence is affected for some realization of travel times.

Below is a set of natural rules that all flights are assumed to follow once a node $v_c$ is disabled at time $t_c$:
\begin{enumerate}
    \item flights not affected will continue normal operation;
    \item any affected flight that has not yet departed ($\delta_j > t_c$) will be canceled (no longer depart);
    \item an affected flight $j$ with $\delta_j \leq t_c$ traveling on a link $e \in \mc E$ with $\sigma(e) \neq v_c$ will continue to the head node $\sigma(e)$ and stop there indefinitely (block the landing spot indefinitely);
    \item an affected flight $j$ with $\delta_j \leq t_c$ that is temporarily stopped at a node at time $t_c$ will remain there indefinitely;
    \item  an affected flight $j$ with $\delta_j \leq t_c$ traveling on a link $e \in \mc E$ with $\sigma(e) = v_c$ will 
    be rerouted to one of the other backup nodes of the current link in $\mc B_e\backslash v_c$ and stop there indefinitely.
\end{enumerate}
Note that we do not consider the problem of recovering a new schedule after a disabled node becomes operational again, as our focus is on safety. Further, we postulate the above rules as to provide a well-defined problem formulation; alternative rules might be also plausible.

\section{Necessary and Sufficient Conditions for Safe Schedules}
\label{sec:thms}
In this section, we formally define safety and present sufficient and necessary conditions for verification of safety under different criteria. 

Given a network $\mc N = (\mc G, C, \mc R, \ul{x}, \ol{x}, w)$ where $\mc G = (\mc V, \mc E)$ and a feasible schedule $\mc S = \{(R_j,\delta_j)\}_{j \in \mc J}$, 
the closure of a node can affect the set of schedules in different ways. In particular, the set of schedules is:
\begin{enumerate}
    \item \emph{worst-case (resp., best-case) time-node conditionally safe} for node $v_c$ and time $t_c$ if, supposing that $v_c$ is disabled at time $t_c$, then all possibly affected flights are able to land at their designated backup nodes while not interfering with any unaffected flights, for all (resp., for some) realization of link travel times.
    \item \emph{worst-case (resp., best-case) node conditionally safe} for node $v_c$ if it is worst-case (resp., best-case) time-node conditionally safe for node $v_c$ for all time $t_c \geq 0$.
    \item  \ \emph{worst-case (resp. best-case) 1-closure safe} if it is worst-case (resp. best-case) node conditionally safe for any node $v_c \in \mc V$.
\end{enumerate}
 
Note that worst-case safety implies best-case safety.

    
%


\begin{example}
\label{ex:safety}
We illustrate the safety criteria through the simple network shown in Fig.~\ref{fig:case_study_net} with $4$ nodes and $3$ links (the bold lined sub-graph). For this example, the set of nodes $\mc V = \{v_1,v_2, v_3, v_4\}$ and the set of links $\mc E = \{e_1,e_2,e_3 \}$. We assume that the origin $v_1$ does not have a capacity constraint, while $C_{v_2} = 2$, $C_{v_3} = 1$ and $C_{v_4} = 1$. The links are indicated in the figure and the corresponding travel time intervals are labeled beside the links. We consider two possible routes $R^1 = \{e_1,e_2\}$ and $R^2 = \{e_1,e_3\}$. Each flight remains at the intermediate nodes or destination along its path for $w = 1$ time unit after landing. 
The backup nodes for each link are $\mc B_{e_1} = \{v_1,v_2\}$, $\mc B_{e_2} = \{v_2,v_3,v_4\}$ and $\mc B_{e_3} = \{v_2,v_3,v_4\}$.
Consider a feasible schedule $\mc S = \{S_1,S_2,S_3\}$, where $S_1 = (R^2,1)$, $S_2 = (R^2,8)$ and $S_3 = (R^1,\delta_3)$ where we consider several possibilites for $\delta_3$. Assume $v_4$ is disabled at time $t_c=15$. Based on the rerouting rules for the flights,  then at time $t_c$, a flight traveling on link $e_1$ will be rerouted to node $v_2$, while a flight traveling on link $e_3$ can be rerouted to either $v_2$ or $v_3$. Though it is possible that flight $S_1$ has already complete its journey by time $t_c=15$, in the worst case where we consider any link that it may be traveling on, it is possible for flight $S_1$ to be traveling on $e_3$ and need to be rerouted to $v_2$ or $v_3$, so that we have to reserve a landing spot at node $v_2$ or $v_3$ for $S_1$; flight $S_2$ must be traveling on $e_1$ and needs to stay at $v_2$ upon arrival. If $\delta_3 = 0$, then flight $S_3$ is not affected and should continue its journey; however, if $\delta_3 = 10$, then either $v_2$ or $v_3$ will have insufficient landing spots, since the flight $S_1$ must have been rerouted to either $v_2$ or $v_3$ upon the arrival of $S_3$. Therefore, $\mc S$ is worst-case time-node conditionally safe for node $v_4$ at time $15$ if and only if $\delta_3 \leq 4$. In contrast, it is always best-case time-node conditionally safe for node $v_4$ at time $15$ regardless of the choice of $\delta_3$.

Now, suppose $\delta_3 = 10$ and $C_{v_2} = 3$. Then there is sufficient capacity so that the network will be able to accommodate all rerouted flights after closure no matter when node $v_4$ is closed. Therefore, we see that $\mc S$ is worst-case node conditionally safe for node $v_4$ in this case. We further check that this is true for all nodes in the network, and thus $\mc S$ is also worst-case 1-closure safe. In contrast, suppose $C_{v_3} = 2$ while $C_{v_2} = 2$, then $\mc S$ is worst-case time-node conditionally safe for node $v_4$ and $t_c = 15$ with any choice of $\delta_3$, 
but is worst-case node-conditionally safe for $v_4$ if and only if $\delta_3 \geq 4$.
\qed\end{example}

To obtain constraints for 1-closure safety, we start by observing that a feasible schedule is trivially node conditionally safe for any node $v \in S$, where we recall $S$ the set of source nodes that are not the head of any link. Whenever a node $v\in S$ is disabled, there will not be any UAV traveling toward node $v$ while no future journey will depart from $v$.

We next explore safety of a disabled node that is not a source node. 
There are several special sets we now define before presenting conditions for 1-closure safety when disabling a node $v_c \in \mc V \backslash S$.
If $v_c = \ell^{R_j}$, we let $\ell_{v_c,R_j} = \ell$, and $\ell_{v_c,R_j} = 0$ if $v_c \notin V(R_j)$. 

We let the set of links with head $v \in \mc V$ be 
\begin{equation}
    \label{eq:links_affected}
    \mc E_v := \{e\in \mc E \mid \sigma(e) = v\}\, ,
\end{equation}
and we let $\mc B_{e,v_c}$ be the set of nodes that any flight traveling on link $e$ can be rerouted to if $v_c$ is disabled: 
\begin{equation}
    \mc B_{e,v_c} =\begin{cases}
    \{\sigma(e)\} & \text{if } \sigma(e)\neq v_c \,, \\
     \mc B_e\backslash v_c & \text{if } \sigma(e)= v_c \,.
    \end{cases}
\end{equation}
$B_{e,v_c}$ is the set of \emph{possible backup nodes} for link $e$ when $v_c$ is disabled.  
We then denote $b_{e,v_c}$ as the node that a flight traveling on link $e$ will be rerouted to if $v_c$ is disabled, so that $b_{e,v_c} \in \mc B_{e,v_c}$.

We denote the set of links on which flights will possibly be rerouted to node $v$ when $v_c$ is disabled as $\mc B_{v,v_c}$, which includes the links with head node as $v$ when $v \neq v_c$ and the links with head node as $v_c$ 
whose backup nodes include $v$, i.e.,
\begin{align}
     \nonumber \mc B_{v,v_c} :=& \{e \in \mc E \mid \sigma(e) =v, \sigma(e) \neq v_c  \}  \\
    \label{eq:backup2v}   &\cup \{e \in \mc E \mid  v\in \mc B_e \backslash v_c , \sigma(e) = v_c \}   \,.
\end{align}

We then define the set $\mc E^j_{v_c}$ as the links along the route of flight $j$ whose head node is one of the backup nodes of link $\ell^{R_j}_{v_c,R_j}$, i.e., 
\begin{equation}
    \mc E^j_{v_c} = \{\sigma(\ell^{R_j}) \in \mc B_{e,v_c} \mid \ell < \ell_{v_c,R_j}\} \, .
\end{equation}

We define the set $\mc J_v$ as the index set of the flights with routes passing through node $v$,
\begin{equation}
    \label{eq:j2v}
    \mc J_v := \{j \in \mc J \mid v \in V(R_j)\} \, ,
\end{equation}
and we further define the index set of the journeys that might possibly land at $v$ after time $t$ as
\begin{equation}
    \label{eq:Jstar}
    \mc J^*_v(t) := \{j\in \mc J_v \mid a^j_v+w > t\} \, .
\end{equation}
Therefore, the index set of the possibly affected flights when node $v_c$ is closed at time $t_c$ is $\mc J^*_{v_c}(t_c)$, while the index set for the flights passing through node $v$ that are not possibly affected when node $v_c$ is closed at $t_c$ is $\mc J_{p}(v,t_c,v_c) := \mc J_v \backslash \mc J^*_{v_c}(t_c)$.

We use $\mc J_{c,v_c}$ to represent the set of indices for canceled journeys with departure time greater than the node-disabling time $t_c$:
\begin{equation}
    \label{eq:Jcancel}
    \mc J_{c,v_c}(t_c) := \{j\in \mc J_{v_c} \mid \delta_j>t_c \} \, .
\end{equation}
We then define the index set of rerouting flights $\mc J^{*\backslash c}_{v_c}(t_c)$ as the possibly affected flights not canceled when node $v_c$ is disabled at time $t_c$, i.e.,
\begin{equation}
    \mc J^{*\backslash c}_{v_c}(t_c) :=\mc J^*_{v_c}(t_c)\backslash \mc J_{c,v_c}(t_c) \, .
\end{equation}

We let $N_R(v,t_c,v_c)$ be the maximal number of flights that are possible to land at node $v$ at the same time once node $v_c$ is disabled at time $t_c$, which can be computed as 
\begin{equation}
\label{eq:N_R}
    N_R(v,t_c,v_c) = \sup_{t\geq t_c} \sum_{j\in \mc J_{p}(v,t_c,v_c)}\mathbf{1}(t;\mc M^j_v ) \, .
\end{equation}
All of the above components \eqref{eq:links_affected}--\eqref{eq:N_R} are easily computed from a given feasible schedule. In the rest of the paper, we sometimes drop the arguments in the parentheses, $t, t_c, v, v_c$, when they are clear from the context.

\subsection{Necessary and Sufficient Condition for Worst-Case Safe Schedules}
\begin{thm}
\label{thm:sufficient}
Consider a network $\mc N = (\mc G,C,\mc R,\ul{x}, \ol{x},w)$, where $\mc G =  (\mc V, \mc E)$ and given backup nodes assignment $\mc B_e$ for all $e\in \mc E$. Assume given a feasible schedule $\mc S =\{(R_j,\delta_j)\}_{j\in \mc J}$. 

The schedule $\mc S$ is worst-case time-node conditionally safe for node $v_c$ and time $t_c$ if and only if there exists an integer set $\{N_{e,v}(t_c,v_c)\}_{e \in \mc E,v \in \mc V}$ that satisfies the following constraints for all $v \in \mc V$ and $e \in \mc E$:
\begin{align}
\label{eq:sufficient_v} C_v- \sum_{e\in \mc B_{v,v_c}} N_{e,v}(t_c,v_c) &\geq  N_R(v,t_c,v_c) \,, \quad \forall v\in \mc V \,, \\
\nonumber \sum_{v \in \mc{B}_{e,v_c}} N_{e,v}(t_c,v_c) &=\\
\nonumber \sum_{j \in \mc J^{*\backslash c}_{v_c}(t_c)}\mathbf{1}\Big(&t_c;[L^j_{e}, U^j_{e}]  \backslash \{\cup_{\ell \in \mc E^j_{v_c}}[L^j_{\ell^{R_j}}, U^j_{\ell^{R_j}}]\}\Big) \,, \\
\label{eq:sufficient_e} &\qquad \qquad \forall e\in \mc E_{v_c} \,, \\
\nonumber N_{e,\sigma(e)}(t_c,v_c) & = \sum_{j\in \mc J^{*\backslash c}_{v_c}(t_c)} \mathbf{1}(t_c;[L^j_{e}, U^j_{e}]) \,, \\  
\label{eq:N_e_fix} &\qquad \qquad \forall e\notin \mc E_{v_c} \,,  \\
\label{eq:N_e_0} N_{e,v}(t_c,v_c) &= 0 \,, \quad \forall e\in \mc E_{v_c}, v \notin \mc B_{e,v_c} \,, \\
\label{eq:N_e_var}  N_{e,v}(t_c,v_c) &\geq 0 \,, \quad \forall e\in \mc E_{v_c}, v \in \mc B_{e,v_c} \,,
\end{align}
where for all $j\in \mc J$, the lower and upper bounds of the time interval are defined as
\begin{align}
  \label{eq:lowerI}
  L^j_{e} = \begin{cases}
                  \inf\{\mc M^j_{\tau(e)}\}+w &\text{if }\tau(e) \neq 0^{R_j}\\
                  \delta_j & \text{if } \tau(e) = 0^{R_j} \, ,
\end{cases}
\end{align}
and 
\begin{align}
  \label{eq:upperI}
  U^j_{e} = \begin{cases}
                  \sup\{\mc M^j_{\sigma(e)}\} &\text{if }\sigma(e) \neq v_c\\
                  \sup\{\mc M^j_{\sigma(e)}\}-w & \text{if } \sigma(e) = v_c \, .
\end{cases}
\end{align}

Further, $\mc S$ is worst-case node-conditionally safe for node~$v_c$ if and only if such a set $\{N_{e,v}(t_c,v_c)\}_{e \in \mc E,v \in \mc V}$ satisfying \eqref{eq:sufficient_v}--\eqref{eq:N_e_var} exists for the finite number of times $t_c$ where the values of $N_R(v,t_c,v_c)$ and the time-varying index sets $\mc J^*_{v_c}(t_c)$, $\mc J_{c,v_c}(t_c)$ possibly change, i.e., at both endpoints of the interval $\mc M_v^j$ for all $v \in \mc V$, at times $\delta_j$ for all $j \in \mc J$, and at times $L^j_{e}$, $U^j_{e}$ for all $j \in \mc J$ and $e\in \mc E$.
\end{thm}
The second part of Theorem \ref{thm:sufficient} states that, while the definition for a schedule to be worst-case node-conditionally safe requires checking safety for all times $t_c\geq 0$, such conditions in fact only need to be checked at a finite number of times.
\begin{proof}
The schedule is worst-case time-node conditionally safe for node $v_c$ and time $t_c$ if and only if, for any possibly affected flight that is not canceled and may be rerouted to some node in $\mc V$ at time $t_c$, an available landing spot needs to be reserved. Hence the problem becomes to ensure the flights surely not affected will have no capacity conflict with any possibly rerouted flights. We then consider the maximum (worst-case) occupation of the node in $\mc V$.

We let the set $\{N_{e,v}(t_c,v_c)\}_{e \in \mc E, v \in \mc V}$ be the set of variables that denote the number of possibly affected flights that may proceed to node $v \in \mc V$ when traveling on the link $e \in \mc E$. Hence, for all $v \in \mc V, e \in \mc E$, $N_{e,v}(t_c,v_c)$ is required to be a non-negative integer. The interval defined as $[L^j_e,U^j_e]$ is the time interval during which flight $j$ will possibly be rerouted to $b_{e,v_c}$ if $v_c$ is closed, where the lower bound $L^j_e$ is the earliest time that the flight may leave the previous node $\tau(e)$, and, if $\sigma(e)$ is not disabled, the upper bound $U^j_e$ is the latest time that the flight may leave the head node while, in the case that $\sigma(e)$ is disabled, the upper bound $U^j_e$ for the time interval that the flight may be rerouted to the backup node $\tau(e)$ will be the latest time that the corresponding flight may arrive at node $v_c$, since otherwise it will continue its normal operation without rerouting. For any $e\in \mc E$, if $\sigma(e) \neq v_c$, then the possibly affected flights traveling on the link will land at its head node $\sigma(e)$, and thus the number of possibly affected flights rerouting to node $\sigma(e)$ from link $e$, $N_{e,\sigma(e)}(t_c,v_c)$ is deterministic, which can be simply counted as in~\eqref{eq:N_e_fix}. The constraint~\eqref{eq:N_e_0} prevents flights from proceeding to any node $v$ not in the set of possible backup nodes for link $e$ when node $v_c$ is disabled, $\mc B_{e,v_c}$.

As a safety requirement, when $v_c$ is disabled at time $t_c$, any possibly affected flight needs to be rerouted to a node. Consider a fixed $e \in \mc E_{v_c}$, a flight whose possibly traveling on this link at time $t_c$ is obviously a possibly affected flight when node $v_c$ is disabled at time $t_c$ and needs to be rerouted to one of its backup nodes. Therefore, \eqref{eq:sufficient_e} is the link safety constraint depicting that all flights possibly traveling on $e$ at $t_c$ needs to be rerouted to one of the possible backup nodes for link $e$ when $v_c$ is disabled. Notice that, supposing the backup nodes of the link $e$ include a node $v'\leq v_c$ that is along the route of the flight, and the flight is also possibly traveling on a link whose head node is $v'$ at $t_c$, then this means a landing spot at node $v'$ has to be reserved, and we do not need to prepare another one. This situation is reflected through $\{\cup_{\ell \in \mc E^j_{v_c}}[L^j_{\ell^{R_j}}, U^j_{\ell^{R_j}}]\}$ in~ \eqref{eq:sufficient_e}. Finally, $N_R(v,t_c,v_c)$ is the maximum number of flights not possibly affected that may park at node $v$ at any time once $v_c$ is disabled at $t_c$, and the summation $\sum_{e\in \mc B_{v,v_c}} N_{e,v}(t_c,v_c)$ is the total number of possibly affected flights rerouting to node $v$. Therefore~\eqref{eq:sufficient_v} is a necessary and sufficient condition to avoid the capacity conflict between the rerouted flights and those surely not affected for all realization of link travel times.
\end{proof}

Theorem~\ref{thm:sufficient} provides a finite number of conditions to verify a schedule is worst-case node conditionally safe for node~$v_c$.
Furthermore, by checking that a schedule is worst-case node conditionally safe for all $v_c \in \mc V$, we can conclude the 1-closure safety. However, we notice that looking for the existence of an integer set satisfying the constraints \eqref{eq:sufficient_v}--\eqref{eq:N_e_var} in Theorem~\ref{thm:sufficient}leads to a Mixed Integer Linear Programming (MILP) Problem, which is sensitive to scale and can be time-consuming once the size of the schedule under verification grows. We  recast the MILP as a linear program in Section~\ref{sec:simplification}, leading to an efficient safety-verification algorithm. In the following subsection, we explore the safety constraints for a given UAM schedule in the best-case scenario.

\subsection{Necessary and Sufficient Condition for Best-Case Safe Schedules}
\label{sec:best-case}
Theorem~\ref{thm:sufficient} provides a set of constraints that serve as a necessary and sufficient condition for a feasible schedule to be worst-case time-node conditionally, node conditionally or 1-closure safe. 
In this subsection, we provide constraints for a feasible schedule to be best-case safe. In the best-case scenario, 
we  consider the realization with the least number of rerouting flights and most flexible rerouting plan needed among all possible realizations. Therefore, we assume that all flights possible to have arrived at or passed through the closed node $v_c$ have already arrived or left by the time of node failure. 

We denote the index set of the \emph{definitely affected} flights as
\begin{equation}
    \mc J^m_{v_c}(t_c) = \{ j\in \mc J^*_{v_c}\backslash \mc J_{c,v_c}(t_c) \mid \inf\{\mc M^j_{v_c}\} \geq t_c \} \,.
\end{equation}
The definitely affected flights are the flights that must be rerouted under any possible realization.

\begin{thm}
\label{thm:best-case}
Consider a network $\mc N = (\mc G,C,\mc R,\ul{x}, \ol{x},w)$, where $\mc G =  (\mc V, \mc E)$ and given backup nodes assignment $\mc B_e$ for all $e\in \mc E$. A given feasible schedule $\mc S =\{(R_j,\delta_j)\}_{j\in \mc J}$ is best-case time-node conditionally safe for node $v_c$ and time $t_c$ if and only if there exists a non-negative integer set $\{N_{j,v}(t_c,v_c)\}_{j \in  \mc J^m_{v_c}(t_c), v \in \mc V}$ that satisfies the following constraints for all $j \in  \mc J^m_{v_c}(t_c)$ and $v \in \mc V$:
\begin{align}
\label{eq:best_v}    C_v - \sum_{j\in \mc J^m_{v_c}(t_c)}N_{j,v}(t_c,v_c) &\geq N_R(v,t_c,v_c) \,, \quad \forall v\in \mc V \,,\\
\label{eq:best_sum}    \sum_{v\in\mc V}N_{j,v}(t_c,v_c) &= 1 \,, \quad \forall j \in \mc J^m_{v_c}(t_c) \,, \\
\nonumber    0 \leq N_{j,v}(t_c,v_c) &\leq \max_{e\in \mc R_j} \mathbf{1}(t_c;[L_e^j,\hat{U}_e^j])\cdot \mathbf{1}(v;\mc B_{e,v_c}) \,, \\ 
\label{eq:best_binary}    & \qquad \forall v \in \mc V \,,  j \in \mc J^m_{v_c}(t_c)\,,
\end{align}
where $L^j_e$ is defined in \eqref{eq:lowerI} and
\begin{align}
  \label{eq:upperI_best}
  \hat{U}^j_{e} = \begin{cases}
                  \sup\{\mc M^j_{\sigma(e)}\} &\text{if }\sigma(e) \neq v_c \,,\\
                  \inf\{\mc M^j_{\sigma(e)}\} & \text{if } \sigma(e) = v_c \, .
\end{cases}
\end{align}

Further, $\mc S$ is best node-conditionally safe for node $v_c$ if and only if the set $\{N_{j,v}(t_c,v_c)\}_{j \in \mc J,v \in \mc V}$ that satisfies \eqref{eq:best_v}--\eqref{eq:best_binary} exists and the conditions holds for the finite number of times $t_c$ where the values of $N_R(v,t_c,v_c)$ and the time-varying index sets $\mc J^m_{v_c}(t_c)$, $\mc J_{c,v_c}(t_c)$ possibly change, i.e., at both endpoints of the interval $\mc M_v^j$ for all $v \in \mc V$, $\delta_j$ for all $j \in \mc J$ and at times $L^j_{e}$, $\hat{U}^j_{e}$ for all $j \in \mc J$ and $e\in \mc E$.
\end{thm}

\begin{proof}
The proof of Theorem~\ref{thm:best-case} applies the similar logic as in Theorem \ref{thm:sufficient} to the best-case scenario, while from the perspective of flights instead of the links. First of all, we can focus only on the definitely affected flights, since any flight that is possibly affected but not definitely affected is either canceled or has at least a realization of travel time such that the flight has already passed through or landed at node $v_c$ and does not need to be rerouted. 

We regard $N_{j,v}(t_c, v_c)$ as the indicator of $j$'th flight to be rerouted to node $v$ if node $v_c$ is disabled at time $t_c$, for all $j\in \mc J$ and $v\in \mc V$. As a result, the non-negative variable $N_{j,v}(t_c, v_c)$ is actually binary. We enforce this binary condition in \eqref{eq:best_binary}, where $N_{j,v}(t_c, v_c) \leq 1$ if there exists $e \in R_j$ such that $v$ is one of its possibly backup nodes when $v_c$ is closed and the flight is definitely affected and possibly traveling on the link $e$ at time $t_c$ and $N_{j,v}(t_c, v_c) = 0$ otherwise. Notice that the upper-bound of the time interval for the flight to travel through link $e$ and its head node and be rerouted to one of its possible backup nodes, $\hat{U}^j_e$, is adjusted comparing to $U^j_e$ defined in \eqref{eq:upperI} to include only the definitely affected flights. 
Once node $v_c$ is disabled at time $t_c$, a flight will actually be reroute to exactly one node, which is depicted in \eqref{eq:best_sum}. Moreover, \eqref{eq:best_v} is the capacity constraint, where $\sum_{j\in \mc J^m_{v_c}}N_{j,v}(t_c,v_c)$ is the number of definitely affected flights rerouted to node $v$.

If we are not able to find a set of non-negative integers $\{N_v\}_{v\in \mc V}$ that satisfies \eqref{eq:best_v}--\eqref{eq:best_binary}, then there must exist a conflict of occupation at one or more nodes once $v_c$ is closed at time $t_c$, and hence \eqref{eq:best_v}--\eqref{eq:best_binary} are sufficient and necessary conditions for the set of schedules to be best-case time-node conditionally safe for node $v_c$ and time $t_c$. 
\end{proof}

Theorem~\ref{thm:sufficient} is both sufficient and necessary for worst-case 1-closure safety, while Theorem~\ref{thm:best-case} is sufficient and necessary for best-case 1-closure safety. Since worst-case safety implies best-case safety, satisfaction of the conditions in Theorem~\ref{thm:sufficient} implies satisfaction of the conditions in Theorem~\ref{thm:best-case}.

\section{Simplification for Verification}
\label{sec:simplification}
The necessary and sufficient conditions for safety derived in Section \ref{sec:thms} involve integer constraints and therefore are inefficient for use in a direct numerical implementation. In this section, we show that these conditions can in fact be translated to efficient linear programming (LP) constraints.  
We first establish a lemma explaining the mathematical foundation for our simplification of Theorem~\ref{thm:sufficient} and~\ref{thm:best-case}, followed by a theorem that turns the MILP problem in Theorem~\ref{thm:sufficient} and~\ref{thm:best-case} into an LP problem. In particular, the following lemma shows that, for a special set of constraints on a set of variables, the existence of a solution over the real numbers induce the existence of a solution over the integers.

\begin{lemma}
\label{lem:simplify}
Given a set $L = \{(l_1,l_2) \in \mathbb{N}_{>0}^2\mid l_1\leq N_1, l_2 \leq N_2 \}$ for some positive integers $N_1, N_2$, and let $L_{var}$ be a subset of $L$. Let $\alpha_{l_1}$ (resp., $\beta_{l_2}$) be non-negative integers for all $l_1 = 1,\dots, N_1$ (resp., $l_2= 1, \dots,N_2$), and $\gamma_{l_1,l_2}$ be integers for all $(l_1,l_2) \in L_{var}$. If there exists a set of real numbers $\{n_{l_1,l_2}\}_{(l_1,l_2)\in L}$ that satisfies
\begin{align}
    \label{eq:simp1}    \sum_{l_2=1}^{N_2} n_{l_1,l_2} &= \alpha_{l_1} \,, & \forall& l_1 = 1,\dots, N_1 \,, \\
    \label{eq:simp2}    \sum_{l_1=1}^{N_1} n_{l_1,l_2} &\leq \beta_{l_2} \,,  &  \forall& l_2=1,\dots, N_2  \,, \\
    \label{eq:simp3}    n_{l_1,l_2} &= \gamma_{l_1,l_2} \,,  &   \forall& (l_1,l_2) \in L_{var} \,, \\
    \label{eq:simp4}    n_{l_1,l_2} &\geq 0  \,, & \forall& (l_1,l_2) \in L \,,
\end{align}
then there exists a set of integers $\{n'_{l_1,l_2}\}_{(l_1,l_2)\in L}$ also satisfying \eqref{eq:simp1}--\eqref{eq:simp4}.
\end{lemma}
The proof for Lemma~\ref{lem:simplify} can be found in Appendix~\ref{appendix:lemma}. Further, the remark below can be shown with some trivial revisions to the proof.

\begin{remark}
\label{remark:simplify}
Lemma~\ref{lem:simplify} holds if $n_{l_1,l_2}$ is bounded from above by an integer, i.e.,
\eqref{eq:simp4} is changed to $0 \leq n_{l_1,l_2} \leq U_{l_1,l_2}$ for all $(l_1,l_2) \in L$ for some integer number $U_{l_1,l_2}>0$.
\end{remark}

The simplified corollary below makes use of Lemma~\ref{lem:simplify} above and provides an LP alternative to the MILP problem in Theorem ~\ref{thm:sufficient}.

\begin{cor}
\label{cor:suff2}
Consider a network $\mc N = (\mc G,C,\mc R,\ul{x}, \ol{x},w)$, where $\mc G =  (\mc V, \mc E)$. Assume given a feasible schedule $\mc S = \{(R_j, \delta_j)\}_{j\in \mc J}$. There exists a set of real numbers $\{N_{e,v}(t_c,v_c)\}_{e \in \mc E,v \in \mc V}$ that satisfies the constraints \eqref{eq:sufficient_v}--\eqref{eq:N_e_var} if and only if there exists a set of integers $\{N'_{e,v}(t_c,v_c)\}_{e \in \mc E,v \in \mc V}$ that satisfies the constraints. 
\end{cor}

\begin{proof}
We first show that the conditions \eqref{eq:sufficient_v} and \eqref{eq:sufficient_e} conform to the form in Lemma \ref{lem:simplify}. 
For the sake of convenience, we fix $t_c$ and $v_c$ and drop the notation from $N_{e,v}(t_c,v_c)$ and $N_R(v,t_c,v_c)$, i.e., we write them as $N_{e,v}$ and $N_R(v)$ in this proof.

We can observe that, by the definition of $\mc B_{v,v_c}$ in \eqref{eq:backup2v}, assume $e \notin \mc B_{v,v_c}$, if $\sigma(e) = v_c$, then $v \notin \mc B_e\backslash v_c$, otherwise $\sigma(e) \neq v_c$. We can therefore conclude from \eqref{eq:N_e_fix} and \eqref{eq:N_e_0} that $N_{e,v}$ is a fixed number that can be computed if $e \notin \mc B_{v,v_c}$.

By adding the fixed terms of $N_{e,v}$ for $e \notin \mc B_{v,v_c}$ to both sides of \eqref{eq:sufficient_v} we can then obtain that, for all $v\in \mc V$, 
\begin{align}
\nonumber \sum_{e\in \mc B_{v,v_c}} N_{e,v} &\leq  C_v - N_R(v)\\
\label{eq:reform_1} \sum_{e\in \mc E} N_{e,v} &\leq  C_v -  N_R(v)+ \sum_{e\notin \mc B_{v,v_c}} N_{e,v}\,. 
\end{align}

Similarly, if $v\notin \mc B_{e, v_c}$, then $N_{e,v}$ is a fixed number. By adding $\sum_{v \notin \mc{B}_{e, v_c}} N_{e,v}$ to both sides of \eqref{eq:sufficient_e}, we have
\begin{align}
\nonumber   &\sum_{v \in \mc V} N_{e,v} = \sum_{v \notin \mc{B}_{e, v_c}} N_{e,v}+\\ 
\label{eq:reform_2}&\sum_{j \in \mc J^{*\backslash c}_{v_c}(t_c)}\mathbf{1}\Big(t_c;[L^j_{e}, U^j_{e}]  \backslash \{\cup_{\ell \in \mc E^j_{v_c}}[L^j_{\ell^{R_j}}, U^j_{\ell^{R_j}}]\}\Big) \, .
\end{align}

We then let $L = \{(e,v) \mid e\in \mc E, v \in \mc V \}$, and $L_{var} = L \backslash \{(e,v) \mid \sigma(e) = v_c, v\in \mc B_{e, v_c}\}$. Then we can combine and rewrite \eqref{eq:N_e_fix}--\eqref{eq:N_e_var} as 
\begin{align}
    \label{eq:reform_3} N_{e,v}  &= \gamma_{e,v} \,, \quad & \forall  (e,v) & \in L_{var} \,, \\
    \label{eq:reform_4} N_{e,v} &\geq 0 \,, \quad & \forall  (e,v) &\in L \, ,
\end{align}
where $\gamma_{e,v} = \sum_{j\in \mc J^*_{v_c}(t_c)\backslash \mc J_{c}(t_c)}\mathbf{1}(t_c;[L^j_{e}, U^j_{e}])$ if $\sigma(e) = v \neq v_c$ and $\gamma_{e,v} = 0$ if $\sigma(e) \neq v$ and $v \notin \mc B_{e, v_c}$.

As $\mc E$ and $\mc V$ are both finite sets, while the right sides of the inequalities \eqref{eq:reform_1} and \eqref{eq:reform_2} are integers, the conditions \eqref{eq:reform_1}--\eqref{eq:reform_4} exactly follows the conditions \eqref{eq:simp1}--\eqref{eq:simp4} in Lemma~\ref{lem:simplify}. Therefore, by Lemma~\ref{lem:simplify}, there exists a integral set $\{N_{e,v}\}_{e \in \mc E, v\in \mc V}$. 

As a result, given the schedule $\mc S$, if there exists a set of real numbers $\{N_{e,v}(t_c,v_c)\}_{e \in \mc E, v\in \mc V}$ that satisfies the constraints~\eqref{eq:sufficient_v}--\eqref{eq:N_e_var}, then there exist a set of integers $\{N'_{e,v}(t_c,v_c)\}_{e \in \mc E, v\in \mc V}$ that satisfies the constraints. 
The other direction of the corollary is immediate. 
\end{proof}

Combining Theorem~\ref{thm:sufficient} and Corollary~\ref{cor:suff2}, we are then able to verify 1-closure safety of given feasible schedules by solving an LP.
Similarly, we develop a corollary for simplification of best-case safety mirroring Corollary~\ref{cor:suff2} given Remark~\ref{remark:simplify}. 

\begin{cor}
\label{cor:best_case2}
Consider a network $\mc N = (\mc G,C,\mc R,\ul{x}, \ol{x},w)$, where $\mc G =  (\mc V, \mc E)$. Assume given a feasible schedule $\mc S = \{(R_j, \delta_j)\}_{j\in \mc J}$. There exists a set of real numbers $\{N_{j,v}(t_c,v_c)\}_{ j\in \mc J^m_{v_c},v \in \mc V}$ that satisfies the constraints \eqref{eq:best_v}--\eqref{eq:best_binary} if and only if there exists a set of non-negative integers $\{N'_{j,v}(t_c,v_c)\}_{ j\in \mc J^m_{v_c},v \in \mc V}$ that satisfies the constraints. 
\end{cor}

The proof for Corollary~\ref{cor:best_case2} is immediate from Lemma~\ref{lem:simplify} as we can easily convert the constraints \eqref{eq:best_v}--\eqref{eq:best_binary} to the same form as in Lemma~\ref{lem:simplify} and Remark~\ref{remark:simplify}.

\section{Case Study}
\label{sec:case_study}
In the case study, we demonstrate the verification algorithm based on Theorem~\ref{thm:sufficient} and Corollary~\ref{cor:suff2} on the UAM network in the  Example~\ref{ex:case} with 20 scheduled flights. 
We also demonstrate the efficient scaling of the algorithm on examples with up to 1,000 UAVs.

To ensure the worst-case node conditional safety when $v_c$ is closed, we check whether there exists a set of real numbers $\{N_{e,v}(t_c,v_c)\}_{e \in \mc E,v \in \mc V}$ that satisfies the constraints \eqref{eq:sufficient_v}--\eqref{eq:N_e_var} over the time interval $t_c \in [0,+\infty)$ so that the worst-case safety is guaranteed. As stated in Theorem \ref{thm:sufficient}, we only need to solve the LP feasibility problem at each point of time that any value may change, i.e., $L^j_e$, $U^j_e$, both ends of $\mc M^j_v$, and $\delta_j$ for any counted flight $j \in \mc J$ and link $e \in \mc E$ for some fixed node $v$, since the system of linear inequalities \eqref{eq:sufficient_v}--\eqref{eq:N_e_var} will not change between these points. 
We randomly generate a particular feasible schedule profile with $20$ flights and consider the constraints \eqref{eq:sufficient_v}--\eqref{eq:N_e_var} in Theorem~\ref{thm:sufficient} for time-node conditionally safe for node $v_c = v_5$ and any time $t_c >0$. The verification is implemented in MATLAB\footnote{The related MATLAB code can be found in \url{https://github.com/gtfactslab/Wei_TCNS_ScheduleVerification.git}.}.

In Fig.~\ref{fig:case_1}, we observe the 
worst-case landing-spot occupation at node $v_3$ (top) and the redistribution of flights on link $e_5$ and $e_6$ (middle and bottom) when node $v_5$ is disabled at any time $t_c$. For convenience, we simplify the notation $N_R(\cdot,t_c,v_5)$ and $N_{\cdot,\cdot}(t_c,v_5)$ as $N_R(\cdot)$ and $N_{\cdot,\cdot}$ in the figure.
The top graph of Fig.~\ref{fig:case_1} shows the distribution of UAVs that might possibly land at node $v_3$. The blue rectangles  correspond to flights that are not affected and continue to $v_3$ if node $v_5$ is disabled at time $t_c$, which is $N_R(v_3,t_c,v_5)$ in (\ref{eq:sufficient_v}); the pink rectangles correspond to flights rerouted to node $v_3$ from link $e_2 = (v_2,v_3)$ if $v_5$ is disabled at time $t_c$, which is $N_{e_2,v_3}(t_c,v_5)$; the orange (resp., green) rectangles correspond to the number of flights rerouted to node $v_3$ from link $e_5 = (v_3,v_5)$ (resp., $e_6 = (v_4,v_5)$) if node $v_5$ is disabled at time $t_c$, which is $N_{e_5,v_3}(t_c,v_5)$ (resp., $N_{e_6,v_3}(t_c,v_5)$). Notice that $N_{e_2,v_3}(t_c,v_5)$ is fixed and can be computed by \eqref{eq:N_e_fix}, since any flight traveling on $e_2$ at time $t_c$ has to land at $v_3$ if $v_5$ is disabled at that time; meanwhile, $N_{e_5,v_3}(t_c,v_5)$ (resp., $N_{e_6,v_3}(t_c,v_5)$) is an optimization variable computed through the LP problem~\eqref{eq:sufficient_v}--\eqref{eq:N_e_var}. However, it is possible that there does not exist a solution to the LP problem at certain time instances $t_c$, that is, the schedule is not worst-case time-node conditionally safe when node $v_5$ is disabled at time $t_c$. If that is the case, only the definite parts $N_R(v_3,t_c,v_5)$ and $N_{e_2,v_3}(t_c,v_5)$ (represented by blue and pink rectangles) are shown in the corresponding time interval, while the remaining parts are shown as a grey rectangle to demonstrate the failure. As a reference, the capacity $C_{v_3}=6$ is shown as the dotted, horizontal line so that the height of the entire bar (the sum of all rectangles) must not exceed the capacity for safety. 

The redistribution of flights on link $e_5$ (resp., $e_6$) shown in the middle (resp., bottom) graph of Fig.~\ref{fig:case_1} when node $v_5$ is disabled at any time $t_c$ provide a detailed partition of flights onto the nodes to which they are rerouted. Since the head of the link $e_5$ (resp., $e_6$), $v_5$, is disabled, the flights traveling on the link need to be rerouted to one of the possible backup nodes, $v_3$ or $v_4$ (resp., $v_3, v_4$, or $v_7$). We use orange and purple (resp., green, red, and blue) rectangles to represent $N_{e_5,v_3}(t_c,v_5)$ and $N_{e_5,v_4}(t_c,v_5)$ (resp., $N_{e_6,v_3}(t_c,v_5)$, $N_{e_6,v_4}(t_c,v_5)$, and $N_{e_6,v_7}(t_c,v_5)$), i.e., the number of affected flights traveling on link $e_5$ (resp., $e_6$) rerouted to the backup nodes $v_3$ and $v_5$ (resp., $v_3, v_4$, and $v_7$). Similar to the top graph of Fig.~\ref{fig:case_1}, we use grey rectangles to indicate the failure of obtaining the solution to the LP problem~\eqref{eq:sufficient_v}--\eqref{eq:N_e_var}. The height of the grey rectangles represents the total number of flights that need to be rerouted from link $e_5$ (resp., $e_6$) when node $v_5$ is disabled at time $t_c$, i.e., $\sum_{v \in \mc{B}_{e_5,v_5}} N_{e_5,v}(t_c,v_5)$. Notice that the solution to the LP problem ($N_{\cdot,\cdot}(t_c,v_5)$) for $t_c >0$, if it exists, is not unique, and hence Fig.~\ref{fig:case_1} is only one possible rerouting arrangement. Thus, as an example, the schedule in this case study is not time-node conditionally safe for node $v_5$ at $t_c = 40$, as the grey rectangle indicates there does not exist a solution to the problem~\eqref{eq:sufficient_v}--\eqref{eq:N_e_var} at time $t_c$. 
Therefore, the network is not able to accommodate the failure of $v_5$ at time $t_c = 40$. 

For the sake of comparison, we increase the capacity of $v_4$ to $C_{v_4} = 8$ while the other parts of the network remain the same. We then verify the safety of the same schedule with the algorithm, and these results are shown as in Fig.~\ref{fig:case_2}. As shown in the plots, after increasing the capacity of $v_4$, which is a backup node for both $e_5$ and $e_6$, the solution to the problem~\eqref{eq:sufficient_v}--\eqref{eq:N_e_var} exists all the time.
To conclude if the schedule is node conditionally safe when $v_5$ is disabled, we would need to observe all affected nodes and links in the network in the same way.

\begin{figure}[!t]
\centering
\input{plots/case_v3_v1.tikz}
\input{plots/case_e5_v1.tikz}
\input{plots/case_e6_v1.tikz}
\caption{Observation of the network when node $v_5$ is disabled at any time $t_c>0$. (Top) Expected landing-spot occupation at node $v_3$. (Middle) Redistribution of flights on link $e_5$ to its backup nodes. (Bottom) Redistribution of flights on link $e_6$ to its backup nodes. We simplify the notation $N_R(\cdot,t_c,v_5)$ and $N_{\cdot,\cdot}(t_c,v_5)$ as $N_R(\cdot)$ and $N_{\cdot,\cdot}$ in the figure.
}
\label{fig:case_1}
\end{figure}
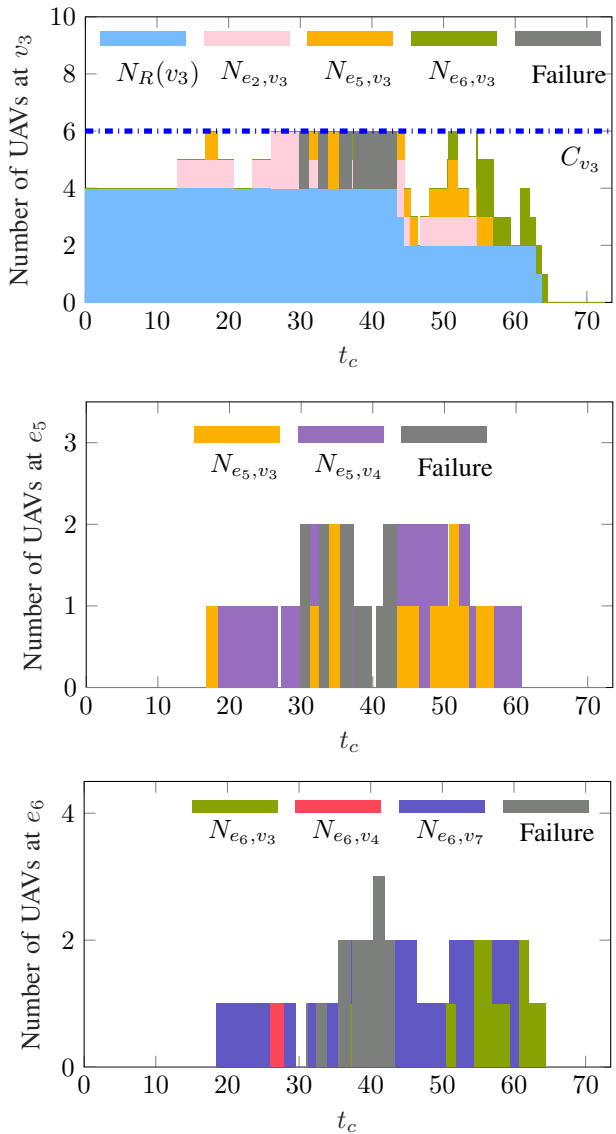

\begin{figure}[!t]
\centering
\input{plots/case_v3_v2.tikz}
\input{plots/case_e5_v2.tikz}
\input{plots/case_e6_v2.tikz}
\caption{Observation of the network when node $v_5$ is disabled at any time $t_c>0$ when we adjust the capacity of node $v_4$ to $C_{v_4} =8$. (Top) Expected landing-spot occupation at node $v_3$. (Middle) Redistribution of flights on link $e_5$ to its backup nodes. (Bottom) Redistribution of flights on link $e_6$ to its backup nodes. We simplify the notation $N_R(\cdot,t_c,v_5)$ and $N_{\cdot,\cdot}(t_c,v_5)$ as $N_R(\cdot)$ and $N_{\cdot,\cdot}$ in the figure. }
\label{fig:case_2}
\end{figure}
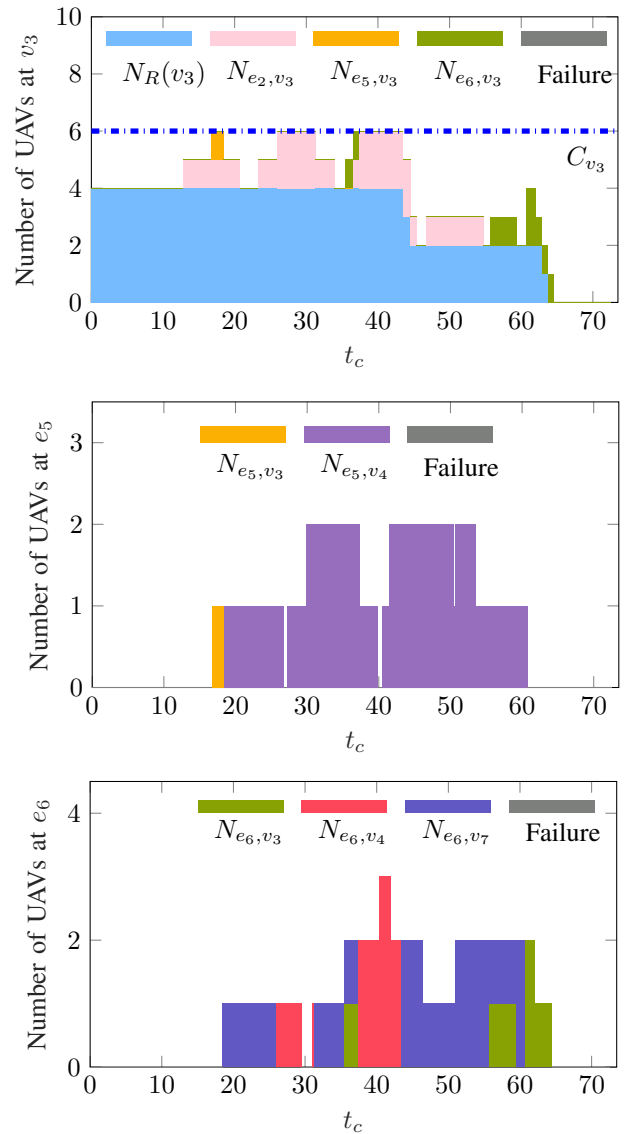

The computation time for $N_R(v,t_c,v_c)$ in (\ref{eq:sufficient_v}) increases quadratically with the size of the schedule, and as indicated in~\cite{megiddo1984linear}, solving the LP problem \eqref{eq:sufficient_v}--\eqref{eq:N_e_var} with a fixed number of variables can be computed within linear time with respect to the number of constraints, while the number of constraints in the LP problem and the number of times the LP needs to be solved  both grow linearly  with the size of schedule. We thus conclude that the verification process is completed in $O(n^2)$ time. This efficient scaling implies that we are  able to verify worst-case safety with large schedule profiles. As an example, consider increasing the capacity for each node of the network in Fig. \ref{fig:case_study_net} by $10$ to produce feasible schedules more easily. We generate $10$ more sets of random feasible schedules with sizes $100, 200, 300, \dots, 1000$ and verify their safety using the same algorithm. Fig.~\ref{fig:size_time} demonstrates the $O(n^2)$ computation complexity 
and shows that we are able to verify safety or demonstrate the safety failure for a schedule profile with 1,000 flights in under 50 seconds. As a baseline comparison, we also implement the verification algorithm with the naive MILP implied by Theorem \ref{thm:sufficient} without the efficient simplification to a LP derived in Section \ref{sec:simplification}. This implementation is solved using the Gurobi~\cite{gurobi} solver through with the YALMIP MATLAB toolbox~\cite{Lofberg2004}. We test the same 20-flight schedule on this MILP algorithm, which takes 7.34 seconds to verify, while the algorithm we use with simplification to LP takes only 1.43 seconds. A 100-flight schedule takes  around 40 seconds to verify with the naive MILP formulation, and 8 seconds with the LP algorithm.
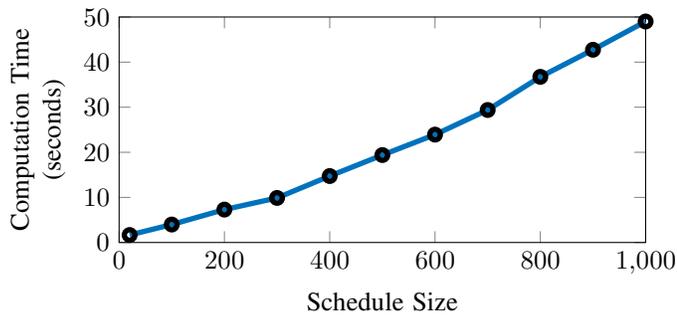
\begin{figure}[!t]
\centering
\input{plots/size_time.tikz}
\caption{The computation time for verifying the worst-case safety of schedules with different sizes. We test on $11$ different sets of schedules with sizes from $20$ to $1000$. The data points demonstrates the $O(n^2)$ computational complexity.
}
\label{fig:size_time}
\end{figure}

\section{Conclusion}
We studied the safety verification problem for  Urban Air Mobility (UAM) schedules in the face of vertiport (i.e., landing site) closures. We adopt a UAM network model that considers a set of finite-capacity vertiports and links between vertiports with uncertain travel time. If a vertiport is closed at some time, then flights destined for the closed vertiport must be rerouted to one of a set of link-dependent backup nodes. A safety violation occurs if the finite landing capacity at any node is exceeded due to the rerouting.

We consider the travel time uncertainty as a nondeterministic uncertainty, and therefore, we define appropriate notions of worst-case and best-case safety. We give necessary and sufficient conditions in both cases. If a given schedule satisfies the conditions for worst-case safety, then it is guaranteed that the schedule will not violate the safety constraints under any possibility of the travel times. On the other hand, if a schedule does not satisfy the conditions for best-case safety, then even if the uncertainty were favorably eliminated from the travel times via, e.g., aggressive low-level motion planning and control schemes, safety violation would still occur, implying the need for a new schedule. 

As formulated, these conditions take the form of mixed integer linear programming (MILP) constraints. We then showed that these numerically inefficient MILP constraints are able to be converted into efficient linear programming (LP) constraints using the theory of totally unimodular matrices (TUMs), resulting in an efficient algorithm for safety verification. We demonstrated our approach through several examples and case studies. 


In this paper, we considered the scenario where only a single node is disabled. An extension of this work could consider multiple simultaneous node closures. 
In particular, in the event that 
the disabled nodes are all strongly connected and are disabled concurrently, simple modifications of the methodology proposed in this paper would apply. However, a more general setting is more challenging to formulate and address and is a possible direction for future work. 
Our modeling approach could further allow other generalizations. For example, we regard a disrupted node as completely malfunctioning, but a partial malfunctioning disruption model, where not all landing spots of the disrupted node are disabled, could also be investigated.

\bibliography{Library} 
\bibliographystyle{IEEEtran}

\begin{appendix}
\label{appendix:lemma}

\section{Proof of  Lemma~\ref{lem:simplify}}
We prove in this appendix. 
The proof makes use of properties of \emph{Totally Unimodular Matrices (TUMs)}.  
\begin{definition}
\label{def:TUM}
(Totally Unimodular Matrix) A matrix is \emph{totally unimodular} if every square submatrix has determinant $0, +1$, or $-1$.
\end{definition}
TUMs are widely used in the context of optimization problems. In particular, it can be shown that for a large class of linear programs defined via TUMs, the resulting optimal solution takes on integer values~\cite{hoffman2010integral}. We use this property in the proof of Lemma~\ref{lem:simplify} next.

\begin{proof}[Proof of Lemma~\ref{lem:simplify}]
We first let the vector $\vec{n}$ be the vectorized sequence $\{n_{l_1,l_2}\}_{l_1=1,l_2=1}^{l_1=N_1,l_2=N_2}$, so that 
\begin{equation}
    \vec{n} = [n_{1,1}, n_{1,2},\dots, n_{1,N_2},n_{2,1},  \dots, n_{N_1,1}, \dots,n_{N_1,N_2}]^T \, .    
\end{equation}
We can then simplify the constraint \eqref{eq:simp1}--\eqref{eq:simp2} as 
\begin{align}
    \label{eq:mat1} A_1 \vec{n} &= \vec{\alpha} \,,\\
    \label{eq:mat2} A_2 \vec{n} &\leq \vec{\beta} \, ,
\end{align}
where $\vec{\alpha} = [\alpha_1,\dots,\alpha_{N_1}]^T$ and $\vec{\beta} = [\beta_1,\dots,\beta_{N_2}]^T$
, and $A_1$ (resp., $A_2$) is a $N_1\times N_1 N_2$ (resp., $N_2\times N_1 N_2$) matrix that reflects the matrix form of the multiplication of the constraints. In particular, the row-$i$-column-$j$ element of $A_1$  is $A_1(i,j) = 1$ if $(i-1)N_2 < j \leq i N_2$ and $A_1(i,j) = 0$ otherwise, and $A_2(i,j)=1$ if $j= m \cdot N_2 +i$ for $m=0,1, \dots, N_1-1$ and  $A_2(i,j) = 0$ otherwise. 

Since for $(l_1,l_2) \in L_{var}$, $n_{l_1, l_2} = \gamma_{l_1,l_2}$, we can then subtract the corresponding entries from the left sides of~\eqref{eq:mat1} and~\eqref{eq:mat2}, and subtract the values from their right sides. 
We let $A_3$ (resp., $A_4$) be the resulting matrices, so that $A_3$ (resp., $A_4$) is a $N_2\times N_1 N_2$ (resp., $N_1\times N_1 N_2$) and the row-$i$-column-$j$ element of $A_3$  is $A_3(i,j) = 0$ if $(i,j-(i-1)N_2) \in L_{var}$ and $A_3(i,j) = A_1(i,j)$  otherwise, and $A_4(i,j)=0$ if $(i,(j-i)/N_2+1) \in L_{var} $ and $A_4(i,j) = A_2(i,j)$ otherwise.  Let
\begin{align}
    \alpha'_{l_1}& = \alpha_{l_1} - \sum_{l_2: (l_1,l_2)\in L_{var}} \gamma_{l_1,l_2}\quad \text{for all $l_1$},\\
  \beta'_{l_2} &= \beta_{l_2} - \sum_{l_1: (l_1,l_2) \in L_{var}} \gamma_{l_1,l_2}\quad \text{for all $l_2$},\\
  \vec{\alpha'} &= [\alpha'_1,\dots,\alpha'_{N_1}]^T,\\
  \vec{\beta'} &= [\beta'_1,\dots,\beta'_{N_2}]^T.
\end{align}
We can then reformulate~\eqref{eq:mat1} and~\eqref{eq:mat2} together with the constraints~\eqref{eq:simp3}--\eqref{eq:simp4} as 
\def\stackbelow#1#2{\underset{\displaystyle\overset{\displaystyle\shortparallel}{#2}}{#1}}

\begin{equation}
\label{eq:mat_total}
    \underbrace{\begin{bmatrix}
    &A_3\\ &-A_3\\ &A_4\\ &-I_{N_1 N_2}
    \end{bmatrix}}_{=:A} \vec{n} \leq 
    \underbrace{\begin{bmatrix}
    &\vec{\alpha'} \\ &-\vec{\alpha'}\\ &\vec{\beta'}\\ &\vec{0}_{N_1 N_2}
    \end{bmatrix}}_{=:\vec{b}} \,, 
\end{equation}
where $I_{N_1 N_2}$ is the identity matrix with $N_1 N_2$ rows and $\vec{0}_{N_1 N_2}$ is the zero vector of length $N_1 N_2$. 

The first part of the lemma is then turned into the standard linear programming problem, which is finding the existence of $\vec{n}$ that satisfies $A\vec{n} \leq \vec{b}$. The next step is to prove that $A$ is a totally unimodular matrix (TUM) as defined in Definition~\ref{def:TUM}.

We first consider the matrix $\begin{bmatrix} A_3 \\ A_4 \end{bmatrix}$. Notice that for each column of $A_3$ and $A_4$, there exists at most one nonzero entry, $1$, therefore, for each column of the matrix $\begin{bmatrix} A_3 \\ A_4 \end{bmatrix}$, there exists at most two nonzero entries, and for any column with two non-zero entries, both of them will be $1$, and the row of one is in $A_3$ while the other in $A_4$. According to Hoffman's sufficient conditions~\cite[Appendix]{heller1956extension}, $\begin{bmatrix} A_3 \\ A_4 \end{bmatrix}$ is a TUM. By the general rule of TUM, $-\begin{bmatrix} A_3 \\ A_4 \end{bmatrix}$ is a TUM and thus $\begin{bmatrix} A_3 \\ A_4\\-A_3\\-A_4 \end{bmatrix}$ is also a TUM. According to the definition of TUM, it is obvious that deleting some rows from a TUM will produce a TUM, as any square non-singular submatrix of the new matrix will still be unimodular. As a result,  $\begin{bmatrix} A_3 \\ A_4\\-A_3 \end{bmatrix}$ and  $-\begin{bmatrix} A_3 \\ A_4\\-A_3 \end{bmatrix}$ are TUM. By the general rule of TUM, $\begin{bmatrix} -A_3 \\ -A_4\\ A_3 \\ I_{N_1 N_2}\end{bmatrix}$ is TUM and $\begin{bmatrix} A_3 \\ A_4\\ -A_3 \\ -I_{N_1 N_2}\end{bmatrix}$ is also a TUM. As switching rows does not affect the absolute value of the determinant of a matrix, then we conclude from above that $A$ is a TUM. 

Therefore,~\cite[Theorem 2]{hoffman2010integral} implies that if there exists a solution for the LP in~\eqref{eq:mat_total}, then there exists a integral solution for the same LP problem, which concludes the lemma.
\end{proof}
\end{appendix}

\end{document}

%% file: plots/case_v3_v1.tikz
%
%
\definecolor{mycolor1}{rgb}{0.45882,0.73333,0.99216}%
\definecolor{mycolor2}{rgb}{1.00000,0.81176,0.86275}%
\definecolor{mycolor3}{rgb}{0.98824,0.69020,0.00392}%
\definecolor{mycolor4}{rgb}{0.53725,0.63529,0.01176}%
\definecolor{mycolor5}{rgb}{0.49020,0.49804,0.48627}%
\begin{tikzpicture}

\begin{axis}[%
width=7cm,
height=3.8cm,
at={(0.534in,0.517in)},
scale only axis,
xmin=0,
xmax=73.5,
xlabel style={font=\color{white!15!black}},
xlabel={$t_c$},
ymin=0,
ymax=10,
ylabel style={font=\color{white!15!black}},
ylabel={Number of UAVs at $v_3$}
]
\draw[draw=none, only marks, fill=mycolor1] (axis cs:0,0) rectangle (axis cs:12.6,4);
\draw[draw=none, only marks, fill=mycolor2] (axis cs:0,4) rectangle (axis cs:12.6,4);
\draw[draw=none, only marks, fill=mycolor3] (axis cs:0,4) rectangle (axis cs:12.6,4);
\draw[draw=none, only marks, fill=mycolor4] (axis cs:0,4) rectangle (axis cs:12.6,4);
\draw[draw=none, only marks, fill=mycolor1] (axis cs:12.6,0) rectangle (axis cs:12.8,4);
\draw[draw=none, only marks, fill=mycolor2] (axis cs:12.6,4) rectangle (axis cs:12.8,4);
\draw[draw=none, only marks, fill=mycolor3] (axis cs:12.6,4) rectangle (axis cs:12.8,4);
\draw[draw=none, only marks, fill=mycolor4] (axis cs:12.6,4) rectangle (axis cs:12.8,4);
\draw[draw=none, only marks, fill=mycolor1] (axis cs:12.8,0) rectangle (axis cs:14.7,4);
\draw[draw=none, only marks, fill=mycolor2] (axis cs:12.8,4) rectangle (axis cs:14.7,5);
\draw[draw=none, only marks, fill=mycolor3] (axis cs:12.8,5) rectangle (axis cs:14.7,5);
\draw[draw=none, only marks, fill=mycolor4] (axis cs:12.8,5) rectangle (axis cs:14.7,5);
\draw[draw=none, only marks, fill=mycolor1] (axis cs:14.7,0) rectangle (axis cs:15.5,4);
\draw[draw=none, only marks, fill=mycolor2] (axis cs:14.7,4) rectangle (axis cs:15.5,5);
\draw[draw=none, only marks, fill=mycolor3] (axis cs:14.7,5) rectangle (axis cs:15.5,5);
\draw[draw=none, only marks, fill=mycolor4] (axis cs:14.7,5) rectangle (axis cs:15.5,5);
\draw[draw=none, only marks, fill=mycolor1] (axis cs:15.5,0) rectangle (axis cs:15.8,4);
\draw[draw=none, only marks, fill=mycolor2] (axis cs:15.5,4) rectangle (axis cs:15.8,5);
\draw[draw=none, only marks, fill=mycolor3] (axis cs:15.5,5) rectangle (axis cs:15.8,5);
\draw[draw=none, only marks, fill=mycolor4] (axis cs:15.5,5) rectangle (axis cs:15.8,5);
\draw[draw=none, only marks, fill=mycolor1] (axis cs:15.8,0) rectangle (axis cs:16.8,4);
\draw[draw=none, only marks, fill=mycolor2] (axis cs:15.8,4) rectangle (axis cs:16.8,5);
\draw[draw=none, only marks, fill=mycolor3] (axis cs:15.8,5) rectangle (axis cs:16.8,5);
\draw[draw=none, only marks, fill=mycolor4] (axis cs:15.8,5) rectangle (axis cs:16.8,5);
\draw[draw=none, only marks, fill=mycolor1] (axis cs:16.8,0) rectangle (axis cs:17.5,4);
\draw[draw=none, only marks, fill=mycolor2] (axis cs:16.8,4) rectangle (axis cs:17.5,5);
\draw[draw=none, only marks, fill=mycolor3] (axis cs:16.8,5) rectangle (axis cs:17.5,6);
\draw[draw=none, only marks, fill=mycolor4] (axis cs:16.8,6) rectangle (axis cs:17.5,6);
\draw[draw=none, only marks, fill=mycolor1] (axis cs:17.5,0) rectangle (axis cs:17.6,4);
\draw[draw=none, only marks, fill=mycolor2] (axis cs:17.5,4) rectangle (axis cs:17.6,5);
\draw[draw=none, only marks, fill=mycolor3] (axis cs:17.5,5) rectangle (axis cs:17.6,6);
\draw[draw=none, only marks, fill=mycolor4] (axis cs:17.5,6) rectangle (axis cs:17.6,6);
\draw[draw=none, only marks, fill=mycolor1] (axis cs:17.6,0) rectangle (axis cs:18.5,4);
\draw[draw=none, only marks, fill=mycolor2] (axis cs:17.6,4) rectangle (axis cs:18.5,5);
\draw[draw=none, only marks, fill=mycolor3] (axis cs:17.6,5) rectangle (axis cs:18.5,6);
\draw[draw=none, only marks, fill=mycolor4] (axis cs:17.6,6) rectangle (axis cs:18.5,6);
\draw[draw=none, only marks, fill=mycolor1] (axis cs:18.5,0) rectangle (axis cs:19.5,4);
\draw[draw=none, only marks, fill=mycolor2] (axis cs:18.5,4) rectangle (axis cs:19.5,5);
\draw[draw=none, only marks, fill=mycolor3] (axis cs:18.5,5) rectangle (axis cs:19.5,5);
\draw[draw=none, only marks, fill=mycolor4] (axis cs:18.5,5) rectangle (axis cs:19.5,5);
\draw[draw=none, only marks, fill=mycolor1] (axis cs:19.5,0) rectangle (axis cs:19.7,4);
\draw[draw=none, only marks, fill=mycolor2] (axis cs:19.5,4) rectangle (axis cs:19.7,5);
\draw[draw=none, only marks, fill=mycolor3] (axis cs:19.5,5) rectangle (axis cs:19.7,5);
\draw[draw=none, only marks, fill=mycolor4] (axis cs:19.5,5) rectangle (axis cs:19.7,5);
\draw[draw=none, only marks, fill=mycolor1] (axis cs:19.7,0) rectangle (axis cs:20.8,4);
\draw[draw=none, only marks, fill=mycolor2] (axis cs:19.7,4) rectangle (axis cs:20.8,5);
\draw[draw=none, only marks, fill=mycolor3] (axis cs:19.7,5) rectangle (axis cs:20.8,5);
\draw[draw=none, only marks, fill=mycolor4] (axis cs:19.7,5) rectangle (axis cs:20.8,5);
\draw[draw=none, only marks, fill=mycolor1] (axis cs:20.8,0) rectangle (axis cs:21.8,4);
\draw[draw=none, only marks, fill=mycolor2] (axis cs:20.8,4) rectangle (axis cs:21.8,4);
\draw[draw=none, only marks, fill=mycolor3] (axis cs:20.8,4) rectangle (axis cs:21.8,4);
\draw[draw=none, only marks, fill=mycolor4] (axis cs:20.8,4) rectangle (axis cs:21.8,4);
\draw[draw=none, only marks, fill=mycolor1] (axis cs:21.8,0) rectangle (axis cs:23.3,4);
\draw[draw=none, only marks, fill=mycolor2] (axis cs:21.8,4) rectangle (axis cs:23.3,4);
\draw[draw=none, only marks, fill=mycolor3] (axis cs:21.8,4) rectangle (axis cs:23.3,4);
\draw[draw=none, only marks, fill=mycolor4] (axis cs:21.8,4) rectangle (axis cs:23.3,4);
\draw[draw=none, only marks, fill=mycolor1] (axis cs:23.3,0) rectangle (axis cs:23.5,4);
\draw[draw=none, only marks, fill=mycolor2] (axis cs:23.3,4) rectangle (axis cs:23.5,5);
\draw[draw=none, only marks, fill=mycolor3] (axis cs:23.3,5) rectangle (axis cs:23.5,5);
\draw[draw=none, only marks, fill=mycolor4] (axis cs:23.3,5) rectangle (axis cs:23.5,5);
\draw[draw=none, only marks, fill=mycolor1] (axis cs:23.5,0) rectangle (axis cs:24.5,4);
\draw[draw=none, only marks, fill=mycolor2] (axis cs:23.5,4) rectangle (axis cs:24.5,5);
\draw[draw=none, only marks, fill=mycolor3] (axis cs:23.5,5) rectangle (axis cs:24.5,5);
\draw[draw=none, only marks, fill=mycolor4] (axis cs:23.5,5) rectangle (axis cs:24.5,5);
\draw[draw=none, only marks, fill=mycolor1] (axis cs:24.5,0) rectangle (axis cs:25.8,4);
\draw[draw=none, only marks, fill=mycolor2] (axis cs:24.5,4) rectangle (axis cs:25.8,5);
\draw[draw=none, only marks, fill=mycolor3] (axis cs:24.5,5) rectangle (axis cs:25.8,5);
\draw[draw=none, only marks, fill=mycolor4] (axis cs:24.5,5) rectangle (axis cs:25.8,5);
\draw[draw=none, only marks, fill=mycolor1] (axis cs:25.8,0) rectangle (axis cs:25.9,4);
\draw[draw=none, only marks, fill=mycolor2] (axis cs:25.8,4) rectangle (axis cs:25.9,5);
\draw[draw=none, only marks, fill=mycolor3] (axis cs:25.8,5) rectangle (axis cs:25.9,5);
\draw[draw=none, only marks, fill=mycolor4] (axis cs:25.8,5) rectangle (axis cs:25.9,5);
\draw[draw=none, only marks, fill=mycolor1] (axis cs:25.9,0) rectangle (axis cs:26.3,4);
\draw[draw=none, only marks, fill=mycolor2] (axis cs:25.9,4) rectangle (axis cs:26.3,6);
\draw[draw=none, only marks, fill=mycolor3] (axis cs:25.9,6) rectangle (axis cs:26.3,6);
\draw[draw=none, only marks, fill=mycolor4] (axis cs:25.9,6) rectangle (axis cs:26.3,6);
\draw[draw=none, only marks, fill=mycolor1] (axis cs:26.3,0) rectangle (axis cs:26.8,4);
\draw[draw=none, only marks, fill=mycolor2] (axis cs:26.3,4) rectangle (axis cs:26.8,6);
\draw[draw=none, only marks, fill=mycolor3] (axis cs:26.3,6) rectangle (axis cs:26.8,6);
\draw[draw=none, only marks, fill=mycolor4] (axis cs:26.3,6) rectangle (axis cs:26.8,6);
\draw[draw=none, only marks, fill=mycolor1] (axis cs:26.8,0) rectangle (axis cs:27.3,4);
\draw[draw=none, only marks, fill=mycolor2] (axis cs:26.8,4) rectangle (axis cs:27.3,6);
\draw[draw=none, only marks, fill=mycolor3] (axis cs:26.8,6) rectangle (axis cs:27.3,6);
\draw[draw=none, only marks, fill=mycolor4] (axis cs:26.8,6) rectangle (axis cs:27.3,6);
\draw[draw=none, only marks, fill=mycolor1] (axis cs:27.3,0) rectangle (axis cs:28,4);
\draw[draw=none, only marks, fill=mycolor2] (axis cs:27.3,4) rectangle (axis cs:28,6);
\draw[draw=none, only marks, fill=mycolor3] (axis cs:27.3,6) rectangle (axis cs:28,6);
\draw[draw=none, only marks, fill=mycolor4] (axis cs:27.3,6) rectangle (axis cs:28,6);
\draw[draw=none, only marks, fill=mycolor1] (axis cs:28,0) rectangle (axis cs:28.9,4);
\draw[draw=none, only marks, fill=mycolor2] (axis cs:28,4) rectangle (axis cs:28.9,6);
\draw[draw=none, only marks, fill=mycolor3] (axis cs:28,6) rectangle (axis cs:28.9,6);
\draw[draw=none, only marks, fill=mycolor4] (axis cs:28,6) rectangle (axis cs:28.9,6);
\draw[draw=none, only marks, fill=mycolor1] (axis cs:28.9,0) rectangle (axis cs:29.5,4);
\draw[draw=none, only marks, fill=mycolor2] (axis cs:28.9,4) rectangle (axis cs:29.5,6);
\draw[draw=none, only marks, fill=mycolor3] (axis cs:28.9,6) rectangle (axis cs:29.5,6);
\draw[draw=none, only marks, fill=mycolor4] (axis cs:28.9,6) rectangle (axis cs:29.5,6);
\draw[draw=none, only marks, fill=mycolor1] (axis cs:29.5,0) rectangle (axis cs:29.6,4);
\draw[draw=none, only marks, fill=mycolor2] (axis cs:29.5,4) rectangle (axis cs:29.6,6);
\draw[draw=none, only marks, fill=mycolor3] (axis cs:29.5,6) rectangle (axis cs:29.6,6);
\draw[draw=none, only marks, fill=mycolor4] (axis cs:29.5,6) rectangle (axis cs:29.6,6);
\draw[draw=none, only marks, fill=mycolor1] (axis cs:29.6,0) rectangle (axis cs:29.9,4);
\draw[draw=none, only marks, fill=mycolor2] (axis cs:29.6,4) rectangle (axis cs:29.9,6);
\draw[draw=none, only marks, fill=mycolor3] (axis cs:29.6,6) rectangle (axis cs:29.9,6);
\draw[draw=none, only marks, fill=mycolor4] (axis cs:29.6,6) rectangle (axis cs:29.9,6);
\draw[draw=none, only marks, fill=mycolor1] (axis cs:29.9,0) rectangle (axis cs:30,4);
\draw[draw=none, only marks, fill=mycolor5] (axis cs:29.9,4) rectangle (axis cs:30,6);
\draw[draw=none, only marks, fill=mycolor1] (axis cs:30,0) rectangle (axis cs:31,4);
\draw[draw=none, only marks, fill=mycolor5] (axis cs:30,4) rectangle (axis cs:31,6);
\draw[draw=none, only marks, fill=mycolor1] (axis cs:31,0) rectangle (axis cs:31.3,4);
\draw[draw=none, only marks, fill=mycolor5] (axis cs:31,4) rectangle (axis cs:31.3,6);
\draw[draw=none, only marks, fill=mycolor1] (axis cs:31.3,0) rectangle (axis cs:32,4);
\draw[draw=none, only marks, fill=mycolor2] (axis cs:31.3,4) rectangle (axis cs:32,5);
\draw[draw=none, only marks, fill=mycolor3] (axis cs:31.3,5) rectangle (axis cs:32,6);
\draw[draw=none, only marks, fill=mycolor4] (axis cs:31.3,6) rectangle (axis cs:32,6);
\draw[draw=none, only marks, fill=mycolor1] (axis cs:32,0) rectangle (axis cs:32.3,4);
\draw[draw=none, only marks, fill=mycolor2] (axis cs:32,4) rectangle (axis cs:32.3,5);
\draw[draw=none, only marks, fill=mycolor3] (axis cs:32,5) rectangle (axis cs:32.3,6);
\draw[draw=none, only marks, fill=mycolor4] (axis cs:32,6) rectangle (axis cs:32.3,6);
\draw[draw=none, only marks, fill=mycolor1] (axis cs:32.3,0) rectangle (axis cs:32.5,4);
\draw[draw=none, only marks, fill=mycolor2] (axis cs:32.3,4) rectangle (axis cs:32.5,5);
\draw[draw=none, only marks, fill=mycolor3] (axis cs:32.3,5) rectangle (axis cs:32.5,6);
\draw[draw=none, only marks, fill=mycolor4] (axis cs:32.3,6) rectangle (axis cs:32.5,6);
\draw[draw=none, only marks, fill=mycolor1] (axis cs:32.5,0) rectangle (axis cs:33.9,4);
\draw[draw=none, only marks, fill=mycolor5] (axis cs:32.5,4) rectangle (axis cs:33.9,6);
\draw[draw=none, only marks, fill=mycolor1] (axis cs:33.9,0) rectangle (axis cs:34.5,4);
\draw[draw=none, only marks, fill=mycolor2] (axis cs:33.9,4) rectangle (axis cs:34.5,4);
\draw[draw=none, only marks, fill=mycolor3] (axis cs:33.9,4) rectangle (axis cs:34.5,6);
\draw[draw=none, only marks, fill=mycolor4] (axis cs:33.9,6) rectangle (axis cs:34.5,6);
\draw[draw=none, only marks, fill=mycolor1] (axis cs:34.5,0) rectangle (axis cs:34.6,4);
\draw[draw=none, only marks, fill=mycolor2] (axis cs:34.5,4) rectangle (axis cs:34.6,4);
\draw[draw=none, only marks, fill=mycolor3] (axis cs:34.5,4) rectangle (axis cs:34.6,6);
\draw[draw=none, only marks, fill=mycolor4] (axis cs:34.5,6) rectangle (axis cs:34.6,6);
\draw[draw=none, only marks, fill=mycolor1] (axis cs:34.6,0) rectangle (axis cs:34.8,4);
\draw[draw=none, only marks, fill=mycolor2] (axis cs:34.6,4) rectangle (axis cs:34.8,4);
\draw[draw=none, only marks, fill=mycolor3] (axis cs:34.6,4) rectangle (axis cs:34.8,6);
\draw[draw=none, only marks, fill=mycolor4] (axis cs:34.6,6) rectangle (axis cs:34.8,6);
\draw[draw=none, only marks, fill=mycolor1] (axis cs:34.8,0) rectangle (axis cs:34.9,4);
\draw[draw=none, only marks, fill=mycolor2] (axis cs:34.8,4) rectangle (axis cs:34.9,4);
\draw[draw=none, only marks, fill=mycolor3] (axis cs:34.8,4) rectangle (axis cs:34.9,6);
\draw[draw=none, only marks, fill=mycolor4] (axis cs:34.8,6) rectangle (axis cs:34.9,6);
\draw[draw=none, only marks, fill=mycolor1] (axis cs:34.9,0) rectangle (axis cs:35.5,4);
\draw[draw=none, only marks, fill=mycolor2] (axis cs:34.9,4) rectangle (axis cs:35.5,4);
\draw[draw=none, only marks, fill=mycolor3] (axis cs:34.9,4) rectangle (axis cs:35.5,6);
\draw[draw=none, only marks, fill=mycolor4] (axis cs:34.9,6) rectangle (axis cs:35.5,6);
\draw[draw=none, only marks, fill=mycolor1] (axis cs:35.5,0) rectangle (axis cs:36,4);
\draw[draw=none, only marks, fill=mycolor5] (axis cs:35.5,4) rectangle (axis cs:36,6);
\draw[draw=none, only marks, fill=mycolor1] (axis cs:36,0) rectangle (axis cs:36.3,4);
\draw[draw=none, only marks, fill=mycolor5] (axis cs:36,4) rectangle (axis cs:36.3,6);
\draw[draw=none, only marks, fill=mycolor1] (axis cs:36.3,0) rectangle (axis cs:36.5,4);
\draw[draw=none, only marks, fill=mycolor5] (axis cs:36.3,4) rectangle (axis cs:36.5,6);
\draw[draw=none, only marks, fill=mycolor1] (axis cs:36.5,0) rectangle (axis cs:37,4);
\draw[draw=none, only marks, fill=mycolor5] (axis cs:36.5,4) rectangle (axis cs:37,6);
\draw[draw=none, only marks, fill=mycolor1] (axis cs:37,0) rectangle (axis cs:37.3,4);
\draw[draw=none, only marks, fill=mycolor5] (axis cs:37,4) rectangle (axis cs:37.3,6);
\draw[draw=none, only marks, fill=mycolor1] (axis cs:37.3,0) rectangle (axis cs:37.4,4);
\draw[draw=none, only marks, fill=mycolor2] (axis cs:37.3,4) rectangle (axis cs:37.4,5);
\draw[draw=none, only marks, fill=mycolor3] (axis cs:37.3,5) rectangle (axis cs:37.4,5);
\draw[draw=none, only marks, fill=mycolor4] (axis cs:37.3,5) rectangle (axis cs:37.4,6);
\draw[draw=none, only marks, fill=mycolor1] (axis cs:37.4,0) rectangle (axis cs:37.5,4);
\draw[draw=none, only marks, fill=mycolor5] (axis cs:37.4,4) rectangle (axis cs:37.5,6);
\draw[draw=none, only marks, fill=mycolor1] (axis cs:37.5,0) rectangle (axis cs:38.5,4);
\draw[draw=none, only marks, fill=mycolor5] (axis cs:37.5,4) rectangle (axis cs:38.5,6);
\draw[draw=none, only marks, fill=mycolor1] (axis cs:38.5,0) rectangle (axis cs:38.9,4);
\draw[draw=none, only marks, fill=mycolor5] (axis cs:38.5,4) rectangle (axis cs:38.9,6);
\draw[draw=none, only marks, fill=mycolor1] (axis cs:38.9,0) rectangle (axis cs:39.4,4);
\draw[draw=none, only marks, fill=mycolor5] (axis cs:38.9,4) rectangle (axis cs:39.4,6);
\draw[draw=none, only marks, fill=mycolor1] (axis cs:39.4,0) rectangle (axis cs:39.5,4);
\draw[draw=none, only marks, fill=mycolor5] (axis cs:39.4,4) rectangle (axis cs:39.5,6);
\draw[draw=none, only marks, fill=mycolor1] (axis cs:39.5,0) rectangle (axis cs:39.9,4);
\draw[draw=none, only marks, fill=mycolor5] (axis cs:39.5,4) rectangle (axis cs:39.9,6);
\draw[draw=none, only marks, fill=mycolor1] (axis cs:39.9,0) rectangle (axis cs:40.4,4);
\draw[draw=none, only marks, fill=mycolor5] (axis cs:39.9,4) rectangle (axis cs:40.4,6);
\draw[draw=none, only marks, fill=mycolor1] (axis cs:40.4,0) rectangle (axis cs:40.5,4);
\draw[draw=none, only marks, fill=mycolor5] (axis cs:40.4,4) rectangle (axis cs:40.5,6);
\draw[draw=none, only marks, fill=mycolor1] (axis cs:40.5,0) rectangle (axis cs:41.3,4);
\draw[draw=none, only marks, fill=mycolor5] (axis cs:40.5,4) rectangle (axis cs:41.3,6);
\draw[draw=none, only marks, fill=mycolor1] (axis cs:41.3,0) rectangle (axis cs:41.4,4);
\draw[draw=none, only marks, fill=mycolor5] (axis cs:41.3,4) rectangle (axis cs:41.4,6);
\draw[draw=none, only marks, fill=mycolor1] (axis cs:41.4,0) rectangle (axis cs:41.5,4);
\draw[draw=none, only marks, fill=mycolor5] (axis cs:41.4,4) rectangle (axis cs:41.5,6);
\draw[draw=none, only marks, fill=mycolor1] (axis cs:41.5,0) rectangle (axis cs:42,4);
\draw[draw=none, only marks, fill=mycolor5] (axis cs:41.5,4) rectangle (axis cs:42,6);
\draw[draw=none, only marks, fill=mycolor1] (axis cs:42,0) rectangle (axis cs:43.5,4);
\draw[draw=none, only marks, fill=mycolor5] (axis cs:42,4) rectangle (axis cs:43.5,6);
\draw[draw=none, only marks, fill=mycolor1] (axis cs:43.5,0) rectangle (axis cs:44.5,3);
\draw[draw=none, only marks, fill=mycolor2] (axis cs:43.5,3) rectangle (axis cs:44.5,5);
\draw[draw=none, only marks, fill=mycolor3] (axis cs:43.5,5) rectangle (axis cs:44.5,6);
\draw[draw=none, only marks, fill=mycolor4] (axis cs:43.5,6) rectangle (axis cs:44.5,6);
\draw[draw=none, only marks, fill=mycolor1] (axis cs:44.5,0) rectangle (axis cs:45.3,2);
\draw[draw=none, only marks, fill=mycolor2] (axis cs:44.5,2) rectangle (axis cs:45.3,3);
\draw[draw=none, only marks, fill=mycolor3] (axis cs:44.5,3) rectangle (axis cs:45.3,4);
\draw[draw=none, only marks, fill=mycolor4] (axis cs:44.5,4) rectangle (axis cs:45.3,4);
\draw[draw=none, only marks, fill=mycolor1] (axis cs:45.3,0) rectangle (axis cs:45.4,2);
\draw[draw=none, only marks, fill=mycolor2] (axis cs:45.3,2) rectangle (axis cs:45.4,3);
\draw[draw=none, only marks, fill=mycolor3] (axis cs:45.3,3) rectangle (axis cs:45.4,4);
\draw[draw=none, only marks, fill=mycolor4] (axis cs:45.3,4) rectangle (axis cs:45.4,4);
\draw[draw=none, only marks, fill=mycolor1] (axis cs:45.4,0) rectangle (axis cs:45.5,2);
\draw[draw=none, only marks, fill=mycolor2] (axis cs:45.4,2) rectangle (axis cs:45.5,2);
\draw[draw=none, only marks, fill=mycolor3] (axis cs:45.4,2) rectangle (axis cs:45.5,3);
\draw[draw=none, only marks, fill=mycolor4] (axis cs:45.4,3) rectangle (axis cs:45.5,3);
\draw[draw=none, only marks, fill=mycolor1] (axis cs:45.5,0) rectangle (axis cs:46.3,2);
\draw[draw=none, only marks, fill=mycolor2] (axis cs:45.5,2) rectangle (axis cs:46.3,2);
\draw[draw=none, only marks, fill=mycolor3] (axis cs:45.5,2) rectangle (axis cs:46.3,3);
\draw[draw=none, only marks, fill=mycolor4] (axis cs:45.5,3) rectangle (axis cs:46.3,3);
\draw[draw=none, only marks, fill=mycolor1] (axis cs:46.3,0) rectangle (axis cs:46.4,2);
\draw[draw=none, only marks, fill=mycolor2] (axis cs:46.3,2) rectangle (axis cs:46.4,2);
\draw[draw=none, only marks, fill=mycolor3] (axis cs:46.3,2) rectangle (axis cs:46.4,3);
\draw[draw=none, only marks, fill=mycolor4] (axis cs:46.3,3) rectangle (axis cs:46.4,3);
\draw[draw=none, only marks, fill=mycolor1] (axis cs:46.4,0) rectangle (axis cs:46.5,2);
\draw[draw=none, only marks, fill=mycolor2] (axis cs:46.4,2) rectangle (axis cs:46.5,2);
\draw[draw=none, only marks, fill=mycolor3] (axis cs:46.4,2) rectangle (axis cs:46.5,2);
\draw[draw=none, only marks, fill=mycolor4] (axis cs:46.4,2) rectangle (axis cs:46.5,2);
\draw[draw=none, only marks, fill=mycolor1] (axis cs:46.5,0) rectangle (axis cs:46.7,2);
\draw[draw=none, only marks, fill=mycolor2] (axis cs:46.5,2) rectangle (axis cs:46.7,2);
\draw[draw=none, only marks, fill=mycolor3] (axis cs:46.5,2) rectangle (axis cs:46.7,2);
\draw[draw=none, only marks, fill=mycolor4] (axis cs:46.5,2) rectangle (axis cs:46.7,2);
\draw[draw=none, only marks, fill=mycolor1] (axis cs:46.7,0) rectangle (axis cs:47.9,2);
\draw[draw=none, only marks, fill=mycolor2] (axis cs:46.7,2) rectangle (axis cs:47.9,3);
\draw[draw=none, only marks, fill=mycolor3] (axis cs:46.7,3) rectangle (axis cs:47.9,3);
\draw[draw=none, only marks, fill=mycolor4] (axis cs:46.7,3) rectangle (axis cs:47.9,3);
\draw[draw=none, only marks, fill=mycolor1] (axis cs:47.9,0) rectangle (axis cs:48,2);
\draw[draw=none, only marks, fill=mycolor2] (axis cs:47.9,2) rectangle (axis cs:48,3);
\draw[draw=none, only marks, fill=mycolor3] (axis cs:47.9,3) rectangle (axis cs:48,3);
\draw[draw=none, only marks, fill=mycolor4] (axis cs:47.9,3) rectangle (axis cs:48,3);
\draw[draw=none, only marks, fill=mycolor1] (axis cs:48,0) rectangle (axis cs:49.5,2);
\draw[draw=none, only marks, fill=mycolor2] (axis cs:48,2) rectangle (axis cs:49.5,3);
\draw[draw=none, only marks, fill=mycolor3] (axis cs:48,3) rectangle (axis cs:49.5,4);
\draw[draw=none, only marks, fill=mycolor4] (axis cs:48,4) rectangle (axis cs:49.5,4);
\draw[draw=none, only marks, fill=mycolor1] (axis cs:49.5,0) rectangle (axis cs:49.7,2);
\draw[draw=none, only marks, fill=mycolor2] (axis cs:49.5,2) rectangle (axis cs:49.7,3);
\draw[draw=none, only marks, fill=mycolor3] (axis cs:49.5,3) rectangle (axis cs:49.7,4);
\draw[draw=none, only marks, fill=mycolor4] (axis cs:49.5,4) rectangle (axis cs:49.7,4);
\draw[draw=none, only marks, fill=mycolor1] (axis cs:49.7,0) rectangle (axis cs:49.9,2);
\draw[draw=none, only marks, fill=mycolor2] (axis cs:49.7,2) rectangle (axis cs:49.9,3);
\draw[draw=none, only marks, fill=mycolor3] (axis cs:49.7,3) rectangle (axis cs:49.9,4);
\draw[draw=none, only marks, fill=mycolor4] (axis cs:49.7,4) rectangle (axis cs:49.9,4);
\draw[draw=none, only marks, fill=mycolor1] (axis cs:49.9,0) rectangle (axis cs:50,2);
\draw[draw=none, only marks, fill=mycolor2] (axis cs:49.9,2) rectangle (axis cs:50,3);
\draw[draw=none, only marks, fill=mycolor3] (axis cs:49.9,3) rectangle (axis cs:50,4);
\draw[draw=none, only marks, fill=mycolor4] (axis cs:49.9,4) rectangle (axis cs:50,4);
\draw[draw=none, only marks, fill=mycolor1] (axis cs:50,0) rectangle (axis cs:50.4,2);
\draw[draw=none, only marks, fill=mycolor2] (axis cs:50,2) rectangle (axis cs:50.4,3);
\draw[draw=none, only marks, fill=mycolor3] (axis cs:50,3) rectangle (axis cs:50.4,4);
\draw[draw=none, only marks, fill=mycolor4] (axis cs:50,4) rectangle (axis cs:50.4,4);
\draw[draw=none, only marks, fill=mycolor1] (axis cs:50.4,0) rectangle (axis cs:50.5,2);
\draw[draw=none, only marks, fill=mycolor2] (axis cs:50.4,2) rectangle (axis cs:50.5,3);
\draw[draw=none, only marks, fill=mycolor3] (axis cs:50.4,3) rectangle (axis cs:50.5,4);
\draw[draw=none, only marks, fill=mycolor4] (axis cs:50.4,4) rectangle (axis cs:50.5,4);
\draw[draw=none, only marks, fill=mycolor1] (axis cs:50.5,0) rectangle (axis cs:50.7,2);
\draw[draw=none, only marks, fill=mycolor2] (axis cs:50.5,2) rectangle (axis cs:50.7,3);
\draw[draw=none, only marks, fill=mycolor3] (axis cs:50.5,3) rectangle (axis cs:50.7,4);
\draw[draw=none, only marks, fill=mycolor4] (axis cs:50.5,4) rectangle (axis cs:50.7,5);
\draw[draw=none, only marks, fill=mycolor1] (axis cs:50.7,0) rectangle (axis cs:51,2);
\draw[draw=none, only marks, fill=mycolor2] (axis cs:50.7,2) rectangle (axis cs:51,3);
\draw[draw=none, only marks, fill=mycolor3] (axis cs:50.7,3) rectangle (axis cs:51,5);
\draw[draw=none, only marks, fill=mycolor4] (axis cs:50.7,5) rectangle (axis cs:51,6);
\draw[draw=none, only marks, fill=mycolor1] (axis cs:51,0) rectangle (axis cs:52,2);
\draw[draw=none, only marks, fill=mycolor2] (axis cs:51,2) rectangle (axis cs:52,3);
\draw[draw=none, only marks, fill=mycolor3] (axis cs:51,3) rectangle (axis cs:52,5);
\draw[draw=none, only marks, fill=mycolor4] (axis cs:51,5) rectangle (axis cs:52,6);
\draw[draw=none, only marks, fill=mycolor1] (axis cs:52,0) rectangle (axis cs:52.5,2);
\draw[draw=none, only marks, fill=mycolor2] (axis cs:52,2) rectangle (axis cs:52.5,3);
\draw[draw=none, only marks, fill=mycolor3] (axis cs:52,3) rectangle (axis cs:52.5,4);
\draw[draw=none, only marks, fill=mycolor4] (axis cs:52,4) rectangle (axis cs:52.5,4);
\draw[draw=none, only marks, fill=mycolor1] (axis cs:52.5,0) rectangle (axis cs:53.5,2);
\draw[draw=none, only marks, fill=mycolor2] (axis cs:52.5,2) rectangle (axis cs:53.5,3);
\draw[draw=none, only marks, fill=mycolor3] (axis cs:52.5,3) rectangle (axis cs:53.5,4);
\draw[draw=none, only marks, fill=mycolor4] (axis cs:52.5,4) rectangle (axis cs:53.5,4);
\draw[draw=none, only marks, fill=mycolor1] (axis cs:53.5,0) rectangle (axis cs:54.5,2);
\draw[draw=none, only marks, fill=mycolor2] (axis cs:53.5,2) rectangle (axis cs:54.5,3);
\draw[draw=none, only marks, fill=mycolor3] (axis cs:53.5,3) rectangle (axis cs:54.5,3);
\draw[draw=none, only marks, fill=mycolor4] (axis cs:53.5,3) rectangle (axis cs:54.5,3);
\draw[draw=none, only marks, fill=mycolor1] (axis cs:54.5,0) rectangle (axis cs:54.7,2);
\draw[draw=none, only marks, fill=mycolor2] (axis cs:54.5,2) rectangle (axis cs:54.7,3);
\draw[draw=none, only marks, fill=mycolor3] (axis cs:54.5,3) rectangle (axis cs:54.7,4);
\draw[draw=none, only marks, fill=mycolor4] (axis cs:54.5,4) rectangle (axis cs:54.7,6);
\draw[draw=none, only marks, fill=mycolor1] (axis cs:54.7,0) rectangle (axis cs:54.9,2);
\draw[draw=none, only marks, fill=mycolor2] (axis cs:54.7,2) rectangle (axis cs:54.9,2);
\draw[draw=none, only marks, fill=mycolor3] (axis cs:54.7,2) rectangle (axis cs:54.9,3);
\draw[draw=none, only marks, fill=mycolor4] (axis cs:54.7,3) rectangle (axis cs:54.9,5);
\draw[draw=none, only marks, fill=mycolor1] (axis cs:54.9,0) rectangle (axis cs:55.7,2);
\draw[draw=none, only marks, fill=mycolor2] (axis cs:54.9,2) rectangle (axis cs:55.7,2);
\draw[draw=none, only marks, fill=mycolor3] (axis cs:54.9,2) rectangle (axis cs:55.7,3);
\draw[draw=none, only marks, fill=mycolor4] (axis cs:54.9,3) rectangle (axis cs:55.7,5);
\draw[draw=none, only marks, fill=mycolor1] (axis cs:55.7,0) rectangle (axis cs:56,2);
\draw[draw=none, only marks, fill=mycolor2] (axis cs:55.7,2) rectangle (axis cs:56,2);
\draw[draw=none, only marks, fill=mycolor3] (axis cs:55.7,2) rectangle (axis cs:56,3);
\draw[draw=none, only marks, fill=mycolor4] (axis cs:55.7,3) rectangle (axis cs:56,5);
\draw[draw=none, only marks, fill=mycolor1] (axis cs:56,0) rectangle (axis cs:57,2);
\draw[draw=none, only marks, fill=mycolor2] (axis cs:56,2) rectangle (axis cs:57,2);
\draw[draw=none, only marks, fill=mycolor3] (axis cs:56,2) rectangle (axis cs:57,3);
\draw[draw=none, only marks, fill=mycolor4] (axis cs:56,3) rectangle (axis cs:57,5);
\draw[draw=none, only marks, fill=mycolor1] (axis cs:57,0) rectangle (axis cs:57.9,2);
\draw[draw=none, only marks, fill=mycolor2] (axis cs:57,2) rectangle (axis cs:57.9,2);
\draw[draw=none, only marks, fill=mycolor3] (axis cs:57,2) rectangle (axis cs:57.9,2);
\draw[draw=none, only marks, fill=mycolor4] (axis cs:57,2) rectangle (axis cs:57.9,3);
\draw[draw=none, only marks, fill=mycolor1] (axis cs:57.9,0) rectangle (axis cs:58.5,2);
\draw[draw=none, only marks, fill=mycolor2] (axis cs:57.9,2) rectangle (axis cs:58.5,2);
\draw[draw=none, only marks, fill=mycolor3] (axis cs:57.9,2) rectangle (axis cs:58.5,2);
\draw[draw=none, only marks, fill=mycolor4] (axis cs:57.9,2) rectangle (axis cs:58.5,3);
\draw[draw=none, only marks, fill=mycolor1] (axis cs:58.5,0) rectangle (axis cs:58.7,2);
\draw[draw=none, only marks, fill=mycolor2] (axis cs:58.5,2) rectangle (axis cs:58.7,2);
\draw[draw=none, only marks, fill=mycolor3] (axis cs:58.5,2) rectangle (axis cs:58.7,2);
\draw[draw=none, only marks, fill=mycolor4] (axis cs:58.5,2) rectangle (axis cs:58.7,3);
\draw[draw=none, only marks, fill=mycolor1] (axis cs:58.7,0) rectangle (axis cs:59.4,2);
\draw[draw=none, only marks, fill=mycolor2] (axis cs:58.7,2) rectangle (axis cs:59.4,2);
\draw[draw=none, only marks, fill=mycolor3] (axis cs:58.7,2) rectangle (axis cs:59.4,2);
\draw[draw=none, only marks, fill=mycolor4] (axis cs:58.7,2) rectangle (axis cs:59.4,3);
\draw[draw=none, only marks, fill=mycolor1] (axis cs:59.4,0) rectangle (axis cs:59.5,2);
\draw[draw=none, only marks, fill=mycolor2] (axis cs:59.4,2) rectangle (axis cs:59.5,2);
\draw[draw=none, only marks, fill=mycolor3] (axis cs:59.4,2) rectangle (axis cs:59.5,2);
\draw[draw=none, only marks, fill=mycolor4] (axis cs:59.4,2) rectangle (axis cs:59.5,2);
\draw[draw=none, only marks, fill=mycolor1] (axis cs:59.5,0) rectangle (axis cs:59.7,2);
\draw[draw=none, only marks, fill=mycolor2] (axis cs:59.5,2) rectangle (axis cs:59.7,2);
\draw[draw=none, only marks, fill=mycolor3] (axis cs:59.5,2) rectangle (axis cs:59.7,2);
\draw[draw=none, only marks, fill=mycolor4] (axis cs:59.5,2) rectangle (axis cs:59.7,2);
\draw[draw=none, only marks, fill=mycolor1] (axis cs:59.7,0) rectangle (axis cs:60.7,2);
\draw[draw=none, only marks, fill=mycolor2] (axis cs:59.7,2) rectangle (axis cs:60.7,2);
\draw[draw=none, only marks, fill=mycolor3] (axis cs:59.7,2) rectangle (axis cs:60.7,2);
\draw[draw=none, only marks, fill=mycolor4] (axis cs:59.7,2) rectangle (axis cs:60.7,2);
\draw[draw=none, only marks, fill=mycolor1] (axis cs:60.7,0) rectangle (axis cs:62,2);
\draw[draw=none, only marks, fill=mycolor2] (axis cs:60.7,2) rectangle (axis cs:62,2);
\draw[draw=none, only marks, fill=mycolor3] (axis cs:60.7,2) rectangle (axis cs:62,2);
\draw[draw=none, only marks, fill=mycolor4] (axis cs:60.7,2) rectangle (axis cs:62,4);
\draw[draw=none, only marks, fill=mycolor1] (axis cs:62,0) rectangle (axis cs:62.9,2);
\draw[draw=none, only marks, fill=mycolor2] (axis cs:62,2) rectangle (axis cs:62.9,2);
\draw[draw=none, only marks, fill=mycolor3] (axis cs:62,2) rectangle (axis cs:62.9,2);
\draw[draw=none, only marks, fill=mycolor4] (axis cs:62,2) rectangle (axis cs:62.9,3);
\draw[draw=none, only marks, fill=mycolor1] (axis cs:62.9,0) rectangle (axis cs:63.7,1);
\draw[draw=none, only marks, fill=mycolor2] (axis cs:62.9,1) rectangle (axis cs:63.7,1);
\draw[draw=none, only marks, fill=mycolor3] (axis cs:62.9,1) rectangle (axis cs:63.7,1);
\draw[draw=none, only marks, fill=mycolor4] (axis cs:62.9,1) rectangle (axis cs:63.7,2);
\draw[draw=none, only marks, fill=mycolor1] (axis cs:63.7,0) rectangle (axis cs:64.5,0);
\draw[draw=none, only marks, fill=mycolor2] (axis cs:63.7,0) rectangle (axis cs:64.5,0);
\draw[draw=none, only marks, fill=mycolor3] (axis cs:63.7,0) rectangle (axis cs:64.5,0);
\draw[draw=none, only marks, fill=mycolor4] (axis cs:63.7,0) rectangle (axis cs:64.5,1);
\draw[draw=none, only marks, fill=mycolor1] (axis cs:64.5,0) rectangle (axis cs:68.7,0);
\draw[draw=none, only marks, fill=mycolor2] (axis cs:64.5,0) rectangle (axis cs:68.7,0);
\draw[draw=none, only marks, fill=mycolor3] (axis cs:64.5,0) rectangle (axis cs:68.7,0);
\draw[draw=none, only marks, fill=mycolor4] (axis cs:64.5,0) rectangle (axis cs:68.7,0);
\draw[draw=none, only marks, fill=mycolor1] (axis cs:68.7,0) rectangle (axis cs:70,0);
\draw[draw=none, only marks, fill=mycolor2] (axis cs:68.7,0) rectangle (axis cs:70,0);
\draw[draw=none, only marks, fill=mycolor3] (axis cs:68.7,0) rectangle (axis cs:70,0);
\draw[draw=none, only marks, fill=mycolor4] (axis cs:68.7,0) rectangle (axis cs:70,0);
\draw[draw=none, only marks, fill=mycolor1] (axis cs:70,0) rectangle (axis cs:72.5,0);
\draw[draw=none, only marks, fill=mycolor2] (axis cs:70,0) rectangle (axis cs:72.5,0);
\draw[draw=none, only marks, fill=mycolor3] (axis cs:70,0) rectangle (axis cs:72.5,0);
\draw[draw=none, only marks, fill=mycolor4] (axis cs:70,0) rectangle (axis cs:72.5,0);
\addplot [color=blue, dashdotted, line width=2pt, forget plot]
  table[row sep=crcr]{%
0	6\\
73.5	6\\
};
\node[right, align=left]
at (axis cs:3,8) {$N_R(v_3)$};
\draw[draw=none, only marks, fill=mycolor1] (axis cs:2,9) rectangle (axis cs:14,9.5);
\node[right, align=left]
at (axis cs:17.5,8) {$N_{e_2,v_3}$};
\draw[draw=none, only marks, fill=mycolor2] (axis cs:16.5,9) rectangle (axis cs:28.5,9.5);
\node[right, align=left]
at (axis cs:32,8) {$N_{e_5,v_3}$};
\draw[draw=none, only marks, fill=mycolor3] (axis cs:31,9) rectangle (axis cs:43,9.5);
\node[right, align=left]
at (axis cs:46.5,8) {$N_{e_6,v_3}$};
\draw[draw=none, only marks, fill=mycolor4] (axis cs:45.5,9) rectangle (axis cs:57.5,9.5);
\node[right, align=left]
at (axis cs:61,8) {Failure};
\draw[draw=none, only marks, fill=mycolor5] (axis cs:60,9) rectangle (axis cs:72,9.5);
\node[right, align=left]
at (axis cs:65,5) {$C_{v_3}$};
\end{axis}

\begin{axis}[%
width=7cm,
height=3.8cm,
at={(0in,0in)},
scale only axis,
xmin=0,
xmax=1,
ymin=0,
ymax=1,
axis line style={draw=none},
ticks=none
]
\end{axis}
\end{tikzpicture}%

%% file: plots/case_e5_v1.tikz
%
%
\definecolor{mycolor1}{rgb}{0.49020,0.49804,0.48627}%
\definecolor{mycolor2}{rgb}{0.98824,0.69020,0.00392}%
\definecolor{mycolor3}{rgb}{0.58824,0.43137,0.74118}%
\begin{tikzpicture}

\begin{axis}[%
width=7cm,
height=3.8cm,
at={(0.541in,0.49in)},
scale only axis,
xmin=0,
xmax=73.5,
xlabel style={font=\color{white!15!black}},
xlabel={$t_c$},
ymin=0,
ymax=3.5,
ylabel style={font=\color{white!15!black}},
ylabel={Number of UAVs at $e_5$},
]
\draw[draw=none, only marks, fill=mycolor1] (axis cs:0,0) rectangle (axis cs:12.6,0);
\draw[draw=none, only marks, fill=mycolor1] (axis cs:12.6,0) rectangle (axis cs:12.8,0);
\draw[draw=none, only marks, fill=mycolor1] (axis cs:12.8,0) rectangle (axis cs:14.7,0);
\draw[draw=none, only marks, fill=mycolor1] (axis cs:14.7,0) rectangle (axis cs:15.5,0);
\draw[draw=none, only marks, fill=mycolor1] (axis cs:15.5,0) rectangle (axis cs:15.8,0);
\draw[draw=none, only marks, fill=mycolor1] (axis cs:15.8,0) rectangle (axis cs:16.8,0);
\draw[draw=none, only marks, fill=mycolor1] (axis cs:16.8,0) rectangle (axis cs:17.5,1);
\draw[draw=none, only marks, fill=mycolor2] (axis cs:16.8,0) rectangle (axis cs:17.5,1);
\draw[draw=none, only marks, fill=mycolor1] (axis cs:17.5,0) rectangle (axis cs:17.6,1);
\draw[draw=none, only marks, fill=mycolor2] (axis cs:17.5,0) rectangle (axis cs:17.6,1);
\draw[draw=none, only marks, fill=mycolor1] (axis cs:17.6,0) rectangle (axis cs:18.5,1);
\draw[draw=none, only marks, fill=mycolor2] (axis cs:17.6,0) rectangle (axis cs:18.5,1);
\draw[draw=none, only marks, fill=mycolor1] (axis cs:18.5,0) rectangle (axis cs:19.5,1);
\draw[draw=none, only marks, fill=mycolor3] (axis cs:18.5,0) rectangle (axis cs:19.5,1);
\draw[draw=none, only marks, fill=mycolor1] (axis cs:19.5,0) rectangle (axis cs:19.7,1);
\draw[draw=none, only marks, fill=mycolor3] (axis cs:19.5,0) rectangle (axis cs:19.7,1);
\draw[draw=none, only marks, fill=mycolor1] (axis cs:19.7,0) rectangle (axis cs:20.8,1);
\draw[draw=none, only marks, fill=mycolor3] (axis cs:19.7,0) rectangle (axis cs:20.8,1);
\draw[draw=none, only marks, fill=mycolor1] (axis cs:20.8,0) rectangle (axis cs:21.8,1);
\draw[draw=none, only marks, fill=mycolor3] (axis cs:20.8,0) rectangle (axis cs:21.8,1);
\draw[draw=none, only marks, fill=mycolor1] (axis cs:21.8,0) rectangle (axis cs:23.3,1);
\draw[draw=none, only marks, fill=mycolor3] (axis cs:21.8,0) rectangle (axis cs:23.3,1);
\draw[draw=none, only marks, fill=mycolor1] (axis cs:23.3,0) rectangle (axis cs:23.5,1);
\draw[draw=none, only marks, fill=mycolor3] (axis cs:23.3,0) rectangle (axis cs:23.5,1);
\draw[draw=none, only marks, fill=mycolor1] (axis cs:23.5,0) rectangle (axis cs:24.5,1);
\draw[draw=none, only marks, fill=mycolor3] (axis cs:23.5,0) rectangle (axis cs:24.5,1);
\draw[draw=none, only marks, fill=mycolor1] (axis cs:24.5,0) rectangle (axis cs:25.8,1);
\draw[draw=none, only marks, fill=mycolor3] (axis cs:24.5,0) rectangle (axis cs:25.8,1);
\draw[draw=none, only marks, fill=mycolor1] (axis cs:25.8,0) rectangle (axis cs:25.9,1);
\draw[draw=none, only marks, fill=mycolor3] (axis cs:25.8,0) rectangle (axis cs:25.9,1);
\draw[draw=none, only marks, fill=mycolor1] (axis cs:25.9,0) rectangle (axis cs:26.3,1);
\draw[draw=none, only marks, fill=mycolor3] (axis cs:25.9,0) rectangle (axis cs:26.3,1);
\draw[draw=none, only marks, fill=mycolor1] (axis cs:26.3,0) rectangle (axis cs:26.8,1);
\draw[draw=none, only marks, fill=mycolor3] (axis cs:26.3,0) rectangle (axis cs:26.8,1);
\draw[draw=none, only marks, fill=mycolor1] (axis cs:26.8,0) rectangle (axis cs:27.3,0);
\draw[draw=none, only marks, fill=mycolor1] (axis cs:27.3,0) rectangle (axis cs:28,1);
\draw[draw=none, only marks, fill=mycolor3] (axis cs:27.3,0) rectangle (axis cs:28,1);
\draw[draw=none, only marks, fill=mycolor1] (axis cs:28,0) rectangle (axis cs:28.9,1);
\draw[draw=none, only marks, fill=mycolor3] (axis cs:28,0) rectangle (axis cs:28.9,1);
\draw[draw=none, only marks, fill=mycolor1] (axis cs:28.9,0) rectangle (axis cs:29.5,1);
\draw[draw=none, only marks, fill=mycolor3] (axis cs:28.9,0) rectangle (axis cs:29.5,1);
\draw[draw=none, only marks, fill=mycolor1] (axis cs:29.5,0) rectangle (axis cs:29.6,1);
\draw[draw=none, only marks, fill=mycolor3] (axis cs:29.5,0) rectangle (axis cs:29.6,1);
\draw[draw=none, only marks, fill=mycolor1] (axis cs:29.6,0) rectangle (axis cs:29.9,1);
\draw[draw=none, only marks, fill=mycolor3] (axis cs:29.6,0) rectangle (axis cs:29.9,1);
\draw[draw=none, only marks, fill=mycolor1] (axis cs:29.9,0) rectangle (axis cs:30,2);
\draw[draw=none, only marks, fill=mycolor1] (axis cs:30,0) rectangle (axis cs:31,2);
\draw[draw=none, only marks, fill=mycolor1] (axis cs:31,0) rectangle (axis cs:31.3,2);
\draw[draw=none, only marks, fill=mycolor1] (axis cs:31.3,0) rectangle (axis cs:32,2);
\draw[draw=none, only marks, fill=mycolor2] (axis cs:31.3,0) rectangle (axis cs:32,1);
\draw[draw=none, only marks, fill=mycolor3] (axis cs:31.3,1) rectangle (axis cs:32,2);
\draw[draw=none, only marks, fill=mycolor1] (axis cs:32,0) rectangle (axis cs:32.3,2);
\draw[draw=none, only marks, fill=mycolor2] (axis cs:32,0) rectangle (axis cs:32.3,1);
\draw[draw=none, only marks, fill=mycolor3] (axis cs:32,1) rectangle (axis cs:32.3,2);
\draw[draw=none, only marks, fill=mycolor1] (axis cs:32.3,0) rectangle (axis cs:32.5,2);
\draw[draw=none, only marks, fill=mycolor2] (axis cs:32.3,0) rectangle (axis cs:32.5,1);
\draw[draw=none, only marks, fill=mycolor3] (axis cs:32.3,1) rectangle (axis cs:32.5,2);
\draw[draw=none, only marks, fill=mycolor1] (axis cs:32.5,0) rectangle (axis cs:33.9,2);
\draw[draw=none, only marks, fill=mycolor1] (axis cs:33.9,0) rectangle (axis cs:34.5,2);
\draw[draw=none, only marks, fill=mycolor2] (axis cs:33.9,0) rectangle (axis cs:34.5,2);
\draw[draw=none, only marks, fill=mycolor1] (axis cs:34.5,0) rectangle (axis cs:34.6,2);
\draw[draw=none, only marks, fill=mycolor2] (axis cs:34.5,0) rectangle (axis cs:34.6,2);
\draw[draw=none, only marks, fill=mycolor1] (axis cs:34.6,0) rectangle (axis cs:34.8,2);
\draw[draw=none, only marks, fill=mycolor2] (axis cs:34.6,0) rectangle (axis cs:34.8,2);
\draw[draw=none, only marks, fill=mycolor1] (axis cs:34.8,0) rectangle (axis cs:34.9,2);
\draw[draw=none, only marks, fill=mycolor2] (axis cs:34.8,0) rectangle (axis cs:34.9,2);
\draw[draw=none, only marks, fill=mycolor1] (axis cs:34.9,0) rectangle (axis cs:35.5,2);
\draw[draw=none, only marks, fill=mycolor2] (axis cs:34.9,0) rectangle (axis cs:35.5,2);
\draw[draw=none, only marks, fill=mycolor1] (axis cs:35.5,0) rectangle (axis cs:36,2);
\draw[draw=none, only marks, fill=mycolor1] (axis cs:36,0) rectangle (axis cs:36.3,2);
\draw[draw=none, only marks, fill=mycolor1] (axis cs:36.3,0) rectangle (axis cs:36.5,2);
\draw[draw=none, only marks, fill=mycolor1] (axis cs:36.5,0) rectangle (axis cs:37,2);
\draw[draw=none, only marks, fill=mycolor1] (axis cs:37,0) rectangle (axis cs:37.3,2);
\draw[draw=none, only marks, fill=mycolor1] (axis cs:37.3,0) rectangle (axis cs:37.4,1);
\draw[draw=none, only marks, fill=mycolor3] (axis cs:37.3,0) rectangle (axis cs:37.4,1);
\draw[draw=none, only marks, fill=mycolor1] (axis cs:37.4,0) rectangle (axis cs:37.5,1);
\draw[draw=none, only marks, fill=mycolor1] (axis cs:37.5,0) rectangle (axis cs:38.5,1);
\draw[draw=none, only marks, fill=mycolor1] (axis cs:38.5,0) rectangle (axis cs:38.9,1);
\draw[draw=none, only marks, fill=mycolor1] (axis cs:38.9,0) rectangle (axis cs:39.4,1);
\draw[draw=none, only marks, fill=mycolor1] (axis cs:39.4,0) rectangle (axis cs:39.5,1);
\draw[draw=none, only marks, fill=mycolor1] (axis cs:39.5,0) rectangle (axis cs:39.9,1);
\draw[draw=none, only marks, fill=mycolor1] (axis cs:39.9,0) rectangle (axis cs:40.4,0);
\draw[draw=none, only marks, fill=mycolor1] (axis cs:40.4,0) rectangle (axis cs:40.5,0);
\draw[draw=none, only marks, fill=mycolor1] (axis cs:40.5,0) rectangle (axis cs:41.3,1);
\draw[draw=none, only marks, fill=mycolor1] (axis cs:41.3,0) rectangle (axis cs:41.4,1);
\draw[draw=none, only marks, fill=mycolor1] (axis cs:41.4,0) rectangle (axis cs:41.5,2);
\draw[draw=none, only marks, fill=mycolor1] (axis cs:41.5,0) rectangle (axis cs:42,2);
\draw[draw=none, only marks, fill=mycolor1] (axis cs:42,0) rectangle (axis cs:43.5,2);
\draw[draw=none, only marks, fill=mycolor1] (axis cs:43.5,0) rectangle (axis cs:44.5,2);
\draw[draw=none, only marks, fill=mycolor2] (axis cs:43.5,0) rectangle (axis cs:44.5,1);
\draw[draw=none, only marks, fill=mycolor3] (axis cs:43.5,1) rectangle (axis cs:44.5,2);
\draw[draw=none, only marks, fill=mycolor1] (axis cs:44.5,0) rectangle (axis cs:45.3,2);
\draw[draw=none, only marks, fill=mycolor2] (axis cs:44.5,0) rectangle (axis cs:45.3,1);
\draw[draw=none, only marks, fill=mycolor3] (axis cs:44.5,1) rectangle (axis cs:45.3,2);
\draw[draw=none, only marks, fill=mycolor1] (axis cs:45.3,0) rectangle (axis cs:45.4,2);
\draw[draw=none, only marks, fill=mycolor2] (axis cs:45.3,0) rectangle (axis cs:45.4,1);
\draw[draw=none, only marks, fill=mycolor3] (axis cs:45.3,1) rectangle (axis cs:45.4,2);
\draw[draw=none, only marks, fill=mycolor1] (axis cs:45.4,0) rectangle (axis cs:45.5,2);
\draw[draw=none, only marks, fill=mycolor2] (axis cs:45.4,0) rectangle (axis cs:45.5,1);
\draw[draw=none, only marks, fill=mycolor3] (axis cs:45.4,1) rectangle (axis cs:45.5,2);
\draw[draw=none, only marks, fill=mycolor1] (axis cs:45.5,0) rectangle (axis cs:46.3,2);
\draw[draw=none, only marks, fill=mycolor2] (axis cs:45.5,0) rectangle (axis cs:46.3,1);
\draw[draw=none, only marks, fill=mycolor3] (axis cs:45.5,1) rectangle (axis cs:46.3,2);
\draw[draw=none, only marks, fill=mycolor1] (axis cs:46.3,0) rectangle (axis cs:46.4,2);
\draw[draw=none, only marks, fill=mycolor2] (axis cs:46.3,0) rectangle (axis cs:46.4,1);
\draw[draw=none, only marks, fill=mycolor3] (axis cs:46.3,1) rectangle (axis cs:46.4,2);
\draw[draw=none, only marks, fill=mycolor1] (axis cs:46.4,0) rectangle (axis cs:46.5,2);
\draw[draw=none, only marks, fill=mycolor3] (axis cs:46.4,0) rectangle (axis cs:46.5,2);
\draw[draw=none, only marks, fill=mycolor1] (axis cs:46.5,0) rectangle (axis cs:46.7,2);
\draw[draw=none, only marks, fill=mycolor3] (axis cs:46.5,0) rectangle (axis cs:46.7,2);
\draw[draw=none, only marks, fill=mycolor1] (axis cs:46.7,0) rectangle (axis cs:47.9,2);
\draw[draw=none, only marks, fill=mycolor3] (axis cs:46.7,0) rectangle (axis cs:47.9,2);
\draw[draw=none, only marks, fill=mycolor1] (axis cs:47.9,0) rectangle (axis cs:48,2);
\draw[draw=none, only marks, fill=mycolor3] (axis cs:47.9,0) rectangle (axis cs:48,2);
\draw[draw=none, only marks, fill=mycolor1] (axis cs:48,0) rectangle (axis cs:49.5,2);
\draw[draw=none, only marks, fill=mycolor2] (axis cs:48,0) rectangle (axis cs:49.5,1);
\draw[draw=none, only marks, fill=mycolor3] (axis cs:48,1) rectangle (axis cs:49.5,2);
\draw[draw=none, only marks, fill=mycolor1] (axis cs:49.5,0) rectangle (axis cs:49.7,2);
\draw[draw=none, only marks, fill=mycolor2] (axis cs:49.5,0) rectangle (axis cs:49.7,1);
\draw[draw=none, only marks, fill=mycolor3] (axis cs:49.5,1) rectangle (axis cs:49.7,2);
\draw[draw=none, only marks, fill=mycolor1] (axis cs:49.7,0) rectangle (axis cs:49.9,2);
\draw[draw=none, only marks, fill=mycolor2] (axis cs:49.7,0) rectangle (axis cs:49.9,1);
\draw[draw=none, only marks, fill=mycolor3] (axis cs:49.7,1) rectangle (axis cs:49.9,2);
\draw[draw=none, only marks, fill=mycolor1] (axis cs:49.9,0) rectangle (axis cs:50,2);
\draw[draw=none, only marks, fill=mycolor2] (axis cs:49.9,0) rectangle (axis cs:50,1);
\draw[draw=none, only marks, fill=mycolor3] (axis cs:49.9,1) rectangle (axis cs:50,2);
\draw[draw=none, only marks, fill=mycolor1] (axis cs:50,0) rectangle (axis cs:50.4,2);
\draw[draw=none, only marks, fill=mycolor2] (axis cs:50,0) rectangle (axis cs:50.4,1);
\draw[draw=none, only marks, fill=mycolor3] (axis cs:50,1) rectangle (axis cs:50.4,2);
\draw[draw=none, only marks, fill=mycolor1] (axis cs:50.4,0) rectangle (axis cs:50.5,2);
\draw[draw=none, only marks, fill=mycolor2] (axis cs:50.4,0) rectangle (axis cs:50.5,1);
\draw[draw=none, only marks, fill=mycolor3] (axis cs:50.4,1) rectangle (axis cs:50.5,2);
\draw[draw=none, only marks, fill=mycolor1] (axis cs:50.5,0) rectangle (axis cs:50.7,1);
\draw[draw=none, only marks, fill=mycolor2] (axis cs:50.5,0) rectangle (axis cs:50.7,1);
\draw[draw=none, only marks, fill=mycolor1] (axis cs:50.7,0) rectangle (axis cs:51,2);
\draw[draw=none, only marks, fill=mycolor2] (axis cs:50.7,0) rectangle (axis cs:51,2);
\draw[draw=none, only marks, fill=mycolor1] (axis cs:51,0) rectangle (axis cs:52,2);
\draw[draw=none, only marks, fill=mycolor2] (axis cs:51,0) rectangle (axis cs:52,2);
\draw[draw=none, only marks, fill=mycolor1] (axis cs:52,0) rectangle (axis cs:52.5,2);
\draw[draw=none, only marks, fill=mycolor2] (axis cs:52,0) rectangle (axis cs:52.5,1);
\draw[draw=none, only marks, fill=mycolor3] (axis cs:52,1) rectangle (axis cs:52.5,2);
\draw[draw=none, only marks, fill=mycolor1] (axis cs:52.5,0) rectangle (axis cs:53.5,2);
\draw[draw=none, only marks, fill=mycolor2] (axis cs:52.5,0) rectangle (axis cs:53.5,1);
\draw[draw=none, only marks, fill=mycolor3] (axis cs:52.5,1) rectangle (axis cs:53.5,2);
\draw[draw=none, only marks, fill=mycolor1] (axis cs:53.5,0) rectangle (axis cs:54.5,1);
\draw[draw=none, only marks, fill=mycolor3] (axis cs:53.5,0) rectangle (axis cs:54.5,1);
\draw[draw=none, only marks, fill=mycolor1] (axis cs:54.5,0) rectangle (axis cs:54.7,1);
\draw[draw=none, only marks, fill=mycolor2] (axis cs:54.5,0) rectangle (axis cs:54.7,1);
\draw[draw=none, only marks, fill=mycolor1] (axis cs:54.7,0) rectangle (axis cs:54.9,1);
\draw[draw=none, only marks, fill=mycolor2] (axis cs:54.7,0) rectangle (axis cs:54.9,1);
\draw[draw=none, only marks, fill=mycolor1] (axis cs:54.9,0) rectangle (axis cs:55.7,1);
\draw[draw=none, only marks, fill=mycolor2] (axis cs:54.9,0) rectangle (axis cs:55.7,1);
\draw[draw=none, only marks, fill=mycolor1] (axis cs:55.7,0) rectangle (axis cs:56,1);
\draw[draw=none, only marks, fill=mycolor2] (axis cs:55.7,0) rectangle (axis cs:56,1);
\draw[draw=none, only marks, fill=mycolor1] (axis cs:56,0) rectangle (axis cs:57,1);
\draw[draw=none, only marks, fill=mycolor2] (axis cs:56,0) rectangle (axis cs:57,1);
\draw[draw=none, only marks, fill=mycolor1] (axis cs:57,0) rectangle (axis cs:57.9,1);
\draw[draw=none, only marks, fill=mycolor3] (axis cs:57,0) rectangle (axis cs:57.9,1);
\draw[draw=none, only marks, fill=mycolor1] (axis cs:57.9,0) rectangle (axis cs:58.5,1);
\draw[draw=none, only marks, fill=mycolor3] (axis cs:57.9,0) rectangle (axis cs:58.5,1);
\draw[draw=none, only marks, fill=mycolor1] (axis cs:58.5,0) rectangle (axis cs:58.7,1);
\draw[draw=none, only marks, fill=mycolor3] (axis cs:58.5,0) rectangle (axis cs:58.7,1);
\draw[draw=none, only marks, fill=mycolor1] (axis cs:58.7,0) rectangle (axis cs:59.4,1);
\draw[draw=none, only marks, fill=mycolor3] (axis cs:58.7,0) rectangle (axis cs:59.4,1);
\draw[draw=none, only marks, fill=mycolor1] (axis cs:59.4,0) rectangle (axis cs:59.5,1);
\draw[draw=none, only marks, fill=mycolor3] (axis cs:59.4,0) rectangle (axis cs:59.5,1);
\draw[draw=none, only marks, fill=mycolor1] (axis cs:59.5,0) rectangle (axis cs:59.7,1);
\draw[draw=none, only marks, fill=mycolor3] (axis cs:59.5,0) rectangle (axis cs:59.7,1);
\draw[draw=none, only marks, fill=mycolor1] (axis cs:59.7,0) rectangle (axis cs:60.7,1);
\draw[draw=none, only marks, fill=mycolor3] (axis cs:59.7,0) rectangle (axis cs:60.7,1);
\draw[draw=none, only marks, fill=mycolor1] (axis cs:60.7,0) rectangle (axis cs:62,0);
\draw[draw=none, only marks, fill=mycolor1] (axis cs:62,0) rectangle (axis cs:62.9,0);
\draw[draw=none, only marks, fill=mycolor1] (axis cs:62.9,0) rectangle (axis cs:63.7,0);
\draw[draw=none, only marks, fill=mycolor1] (axis cs:63.7,0) rectangle (axis cs:64.5,0);
\draw[draw=none, only marks, fill=mycolor1] (axis cs:64.5,0) rectangle (axis cs:68.7,0);
\draw[draw=none, only marks, fill=mycolor1] (axis cs:68.7,0) rectangle (axis cs:70,0);
\draw[draw=none, only marks, fill=mycolor1] (axis cs:70,0) rectangle (axis cs:72.5,0);
\node[right, align=left]
at (axis cs:16,2.7){$N_{e_5,v_3}$}; 
\draw[draw=none, only marks, fill=mycolor2] (axis cs:15,3) rectangle (axis cs:27,3.2);
\node[right, align=left]
at (axis cs:30.5,2.7){$N_{e_5,v_4}$}; 
\draw[draw=none, only marks, fill=mycolor3] (axis cs:29.5,3) rectangle (axis cs:41.5,3.2);
\node[right, align=left]
at (axis cs:45,2.7) {Failure};
\draw[draw=none, only marks, fill=mycolor1] (axis cs:44,3) rectangle (axis cs:56,3.2);
\end{axis}

\begin{axis}[%
width=7cm,
height=3.8cm,
at={(0in,0in)},
scale only axis,
xmin=0,
xmax=1,
ymin=0,
ymax=1,
axis line style={draw=none},
ticks=none,
]
\end{axis}
\end{tikzpicture}%

%% file: plots/case_e6_v1.tikz
%
%
\definecolor{mycolor1}{rgb}{0.49020,0.49804,0.48627}%
\definecolor{mycolor2}{rgb}{0.38431,0.34510,0.76863}%
\definecolor{mycolor3}{rgb}{0.99216,0.27451,0.34902}%
\definecolor{mycolor4}{rgb}{0.53725,0.63529,0.01176}%
\begin{tikzpicture}

\begin{axis}[%
width=7cm,
height=3.8cm,
at={(0.52in,0.49in)},
scale only axis,
xmin=0,
xmax=73.5,
xlabel style={font=\color{white!15!black}},
xlabel={$t_c$},
ymin=0,
ymax=4.5,
ylabel style={font=\color{white!15!black}},
ylabel={Number of UAVs at $e_6$},
]
\draw[draw=none, only marks, fill=mycolor1] (axis cs:0,0) rectangle (axis cs:12.6,0);
\draw[draw=none, only marks, fill=mycolor1] (axis cs:12.6,0) rectangle (axis cs:12.8,0);
\draw[draw=none, only marks, fill=mycolor1] (axis cs:12.8,0) rectangle (axis cs:14.7,0);
\draw[draw=none, only marks, fill=mycolor1] (axis cs:14.7,0) rectangle (axis cs:15.5,0);
\draw[draw=none, only marks, fill=mycolor1] (axis cs:15.5,0) rectangle (axis cs:15.8,0);
\draw[draw=none, only marks, fill=mycolor1] (axis cs:15.8,0) rectangle (axis cs:16.8,0);
\draw[draw=none, only marks, fill=mycolor1] (axis cs:16.8,0) rectangle (axis cs:17.5,0);
\draw[draw=none, only marks, fill=mycolor1] (axis cs:17.5,0) rectangle (axis cs:17.6,0);
\draw[draw=none, only marks, fill=mycolor1] (axis cs:17.6,0) rectangle (axis cs:18.5,0);
\draw[draw=none, only marks, fill=mycolor1] (axis cs:18.5,0) rectangle (axis cs:19.5,1);
\draw[draw=none, only marks, fill=mycolor2] (axis cs:18.5,0) rectangle (axis cs:19.5,1);
\draw[draw=none, only marks, fill=mycolor1] (axis cs:19.5,0) rectangle (axis cs:19.7,1);
\draw[draw=none, only marks, fill=mycolor2] (axis cs:19.5,0) rectangle (axis cs:19.7,1);
\draw[draw=none, only marks, fill=mycolor1] (axis cs:19.7,0) rectangle (axis cs:20.8,1);
\draw[draw=none, only marks, fill=mycolor2] (axis cs:19.7,0) rectangle (axis cs:20.8,1);
\draw[draw=none, only marks, fill=mycolor1] (axis cs:20.8,0) rectangle (axis cs:21.8,1);
\draw[draw=none, only marks, fill=mycolor2] (axis cs:20.8,0) rectangle (axis cs:21.8,1);
\draw[draw=none, only marks, fill=mycolor1] (axis cs:21.8,0) rectangle (axis cs:23.3,1);
\draw[draw=none, only marks, fill=mycolor2] (axis cs:21.8,0) rectangle (axis cs:23.3,1);
\draw[draw=none, only marks, fill=mycolor1] (axis cs:23.3,0) rectangle (axis cs:23.5,1);
\draw[draw=none, only marks, fill=mycolor2] (axis cs:23.3,0) rectangle (axis cs:23.5,1);
\draw[draw=none, only marks, fill=mycolor1] (axis cs:23.5,0) rectangle (axis cs:24.5,1);
\draw[draw=none, only marks, fill=mycolor2] (axis cs:23.5,0) rectangle (axis cs:24.5,1);
\draw[draw=none, only marks, fill=mycolor1] (axis cs:24.5,0) rectangle (axis cs:25.8,1);
\draw[draw=none, only marks, fill=mycolor2] (axis cs:24.5,0) rectangle (axis cs:25.8,1);
\draw[draw=none, only marks, fill=mycolor1] (axis cs:25.8,0) rectangle (axis cs:25.9,1);
\draw[draw=none, only marks, fill=mycolor2] (axis cs:25.8,0) rectangle (axis cs:25.9,1);
\draw[draw=none, only marks, fill=mycolor1] (axis cs:25.9,0) rectangle (axis cs:26.3,1);
\draw[draw=none, only marks, fill=mycolor3] (axis cs:25.9,0) rectangle (axis cs:26.3,1);
\draw[draw=none, only marks, fill=mycolor1] (axis cs:26.3,0) rectangle (axis cs:26.8,1);
\draw[draw=none, only marks, fill=mycolor3] (axis cs:26.3,0) rectangle (axis cs:26.8,1);
\draw[draw=none, only marks, fill=mycolor1] (axis cs:26.8,0) rectangle (axis cs:27.3,1);
\draw[draw=none, only marks, fill=mycolor3] (axis cs:26.8,0) rectangle (axis cs:27.3,1);
\draw[draw=none, only marks, fill=mycolor1] (axis cs:27.3,0) rectangle (axis cs:28,1);
\draw[draw=none, only marks, fill=mycolor3] (axis cs:27.3,0) rectangle (axis cs:28,1);
\draw[draw=none, only marks, fill=mycolor1] (axis cs:28,0) rectangle (axis cs:28.9,1);
\draw[draw=none, only marks, fill=mycolor2] (axis cs:28,0) rectangle (axis cs:28.9,1);
\draw[draw=none, only marks, fill=mycolor1] (axis cs:28.9,0) rectangle (axis cs:29.5,1);
\draw[draw=none, only marks, fill=mycolor2] (axis cs:28.9,0) rectangle (axis cs:29.5,1);
\draw[draw=none, only marks, fill=mycolor1] (axis cs:29.5,0) rectangle (axis cs:29.6,0);
\draw[draw=none, only marks, fill=mycolor1] (axis cs:29.6,0) rectangle (axis cs:29.9,0);
\draw[draw=none, only marks, fill=mycolor1] (axis cs:29.9,0) rectangle (axis cs:30,0);
\draw[draw=none, only marks, fill=mycolor1] (axis cs:30,0) rectangle (axis cs:31,0);
\draw[draw=none, only marks, fill=mycolor1] (axis cs:31,0) rectangle (axis cs:31.3,1);
\draw[draw=none, only marks, fill=mycolor1] (axis cs:31.3,0) rectangle (axis cs:32,1);
\draw[draw=none, only marks, fill=mycolor2] (axis cs:31.3,0) rectangle (axis cs:32,1);
\draw[draw=none, only marks, fill=mycolor1] (axis cs:32,0) rectangle (axis cs:32.3,1);
\draw[draw=none, only marks, fill=mycolor2] (axis cs:32,0) rectangle (axis cs:32.3,1);
\draw[draw=none, only marks, fill=mycolor1] (axis cs:32.3,0) rectangle (axis cs:32.5,1);
\draw[draw=none, only marks, fill=mycolor2] (axis cs:32.3,0) rectangle (axis cs:32.5,1);
\draw[draw=none, only marks, fill=mycolor1] (axis cs:32.5,0) rectangle (axis cs:33.9,1);
\draw[draw=none, only marks, fill=mycolor1] (axis cs:33.9,0) rectangle (axis cs:34.5,1);
\draw[draw=none, only marks, fill=mycolor2] (axis cs:33.9,0) rectangle (axis cs:34.5,1);
\draw[draw=none, only marks, fill=mycolor1] (axis cs:34.5,0) rectangle (axis cs:34.6,1);
\draw[draw=none, only marks, fill=mycolor2] (axis cs:34.5,0) rectangle (axis cs:34.6,1);
\draw[draw=none, only marks, fill=mycolor1] (axis cs:34.6,0) rectangle (axis cs:34.8,1);
\draw[draw=none, only marks, fill=mycolor2] (axis cs:34.6,0) rectangle (axis cs:34.8,1);
\draw[draw=none, only marks, fill=mycolor1] (axis cs:34.8,0) rectangle (axis cs:34.9,1);
\draw[draw=none, only marks, fill=mycolor2] (axis cs:34.8,0) rectangle (axis cs:34.9,1);
\draw[draw=none, only marks, fill=mycolor1] (axis cs:34.9,0) rectangle (axis cs:35.5,1);
\draw[draw=none, only marks, fill=mycolor2] (axis cs:34.9,0) rectangle (axis cs:35.5,1);
\draw[draw=none, only marks, fill=mycolor1] (axis cs:35.5,0) rectangle (axis cs:36,2);
\draw[draw=none, only marks, fill=mycolor1] (axis cs:36,0) rectangle (axis cs:36.3,2);
\draw[draw=none, only marks, fill=mycolor1] (axis cs:36.3,0) rectangle (axis cs:36.5,2);
\draw[draw=none, only marks, fill=mycolor1] (axis cs:36.5,0) rectangle (axis cs:37,2);
\draw[draw=none, only marks, fill=mycolor1] (axis cs:37,0) rectangle (axis cs:37.3,2);
\draw[draw=none, only marks, fill=mycolor1] (axis cs:37.3,0) rectangle (axis cs:37.4,2);
\draw[draw=none, only marks, fill=mycolor4] (axis cs:37.3,0) rectangle (axis cs:37.4,1);
\draw[draw=none, only marks, fill=mycolor2] (axis cs:37.3,1) rectangle (axis cs:37.4,2);
\draw[draw=none, only marks, fill=mycolor1] (axis cs:37.4,0) rectangle (axis cs:37.5,2);
\draw[draw=none, only marks, fill=mycolor1] (axis cs:37.5,0) rectangle (axis cs:38.5,2);
\draw[draw=none, only marks, fill=mycolor1] (axis cs:38.5,0) rectangle (axis cs:38.9,2);
\draw[draw=none, only marks, fill=mycolor1] (axis cs:38.9,0) rectangle (axis cs:39.4,2);
\draw[draw=none, only marks, fill=mycolor1] (axis cs:39.4,0) rectangle (axis cs:39.5,2);
\draw[draw=none, only marks, fill=mycolor1] (axis cs:39.5,0) rectangle (axis cs:39.9,2);
\draw[draw=none, only marks, fill=mycolor1] (axis cs:39.9,0) rectangle (axis cs:40.4,2);
\draw[draw=none, only marks, fill=mycolor1] (axis cs:40.4,0) rectangle (axis cs:40.5,3);
\draw[draw=none, only marks, fill=mycolor1] (axis cs:40.5,0) rectangle (axis cs:41.3,3);
\draw[draw=none, only marks, fill=mycolor1] (axis cs:41.3,0) rectangle (axis cs:41.4,3);
\draw[draw=none, only marks, fill=mycolor1] (axis cs:41.4,0) rectangle (axis cs:41.5,3);
\draw[draw=none, only marks, fill=mycolor1] (axis cs:41.5,0) rectangle (axis cs:42,3);
\draw[draw=none, only marks, fill=mycolor1] (axis cs:42,0) rectangle (axis cs:43.5,2);
\draw[draw=none, only marks, fill=mycolor1] (axis cs:43.5,0) rectangle (axis cs:44.5,2);
\draw[draw=none, only marks, fill=mycolor2] (axis cs:43.5,0) rectangle (axis cs:44.5,2);
\draw[draw=none, only marks, fill=mycolor1] (axis cs:44.5,0) rectangle (axis cs:45.3,2);
\draw[draw=none, only marks, fill=mycolor2] (axis cs:44.5,0) rectangle (axis cs:45.3,2);
\draw[draw=none, only marks, fill=mycolor1] (axis cs:45.3,0) rectangle (axis cs:45.4,2);
\draw[draw=none, only marks, fill=mycolor2] (axis cs:45.3,0) rectangle (axis cs:45.4,2);
\draw[draw=none, only marks, fill=mycolor1] (axis cs:45.4,0) rectangle (axis cs:45.5,2);
\draw[draw=none, only marks, fill=mycolor2] (axis cs:45.4,0) rectangle (axis cs:45.5,2);
\draw[draw=none, only marks, fill=mycolor1] (axis cs:45.5,0) rectangle (axis cs:46.3,2);
\draw[draw=none, only marks, fill=mycolor2] (axis cs:45.5,0) rectangle (axis cs:46.3,2);
\draw[draw=none, only marks, fill=mycolor1] (axis cs:46.3,0) rectangle (axis cs:46.4,2);
\draw[draw=none, only marks, fill=mycolor2] (axis cs:46.3,0) rectangle (axis cs:46.4,2);
\draw[draw=none, only marks, fill=mycolor1] (axis cs:46.4,0) rectangle (axis cs:46.5,2);
\draw[draw=none, only marks, fill=mycolor2] (axis cs:46.4,0) rectangle (axis cs:46.5,2);
\draw[draw=none, only marks, fill=mycolor1] (axis cs:46.5,0) rectangle (axis cs:46.7,1);
\draw[draw=none, only marks, fill=mycolor2] (axis cs:46.5,0) rectangle (axis cs:46.7,1);
\draw[draw=none, only marks, fill=mycolor1] (axis cs:46.7,0) rectangle (axis cs:47.9,1);
\draw[draw=none, only marks, fill=mycolor2] (axis cs:46.7,0) rectangle (axis cs:47.9,1);
\draw[draw=none, only marks, fill=mycolor1] (axis cs:47.9,0) rectangle (axis cs:48,1);
\draw[draw=none, only marks, fill=mycolor2] (axis cs:47.9,0) rectangle (axis cs:48,1);
\draw[draw=none, only marks, fill=mycolor1] (axis cs:48,0) rectangle (axis cs:49.5,1);
\draw[draw=none, only marks, fill=mycolor2] (axis cs:48,0) rectangle (axis cs:49.5,1);
\draw[draw=none, only marks, fill=mycolor1] (axis cs:49.5,0) rectangle (axis cs:49.7,1);
\draw[draw=none, only marks, fill=mycolor2] (axis cs:49.5,0) rectangle (axis cs:49.7,1);
\draw[draw=none, only marks, fill=mycolor1] (axis cs:49.7,0) rectangle (axis cs:49.9,1);
\draw[draw=none, only marks, fill=mycolor2] (axis cs:49.7,0) rectangle (axis cs:49.9,1);
\draw[draw=none, only marks, fill=mycolor1] (axis cs:49.9,0) rectangle (axis cs:50,1);
\draw[draw=none, only marks, fill=mycolor2] (axis cs:49.9,0) rectangle (axis cs:50,1);
\draw[draw=none, only marks, fill=mycolor1] (axis cs:50,0) rectangle (axis cs:50.4,1);
\draw[draw=none, only marks, fill=mycolor2] (axis cs:50,0) rectangle (axis cs:50.4,1);
\draw[draw=none, only marks, fill=mycolor1] (axis cs:50.4,0) rectangle (axis cs:50.5,1);
\draw[draw=none, only marks, fill=mycolor2] (axis cs:50.4,0) rectangle (axis cs:50.5,1);
\draw[draw=none, only marks, fill=mycolor1] (axis cs:50.5,0) rectangle (axis cs:50.7,1);
\draw[draw=none, only marks, fill=mycolor4] (axis cs:50.5,0) rectangle (axis cs:50.7,1);
\draw[draw=none, only marks, fill=mycolor1] (axis cs:50.7,0) rectangle (axis cs:51,1);
\draw[draw=none, only marks, fill=mycolor4] (axis cs:50.7,0) rectangle (axis cs:51,1);
\draw[draw=none, only marks, fill=mycolor1] (axis cs:51,0) rectangle (axis cs:52,2);
\draw[draw=none, only marks, fill=mycolor4] (axis cs:51,0) rectangle (axis cs:52,1);
\draw[draw=none, only marks, fill=mycolor2] (axis cs:51,1) rectangle (axis cs:52,2);
\draw[draw=none, only marks, fill=mycolor1] (axis cs:52,0) rectangle (axis cs:52.5,2);
\draw[draw=none, only marks, fill=mycolor2] (axis cs:52,0) rectangle (axis cs:52.5,2);
\draw[draw=none, only marks, fill=mycolor1] (axis cs:52.5,0) rectangle (axis cs:53.5,2);
\draw[draw=none, only marks, fill=mycolor2] (axis cs:52.5,0) rectangle (axis cs:53.5,2);
\draw[draw=none, only marks, fill=mycolor1] (axis cs:53.5,0) rectangle (axis cs:54.5,2);
\draw[draw=none, only marks, fill=mycolor2] (axis cs:53.5,0) rectangle (axis cs:54.5,2);
\draw[draw=none, only marks, fill=mycolor1] (axis cs:54.5,0) rectangle (axis cs:54.7,2);
\draw[draw=none, only marks, fill=mycolor4] (axis cs:54.5,0) rectangle (axis cs:54.7,2);
\draw[draw=none, only marks, fill=mycolor1] (axis cs:54.7,0) rectangle (axis cs:54.9,2);
\draw[draw=none, only marks, fill=mycolor4] (axis cs:54.7,0) rectangle (axis cs:54.9,2);
\draw[draw=none, only marks, fill=mycolor1] (axis cs:54.9,0) rectangle (axis cs:55.7,2);
\draw[draw=none, only marks, fill=mycolor4] (axis cs:54.9,0) rectangle (axis cs:55.7,2);
\draw[draw=none, only marks, fill=mycolor1] (axis cs:55.7,0) rectangle (axis cs:56,2);
\draw[draw=none, only marks, fill=mycolor4] (axis cs:55.7,0) rectangle (axis cs:56,2);
\draw[draw=none, only marks, fill=mycolor1] (axis cs:56,0) rectangle (axis cs:57,2);
\draw[draw=none, only marks, fill=mycolor4] (axis cs:56,0) rectangle (axis cs:57,2);
\draw[draw=none, only marks, fill=mycolor1] (axis cs:57,0) rectangle (axis cs:57.9,2);
\draw[draw=none, only marks, fill=mycolor4] (axis cs:57,0) rectangle (axis cs:57.9,1);
\draw[draw=none, only marks, fill=mycolor2] (axis cs:57,1) rectangle (axis cs:57.9,2);
\draw[draw=none, only marks, fill=mycolor1] (axis cs:57.9,0) rectangle (axis cs:58.5,2);
\draw[draw=none, only marks, fill=mycolor4] (axis cs:57.9,0) rectangle (axis cs:58.5,1);
\draw[draw=none, only marks, fill=mycolor2] (axis cs:57.9,1) rectangle (axis cs:58.5,2);
\draw[draw=none, only marks, fill=mycolor1] (axis cs:58.5,0) rectangle (axis cs:58.7,2);
\draw[draw=none, only marks, fill=mycolor4] (axis cs:58.5,0) rectangle (axis cs:58.7,1);
\draw[draw=none, only marks, fill=mycolor2] (axis cs:58.5,1) rectangle (axis cs:58.7,2);
\draw[draw=none, only marks, fill=mycolor1] (axis cs:58.7,0) rectangle (axis cs:59.4,2);
\draw[draw=none, only marks, fill=mycolor4] (axis cs:58.7,0) rectangle (axis cs:59.4,1);
\draw[draw=none, only marks, fill=mycolor2] (axis cs:58.7,1) rectangle (axis cs:59.4,2);
\draw[draw=none, only marks, fill=mycolor1] (axis cs:59.4,0) rectangle (axis cs:59.5,2);
\draw[draw=none, only marks, fill=mycolor2] (axis cs:59.4,0) rectangle (axis cs:59.5,2);
\draw[draw=none, only marks, fill=mycolor1] (axis cs:59.5,0) rectangle (axis cs:59.7,2);
\draw[draw=none, only marks, fill=mycolor2] (axis cs:59.5,0) rectangle (axis cs:59.7,2);
\draw[draw=none, only marks, fill=mycolor1] (axis cs:59.7,0) rectangle (axis cs:60.7,2);
\draw[draw=none, only marks, fill=mycolor2] (axis cs:59.7,0) rectangle (axis cs:60.7,2);
\draw[draw=none, only marks, fill=mycolor1] (axis cs:60.7,0) rectangle (axis cs:62,2);
\draw[draw=none, only marks, fill=mycolor4] (axis cs:60.7,0) rectangle (axis cs:62,2);
\draw[draw=none, only marks, fill=mycolor1] (axis cs:62,0) rectangle (axis cs:62.9,1);
\draw[draw=none, only marks, fill=mycolor4] (axis cs:62,0) rectangle (axis cs:62.9,1);
\draw[draw=none, only marks, fill=mycolor1] (axis cs:62.9,0) rectangle (axis cs:63.7,1);
\draw[draw=none, only marks, fill=mycolor4] (axis cs:62.9,0) rectangle (axis cs:63.7,1);
\draw[draw=none, only marks, fill=mycolor1] (axis cs:63.7,0) rectangle (axis cs:64.5,1);
\draw[draw=none, only marks, fill=mycolor4] (axis cs:63.7,0) rectangle (axis cs:64.5,1);
\draw[draw=none, only marks, fill=mycolor1] (axis cs:64.5,0) rectangle (axis cs:68.7,0);
\draw[draw=none, only marks, fill=mycolor1] (axis cs:68.7,0) rectangle (axis cs:70,0);
\draw[draw=none, only marks, fill=mycolor1] (axis cs:70,0) rectangle (axis cs:72.5,0);
\node[right, align=left]
at (axis cs:16,3.7){$N_{e_6,v_3}$}; 
\draw[draw=none, only marks, fill=mycolor4] (axis cs:15,4) rectangle (axis cs:27,4.2);
\node[right, align=left]
at (axis cs:30.5,3.7){$N_{e_6,v_4}$}; 
\draw[draw=none, only marks, fill=mycolor3] (axis cs:29.5,4) rectangle (axis cs:41.5,4.2);
\node[right, align=left]
at (axis cs:45,3.7){$N_{e_6,v_7}$}; 
\draw[draw=none, only marks, fill=mycolor2] (axis cs:44,4) rectangle (axis cs:56,4.2);
\node[right, align=left]
at (axis cs:59.5,3.7) {Failure};
\draw[draw=none, only marks, fill=mycolor1] (axis cs:58.5,4) rectangle (axis cs:70.5,4.2);
\end{axis}

\begin{axis}[%
width=7cm,
height=3.8cm,
at={(0in,0in)},
scale only axis,
xmin=0,
xmax=1,
ymin=0,
ymax=1,
axis line style={draw=none},
ticks=none,
]
\end{axis}
\end{tikzpicture}%

%% file: plots/case_v3_v2.tikz
%
%
\definecolor{mycolor1}{rgb}{0.45882,0.73333,0.99216}%
\definecolor{mycolor2}{rgb}{1.00000,0.81176,0.86275}%
\definecolor{mycolor3}{rgb}{0.98824,0.69020,0.00392}%
\definecolor{mycolor4}{rgb}{0.53725,0.63529,0.01176}%
\definecolor{mycolor5}{rgb}{0.49020,0.49804,0.48627}%
\begin{tikzpicture}

\begin{axis}[%
width=7cm,
height=3.8cm,
at={(0.534in,0.517in)},
scale only axis,
xmin=0,
xmax=73.5,
xlabel style={font=\color{white!15!black}},
xlabel={$t_c$},
ymin=0,
ymax=10,
ylabel style={font=\color{white!15!black}},
ylabel={Number of UAVs at $v_3$},
]
\draw[draw=none, only marks, fill=mycolor1] (axis cs:0,0) rectangle (axis cs:12.6,4);
\draw[draw=none, only marks, fill=mycolor2] (axis cs:0,4) rectangle (axis cs:12.6,4);
\draw[draw=none, only marks, fill=mycolor3] (axis cs:0,4) rectangle (axis cs:12.6,4);
\draw[draw=none, only marks, fill=mycolor4] (axis cs:0,4) rectangle (axis cs:12.6,4);
\draw[draw=none, only marks, fill=mycolor1] (axis cs:12.6,0) rectangle (axis cs:12.8,4);
\draw[draw=none, only marks, fill=mycolor2] (axis cs:12.6,4) rectangle (axis cs:12.8,4);
\draw[draw=none, only marks, fill=mycolor3] (axis cs:12.6,4) rectangle (axis cs:12.8,4);
\draw[draw=none, only marks, fill=mycolor4] (axis cs:12.6,4) rectangle (axis cs:12.8,4);
\draw[draw=none, only marks, fill=mycolor1] (axis cs:12.8,0) rectangle (axis cs:14.7,4);
\draw[draw=none, only marks, fill=mycolor2] (axis cs:12.8,4) rectangle (axis cs:14.7,5);
\draw[draw=none, only marks, fill=mycolor3] (axis cs:12.8,5) rectangle (axis cs:14.7,5);
\draw[draw=none, only marks, fill=mycolor4] (axis cs:12.8,5) rectangle (axis cs:14.7,5);
\draw[draw=none, only marks, fill=mycolor1] (axis cs:14.7,0) rectangle (axis cs:15.5,4);
\draw[draw=none, only marks, fill=mycolor2] (axis cs:14.7,4) rectangle (axis cs:15.5,5);
\draw[draw=none, only marks, fill=mycolor3] (axis cs:14.7,5) rectangle (axis cs:15.5,5);
\draw[draw=none, only marks, fill=mycolor4] (axis cs:14.7,5) rectangle (axis cs:15.5,5);
\draw[draw=none, only marks, fill=mycolor1] (axis cs:15.5,0) rectangle (axis cs:15.8,4);
\draw[draw=none, only marks, fill=mycolor2] (axis cs:15.5,4) rectangle (axis cs:15.8,5);
\draw[draw=none, only marks, fill=mycolor3] (axis cs:15.5,5) rectangle (axis cs:15.8,5);
\draw[draw=none, only marks, fill=mycolor4] (axis cs:15.5,5) rectangle (axis cs:15.8,5);
\draw[draw=none, only marks, fill=mycolor1] (axis cs:15.8,0) rectangle (axis cs:16.8,4);
\draw[draw=none, only marks, fill=mycolor2] (axis cs:15.8,4) rectangle (axis cs:16.8,5);
\draw[draw=none, only marks, fill=mycolor3] (axis cs:15.8,5) rectangle (axis cs:16.8,5);
\draw[draw=none, only marks, fill=mycolor4] (axis cs:15.8,5) rectangle (axis cs:16.8,5);
\draw[draw=none, only marks, fill=mycolor1] (axis cs:16.8,0) rectangle (axis cs:17.5,4);
\draw[draw=none, only marks, fill=mycolor2] (axis cs:16.8,4) rectangle (axis cs:17.5,5);
\draw[draw=none, only marks, fill=mycolor3] (axis cs:16.8,5) rectangle (axis cs:17.5,6);
\draw[draw=none, only marks, fill=mycolor4] (axis cs:16.8,6) rectangle (axis cs:17.5,6);
\draw[draw=none, only marks, fill=mycolor1] (axis cs:17.5,0) rectangle (axis cs:17.6,4);
\draw[draw=none, only marks, fill=mycolor2] (axis cs:17.5,4) rectangle (axis cs:17.6,5);
\draw[draw=none, only marks, fill=mycolor3] (axis cs:17.5,5) rectangle (axis cs:17.6,6);
\draw[draw=none, only marks, fill=mycolor4] (axis cs:17.5,6) rectangle (axis cs:17.6,6);
\draw[draw=none, only marks, fill=mycolor1] (axis cs:17.6,0) rectangle (axis cs:18.5,4);
\draw[draw=none, only marks, fill=mycolor2] (axis cs:17.6,4) rectangle (axis cs:18.5,5);
\draw[draw=none, only marks, fill=mycolor3] (axis cs:17.6,5) rectangle (axis cs:18.5,6);
\draw[draw=none, only marks, fill=mycolor4] (axis cs:17.6,6) rectangle (axis cs:18.5,6);
\draw[draw=none, only marks, fill=mycolor1] (axis cs:18.5,0) rectangle (axis cs:19.5,4);
\draw[draw=none, only marks, fill=mycolor2] (axis cs:18.5,4) rectangle (axis cs:19.5,5);
\draw[draw=none, only marks, fill=mycolor3] (axis cs:18.5,5) rectangle (axis cs:19.5,5);
\draw[draw=none, only marks, fill=mycolor4] (axis cs:18.5,5) rectangle (axis cs:19.5,5);
\draw[draw=none, only marks, fill=mycolor1] (axis cs:19.5,0) rectangle (axis cs:19.7,4);
\draw[draw=none, only marks, fill=mycolor2] (axis cs:19.5,4) rectangle (axis cs:19.7,5);
\draw[draw=none, only marks, fill=mycolor3] (axis cs:19.5,5) rectangle (axis cs:19.7,5);
\draw[draw=none, only marks, fill=mycolor4] (axis cs:19.5,5) rectangle (axis cs:19.7,5);
\draw[draw=none, only marks, fill=mycolor1] (axis cs:19.7,0) rectangle (axis cs:20.8,4);
\draw[draw=none, only marks, fill=mycolor2] (axis cs:19.7,4) rectangle (axis cs:20.8,5);
\draw[draw=none, only marks, fill=mycolor3] (axis cs:19.7,5) rectangle (axis cs:20.8,5);
\draw[draw=none, only marks, fill=mycolor4] (axis cs:19.7,5) rectangle (axis cs:20.8,5);
\draw[draw=none, only marks, fill=mycolor1] (axis cs:20.8,0) rectangle (axis cs:21.8,4);
\draw[draw=none, only marks, fill=mycolor2] (axis cs:20.8,4) rectangle (axis cs:21.8,4);
\draw[draw=none, only marks, fill=mycolor3] (axis cs:20.8,4) rectangle (axis cs:21.8,4);
\draw[draw=none, only marks, fill=mycolor4] (axis cs:20.8,4) rectangle (axis cs:21.8,4);
\draw[draw=none, only marks, fill=mycolor1] (axis cs:21.8,0) rectangle (axis cs:23.3,4);
\draw[draw=none, only marks, fill=mycolor2] (axis cs:21.8,4) rectangle (axis cs:23.3,4);
\draw[draw=none, only marks, fill=mycolor3] (axis cs:21.8,4) rectangle (axis cs:23.3,4);
\draw[draw=none, only marks, fill=mycolor4] (axis cs:21.8,4) rectangle (axis cs:23.3,4);
\draw[draw=none, only marks, fill=mycolor1] (axis cs:23.3,0) rectangle (axis cs:23.5,4);
\draw[draw=none, only marks, fill=mycolor2] (axis cs:23.3,4) rectangle (axis cs:23.5,5);
\draw[draw=none, only marks, fill=mycolor3] (axis cs:23.3,5) rectangle (axis cs:23.5,5);
\draw[draw=none, only marks, fill=mycolor4] (axis cs:23.3,5) rectangle (axis cs:23.5,5);
\draw[draw=none, only marks, fill=mycolor1] (axis cs:23.5,0) rectangle (axis cs:24.5,4);
\draw[draw=none, only marks, fill=mycolor2] (axis cs:23.5,4) rectangle (axis cs:24.5,5);
\draw[draw=none, only marks, fill=mycolor3] (axis cs:23.5,5) rectangle (axis cs:24.5,5);
\draw[draw=none, only marks, fill=mycolor4] (axis cs:23.5,5) rectangle (axis cs:24.5,5);
\draw[draw=none, only marks, fill=mycolor1] (axis cs:24.5,0) rectangle (axis cs:25.8,4);
\draw[draw=none, only marks, fill=mycolor2] (axis cs:24.5,4) rectangle (axis cs:25.8,5);
\draw[draw=none, only marks, fill=mycolor3] (axis cs:24.5,5) rectangle (axis cs:25.8,5);
\draw[draw=none, only marks, fill=mycolor4] (axis cs:24.5,5) rectangle (axis cs:25.8,5);
\draw[draw=none, only marks, fill=mycolor1] (axis cs:25.8,0) rectangle (axis cs:25.9,4);
\draw[draw=none, only marks, fill=mycolor2] (axis cs:25.8,4) rectangle (axis cs:25.9,5);
\draw[draw=none, only marks, fill=mycolor3] (axis cs:25.8,5) rectangle (axis cs:25.9,5);
\draw[draw=none, only marks, fill=mycolor4] (axis cs:25.8,5) rectangle (axis cs:25.9,5);
\draw[draw=none, only marks, fill=mycolor1] (axis cs:25.9,0) rectangle (axis cs:26.3,4);
\draw[draw=none, only marks, fill=mycolor2] (axis cs:25.9,4) rectangle (axis cs:26.3,6);
\draw[draw=none, only marks, fill=mycolor3] (axis cs:25.9,6) rectangle (axis cs:26.3,6);
\draw[draw=none, only marks, fill=mycolor4] (axis cs:25.9,6) rectangle (axis cs:26.3,6);
\draw[draw=none, only marks, fill=mycolor1] (axis cs:26.3,0) rectangle (axis cs:26.8,4);
\draw[draw=none, only marks, fill=mycolor2] (axis cs:26.3,4) rectangle (axis cs:26.8,6);
\draw[draw=none, only marks, fill=mycolor3] (axis cs:26.3,6) rectangle (axis cs:26.8,6);
\draw[draw=none, only marks, fill=mycolor4] (axis cs:26.3,6) rectangle (axis cs:26.8,6);
\draw[draw=none, only marks, fill=mycolor1] (axis cs:26.8,0) rectangle (axis cs:27.3,4);
\draw[draw=none, only marks, fill=mycolor2] (axis cs:26.8,4) rectangle (axis cs:27.3,6);
\draw[draw=none, only marks, fill=mycolor3] (axis cs:26.8,6) rectangle (axis cs:27.3,6);
\draw[draw=none, only marks, fill=mycolor4] (axis cs:26.8,6) rectangle (axis cs:27.3,6);
\draw[draw=none, only marks, fill=mycolor1] (axis cs:27.3,0) rectangle (axis cs:28,4);
\draw[draw=none, only marks, fill=mycolor2] (axis cs:27.3,4) rectangle (axis cs:28,6);
\draw[draw=none, only marks, fill=mycolor3] (axis cs:27.3,6) rectangle (axis cs:28,6);
\draw[draw=none, only marks, fill=mycolor4] (axis cs:27.3,6) rectangle (axis cs:28,6);
\draw[draw=none, only marks, fill=mycolor1] (axis cs:28,0) rectangle (axis cs:28.9,4);
\draw[draw=none, only marks, fill=mycolor2] (axis cs:28,4) rectangle (axis cs:28.9,6);
\draw[draw=none, only marks, fill=mycolor3] (axis cs:28,6) rectangle (axis cs:28.9,6);
\draw[draw=none, only marks, fill=mycolor4] (axis cs:28,6) rectangle (axis cs:28.9,6);
\draw[draw=none, only marks, fill=mycolor1] (axis cs:28.9,0) rectangle (axis cs:29.5,4);
\draw[draw=none, only marks, fill=mycolor2] (axis cs:28.9,4) rectangle (axis cs:29.5,6);
\draw[draw=none, only marks, fill=mycolor3] (axis cs:28.9,6) rectangle (axis cs:29.5,6);
\draw[draw=none, only marks, fill=mycolor4] (axis cs:28.9,6) rectangle (axis cs:29.5,6);
\draw[draw=none, only marks, fill=mycolor1] (axis cs:29.5,0) rectangle (axis cs:29.6,4);
\draw[draw=none, only marks, fill=mycolor2] (axis cs:29.5,4) rectangle (axis cs:29.6,6);
\draw[draw=none, only marks, fill=mycolor3] (axis cs:29.5,6) rectangle (axis cs:29.6,6);
\draw[draw=none, only marks, fill=mycolor4] (axis cs:29.5,6) rectangle (axis cs:29.6,6);
\draw[draw=none, only marks, fill=mycolor1] (axis cs:29.6,0) rectangle (axis cs:29.9,4);
\draw[draw=none, only marks, fill=mycolor2] (axis cs:29.6,4) rectangle (axis cs:29.9,6);
\draw[draw=none, only marks, fill=mycolor3] (axis cs:29.6,6) rectangle (axis cs:29.9,6);
\draw[draw=none, only marks, fill=mycolor4] (axis cs:29.6,6) rectangle (axis cs:29.9,6);
\draw[draw=none, only marks, fill=mycolor1] (axis cs:29.9,0) rectangle (axis cs:30,4);
\draw[draw=none, only marks, fill=mycolor2] (axis cs:29.9,4) rectangle (axis cs:30,6);
\draw[draw=none, only marks, fill=mycolor3] (axis cs:29.9,6) rectangle (axis cs:30,6);
\draw[draw=none, only marks, fill=mycolor4] (axis cs:29.9,6) rectangle (axis cs:30,6);
\draw[draw=none, only marks, fill=mycolor1] (axis cs:30,0) rectangle (axis cs:31,4);
\draw[draw=none, only marks, fill=mycolor2] (axis cs:30,4) rectangle (axis cs:31,6);
\draw[draw=none, only marks, fill=mycolor3] (axis cs:30,6) rectangle (axis cs:31,6);
\draw[draw=none, only marks, fill=mycolor4] (axis cs:30,6) rectangle (axis cs:31,6);
\draw[draw=none, only marks, fill=mycolor1] (axis cs:31,0) rectangle (axis cs:31.3,4);
\draw[draw=none, only marks, fill=mycolor2] (axis cs:31,4) rectangle (axis cs:31.3,6);
\draw[draw=none, only marks, fill=mycolor3] (axis cs:31,6) rectangle (axis cs:31.3,6);
\draw[draw=none, only marks, fill=mycolor4] (axis cs:31,6) rectangle (axis cs:31.3,6);
\draw[draw=none, only marks, fill=mycolor1] (axis cs:31.3,0) rectangle (axis cs:32,4);
\draw[draw=none, only marks, fill=mycolor2] (axis cs:31.3,4) rectangle (axis cs:32,5);
\draw[draw=none, only marks, fill=mycolor3] (axis cs:31.3,5) rectangle (axis cs:32,5);
\draw[draw=none, only marks, fill=mycolor4] (axis cs:31.3,5) rectangle (axis cs:32,5);
\draw[draw=none, only marks, fill=mycolor1] (axis cs:32,0) rectangle (axis cs:32.3,4);
\draw[draw=none, only marks, fill=mycolor2] (axis cs:32,4) rectangle (axis cs:32.3,5);
\draw[draw=none, only marks, fill=mycolor3] (axis cs:32,5) rectangle (axis cs:32.3,5);
\draw[draw=none, only marks, fill=mycolor4] (axis cs:32,5) rectangle (axis cs:32.3,5);
\draw[draw=none, only marks, fill=mycolor1] (axis cs:32.3,0) rectangle (axis cs:32.5,4);
\draw[draw=none, only marks, fill=mycolor2] (axis cs:32.3,4) rectangle (axis cs:32.5,5);
\draw[draw=none, only marks, fill=mycolor3] (axis cs:32.3,5) rectangle (axis cs:32.5,5);
\draw[draw=none, only marks, fill=mycolor4] (axis cs:32.3,5) rectangle (axis cs:32.5,5);
\draw[draw=none, only marks, fill=mycolor1] (axis cs:32.5,0) rectangle (axis cs:33.9,4);
\draw[draw=none, only marks, fill=mycolor2] (axis cs:32.5,4) rectangle (axis cs:33.9,5);
\draw[draw=none, only marks, fill=mycolor3] (axis cs:32.5,5) rectangle (axis cs:33.9,5);
\draw[draw=none, only marks, fill=mycolor4] (axis cs:32.5,5) rectangle (axis cs:33.9,5);
\draw[draw=none, only marks, fill=mycolor1] (axis cs:33.9,0) rectangle (axis cs:34.5,4);
\draw[draw=none, only marks, fill=mycolor2] (axis cs:33.9,4) rectangle (axis cs:34.5,4);
\draw[draw=none, only marks, fill=mycolor3] (axis cs:33.9,4) rectangle (axis cs:34.5,4);
\draw[draw=none, only marks, fill=mycolor4] (axis cs:33.9,4) rectangle (axis cs:34.5,4);
\draw[draw=none, only marks, fill=mycolor1] (axis cs:34.5,0) rectangle (axis cs:34.6,4);
\draw[draw=none, only marks, fill=mycolor2] (axis cs:34.5,4) rectangle (axis cs:34.6,4);
\draw[draw=none, only marks, fill=mycolor3] (axis cs:34.5,4) rectangle (axis cs:34.6,4);
\draw[draw=none, only marks, fill=mycolor4] (axis cs:34.5,4) rectangle (axis cs:34.6,4);
\draw[draw=none, only marks, fill=mycolor1] (axis cs:34.6,0) rectangle (axis cs:34.8,4);
\draw[draw=none, only marks, fill=mycolor2] (axis cs:34.6,4) rectangle (axis cs:34.8,4);
\draw[draw=none, only marks, fill=mycolor3] (axis cs:34.6,4) rectangle (axis cs:34.8,4);
\draw[draw=none, only marks, fill=mycolor4] (axis cs:34.6,4) rectangle (axis cs:34.8,4);
\draw[draw=none, only marks, fill=mycolor1] (axis cs:34.8,0) rectangle (axis cs:34.9,4);
\draw[draw=none, only marks, fill=mycolor2] (axis cs:34.8,4) rectangle (axis cs:34.9,4);
\draw[draw=none, only marks, fill=mycolor3] (axis cs:34.8,4) rectangle (axis cs:34.9,4);
\draw[draw=none, only marks, fill=mycolor4] (axis cs:34.8,4) rectangle (axis cs:34.9,4);
\draw[draw=none, only marks, fill=mycolor1] (axis cs:34.9,0) rectangle (axis cs:35.5,4);
\draw[draw=none, only marks, fill=mycolor2] (axis cs:34.9,4) rectangle (axis cs:35.5,4);
\draw[draw=none, only marks, fill=mycolor3] (axis cs:34.9,4) rectangle (axis cs:35.5,4);
\draw[draw=none, only marks, fill=mycolor4] (axis cs:34.9,4) rectangle (axis cs:35.5,4);
\draw[draw=none, only marks, fill=mycolor1] (axis cs:35.5,0) rectangle (axis cs:36,4);
\draw[draw=none, only marks, fill=mycolor2] (axis cs:35.5,4) rectangle (axis cs:36,4);
\draw[draw=none, only marks, fill=mycolor3] (axis cs:35.5,4) rectangle (axis cs:36,4);
\draw[draw=none, only marks, fill=mycolor4] (axis cs:35.5,4) rectangle (axis cs:36,5);
\draw[draw=none, only marks, fill=mycolor1] (axis cs:36,0) rectangle (axis cs:36.3,4);
\draw[draw=none, only marks, fill=mycolor2] (axis cs:36,4) rectangle (axis cs:36.3,4);
\draw[draw=none, only marks, fill=mycolor3] (axis cs:36,4) rectangle (axis cs:36.3,4);
\draw[draw=none, only marks, fill=mycolor4] (axis cs:36,4) rectangle (axis cs:36.3,5);
\draw[draw=none, only marks, fill=mycolor1] (axis cs:36.3,0) rectangle (axis cs:36.5,4);
\draw[draw=none, only marks, fill=mycolor2] (axis cs:36.3,4) rectangle (axis cs:36.5,4);
\draw[draw=none, only marks, fill=mycolor3] (axis cs:36.3,4) rectangle (axis cs:36.5,4);
\draw[draw=none, only marks, fill=mycolor4] (axis cs:36.3,4) rectangle (axis cs:36.5,5);
\draw[draw=none, only marks, fill=mycolor1] (axis cs:36.5,0) rectangle (axis cs:37,4);
\draw[draw=none, only marks, fill=mycolor2] (axis cs:36.5,4) rectangle (axis cs:37,5);
\draw[draw=none, only marks, fill=mycolor3] (axis cs:36.5,5) rectangle (axis cs:37,5);
\draw[draw=none, only marks, fill=mycolor4] (axis cs:36.5,5) rectangle (axis cs:37,6);
\draw[draw=none, only marks, fill=mycolor1] (axis cs:37,0) rectangle (axis cs:37.3,4);
\draw[draw=none, only marks, fill=mycolor2] (axis cs:37,4) rectangle (axis cs:37.3,5);
\draw[draw=none, only marks, fill=mycolor3] (axis cs:37,5) rectangle (axis cs:37.3,5);
\draw[draw=none, only marks, fill=mycolor4] (axis cs:37,5) rectangle (axis cs:37.3,6);
\draw[draw=none, only marks, fill=mycolor1] (axis cs:37.3,0) rectangle (axis cs:37.4,4);
\draw[draw=none, only marks, fill=mycolor2] (axis cs:37.3,4) rectangle (axis cs:37.4,5);
\draw[draw=none, only marks, fill=mycolor3] (axis cs:37.3,5) rectangle (axis cs:37.4,5);
\draw[draw=none, only marks, fill=mycolor4] (axis cs:37.3,5) rectangle (axis cs:37.4,6);
\draw[draw=none, only marks, fill=mycolor1] (axis cs:37.4,0) rectangle (axis cs:37.5,4);
\draw[draw=none, only marks, fill=mycolor2] (axis cs:37.4,4) rectangle (axis cs:37.5,6);
\draw[draw=none, only marks, fill=mycolor3] (axis cs:37.4,6) rectangle (axis cs:37.5,6);
\draw[draw=none, only marks, fill=mycolor4] (axis cs:37.4,6) rectangle (axis cs:37.5,6);
\draw[draw=none, only marks, fill=mycolor1] (axis cs:37.5,0) rectangle (axis cs:38.5,4);
\draw[draw=none, only marks, fill=mycolor2] (axis cs:37.5,4) rectangle (axis cs:38.5,6);
\draw[draw=none, only marks, fill=mycolor3] (axis cs:37.5,6) rectangle (axis cs:38.5,6);
\draw[draw=none, only marks, fill=mycolor4] (axis cs:37.5,6) rectangle (axis cs:38.5,6);
\draw[draw=none, only marks, fill=mycolor1] (axis cs:38.5,0) rectangle (axis cs:38.9,4);
\draw[draw=none, only marks, fill=mycolor2] (axis cs:38.5,4) rectangle (axis cs:38.9,6);
\draw[draw=none, only marks, fill=mycolor3] (axis cs:38.5,6) rectangle (axis cs:38.9,6);
\draw[draw=none, only marks, fill=mycolor4] (axis cs:38.5,6) rectangle (axis cs:38.9,6);
\draw[draw=none, only marks, fill=mycolor1] (axis cs:38.9,0) rectangle (axis cs:39.4,4);
\draw[draw=none, only marks, fill=mycolor2] (axis cs:38.9,4) rectangle (axis cs:39.4,6);
\draw[draw=none, only marks, fill=mycolor3] (axis cs:38.9,6) rectangle (axis cs:39.4,6);
\draw[draw=none, only marks, fill=mycolor4] (axis cs:38.9,6) rectangle (axis cs:39.4,6);
\draw[draw=none, only marks, fill=mycolor1] (axis cs:39.4,0) rectangle (axis cs:39.5,4);
\draw[draw=none, only marks, fill=mycolor2] (axis cs:39.4,4) rectangle (axis cs:39.5,6);
\draw[draw=none, only marks, fill=mycolor3] (axis cs:39.4,6) rectangle (axis cs:39.5,6);
\draw[draw=none, only marks, fill=mycolor4] (axis cs:39.4,6) rectangle (axis cs:39.5,6);
\draw[draw=none, only marks, fill=mycolor1] (axis cs:39.5,0) rectangle (axis cs:39.9,4);
\draw[draw=none, only marks, fill=mycolor2] (axis cs:39.5,4) rectangle (axis cs:39.9,6);
\draw[draw=none, only marks, fill=mycolor3] (axis cs:39.5,6) rectangle (axis cs:39.9,6);
\draw[draw=none, only marks, fill=mycolor4] (axis cs:39.5,6) rectangle (axis cs:39.9,6);
\draw[draw=none, only marks, fill=mycolor1] (axis cs:39.9,0) rectangle (axis cs:40.4,4);
\draw[draw=none, only marks, fill=mycolor2] (axis cs:39.9,4) rectangle (axis cs:40.4,6);
\draw[draw=none, only marks, fill=mycolor3] (axis cs:39.9,6) rectangle (axis cs:40.4,6);
\draw[draw=none, only marks, fill=mycolor4] (axis cs:39.9,6) rectangle (axis cs:40.4,6);
\draw[draw=none, only marks, fill=mycolor1] (axis cs:40.4,0) rectangle (axis cs:40.5,4);
\draw[draw=none, only marks, fill=mycolor2] (axis cs:40.4,4) rectangle (axis cs:40.5,6);
\draw[draw=none, only marks, fill=mycolor3] (axis cs:40.4,6) rectangle (axis cs:40.5,6);
\draw[draw=none, only marks, fill=mycolor4] (axis cs:40.4,6) rectangle (axis cs:40.5,6);
\draw[draw=none, only marks, fill=mycolor1] (axis cs:40.5,0) rectangle (axis cs:41.3,4);
\draw[draw=none, only marks, fill=mycolor2] (axis cs:40.5,4) rectangle (axis cs:41.3,6);
\draw[draw=none, only marks, fill=mycolor3] (axis cs:40.5,6) rectangle (axis cs:41.3,6);
\draw[draw=none, only marks, fill=mycolor4] (axis cs:40.5,6) rectangle (axis cs:41.3,6);
\draw[draw=none, only marks, fill=mycolor1] (axis cs:41.3,0) rectangle (axis cs:41.4,4);
\draw[draw=none, only marks, fill=mycolor2] (axis cs:41.3,4) rectangle (axis cs:41.4,6);
\draw[draw=none, only marks, fill=mycolor3] (axis cs:41.3,6) rectangle (axis cs:41.4,6);
\draw[draw=none, only marks, fill=mycolor4] (axis cs:41.3,6) rectangle (axis cs:41.4,6);
\draw[draw=none, only marks, fill=mycolor1] (axis cs:41.4,0) rectangle (axis cs:41.5,4);
\draw[draw=none, only marks, fill=mycolor2] (axis cs:41.4,4) rectangle (axis cs:41.5,6);
\draw[draw=none, only marks, fill=mycolor3] (axis cs:41.4,6) rectangle (axis cs:41.5,6);
\draw[draw=none, only marks, fill=mycolor4] (axis cs:41.4,6) rectangle (axis cs:41.5,6);
\draw[draw=none, only marks, fill=mycolor1] (axis cs:41.5,0) rectangle (axis cs:42,4);
\draw[draw=none, only marks, fill=mycolor2] (axis cs:41.5,4) rectangle (axis cs:42,6);
\draw[draw=none, only marks, fill=mycolor3] (axis cs:41.5,6) rectangle (axis cs:42,6);
\draw[draw=none, only marks, fill=mycolor4] (axis cs:41.5,6) rectangle (axis cs:42,6);
\draw[draw=none, only marks, fill=mycolor1] (axis cs:42,0) rectangle (axis cs:43.5,4);
\draw[draw=none, only marks, fill=mycolor2] (axis cs:42,4) rectangle (axis cs:43.5,6);
\draw[draw=none, only marks, fill=mycolor3] (axis cs:42,6) rectangle (axis cs:43.5,6);
\draw[draw=none, only marks, fill=mycolor4] (axis cs:42,6) rectangle (axis cs:43.5,6);
\draw[draw=none, only marks, fill=mycolor1] (axis cs:43.5,0) rectangle (axis cs:44.5,3);
\draw[draw=none, only marks, fill=mycolor2] (axis cs:43.5,3) rectangle (axis cs:44.5,5);
\draw[draw=none, only marks, fill=mycolor3] (axis cs:43.5,5) rectangle (axis cs:44.5,5);
\draw[draw=none, only marks, fill=mycolor4] (axis cs:43.5,5) rectangle (axis cs:44.5,5);
\draw[draw=none, only marks, fill=mycolor1] (axis cs:44.5,0) rectangle (axis cs:45.3,2);
\draw[draw=none, only marks, fill=mycolor2] (axis cs:44.5,2) rectangle (axis cs:45.3,3);
\draw[draw=none, only marks, fill=mycolor3] (axis cs:44.5,3) rectangle (axis cs:45.3,3);
\draw[draw=none, only marks, fill=mycolor4] (axis cs:44.5,3) rectangle (axis cs:45.3,3);
\draw[draw=none, only marks, fill=mycolor1] (axis cs:45.3,0) rectangle (axis cs:45.4,2);
\draw[draw=none, only marks, fill=mycolor2] (axis cs:45.3,2) rectangle (axis cs:45.4,3);
\draw[draw=none, only marks, fill=mycolor3] (axis cs:45.3,3) rectangle (axis cs:45.4,3);
\draw[draw=none, only marks, fill=mycolor4] (axis cs:45.3,3) rectangle (axis cs:45.4,3);
\draw[draw=none, only marks, fill=mycolor1] (axis cs:45.4,0) rectangle (axis cs:45.5,2);
\draw[draw=none, only marks, fill=mycolor2] (axis cs:45.4,2) rectangle (axis cs:45.5,2);
\draw[draw=none, only marks, fill=mycolor3] (axis cs:45.4,2) rectangle (axis cs:45.5,2);
\draw[draw=none, only marks, fill=mycolor4] (axis cs:45.4,2) rectangle (axis cs:45.5,2);
\draw[draw=none, only marks, fill=mycolor1] (axis cs:45.5,0) rectangle (axis cs:46.3,2);
\draw[draw=none, only marks, fill=mycolor2] (axis cs:45.5,2) rectangle (axis cs:46.3,2);
\draw[draw=none, only marks, fill=mycolor3] (axis cs:45.5,2) rectangle (axis cs:46.3,2);
\draw[draw=none, only marks, fill=mycolor4] (axis cs:45.5,2) rectangle (axis cs:46.3,2);
\draw[draw=none, only marks, fill=mycolor1] (axis cs:46.3,0) rectangle (axis cs:46.4,2);
\draw[draw=none, only marks, fill=mycolor2] (axis cs:46.3,2) rectangle (axis cs:46.4,2);
\draw[draw=none, only marks, fill=mycolor3] (axis cs:46.3,2) rectangle (axis cs:46.4,2);
\draw[draw=none, only marks, fill=mycolor4] (axis cs:46.3,2) rectangle (axis cs:46.4,2);
\draw[draw=none, only marks, fill=mycolor1] (axis cs:46.4,0) rectangle (axis cs:46.5,2);
\draw[draw=none, only marks, fill=mycolor2] (axis cs:46.4,2) rectangle (axis cs:46.5,2);
\draw[draw=none, only marks, fill=mycolor3] (axis cs:46.4,2) rectangle (axis cs:46.5,2);
\draw[draw=none, only marks, fill=mycolor4] (axis cs:46.4,2) rectangle (axis cs:46.5,2);
\draw[draw=none, only marks, fill=mycolor1] (axis cs:46.5,0) rectangle (axis cs:46.7,2);
\draw[draw=none, only marks, fill=mycolor2] (axis cs:46.5,2) rectangle (axis cs:46.7,2);
\draw[draw=none, only marks, fill=mycolor3] (axis cs:46.5,2) rectangle (axis cs:46.7,2);
\draw[draw=none, only marks, fill=mycolor4] (axis cs:46.5,2) rectangle (axis cs:46.7,2);
\draw[draw=none, only marks, fill=mycolor1] (axis cs:46.7,0) rectangle (axis cs:47.9,2);
\draw[draw=none, only marks, fill=mycolor2] (axis cs:46.7,2) rectangle (axis cs:47.9,3);
\draw[draw=none, only marks, fill=mycolor3] (axis cs:46.7,3) rectangle (axis cs:47.9,3);
\draw[draw=none, only marks, fill=mycolor4] (axis cs:46.7,3) rectangle (axis cs:47.9,3);
\draw[draw=none, only marks, fill=mycolor1] (axis cs:47.9,0) rectangle (axis cs:48,2);
\draw[draw=none, only marks, fill=mycolor2] (axis cs:47.9,2) rectangle (axis cs:48,3);
\draw[draw=none, only marks, fill=mycolor3] (axis cs:47.9,3) rectangle (axis cs:48,3);
\draw[draw=none, only marks, fill=mycolor4] (axis cs:47.9,3) rectangle (axis cs:48,3);
\draw[draw=none, only marks, fill=mycolor1] (axis cs:48,0) rectangle (axis cs:49.5,2);
\draw[draw=none, only marks, fill=mycolor2] (axis cs:48,2) rectangle (axis cs:49.5,3);
\draw[draw=none, only marks, fill=mycolor3] (axis cs:48,3) rectangle (axis cs:49.5,3);
\draw[draw=none, only marks, fill=mycolor4] (axis cs:48,3) rectangle (axis cs:49.5,3);
\draw[draw=none, only marks, fill=mycolor1] (axis cs:49.5,0) rectangle (axis cs:49.7,2);
\draw[draw=none, only marks, fill=mycolor2] (axis cs:49.5,2) rectangle (axis cs:49.7,3);
\draw[draw=none, only marks, fill=mycolor3] (axis cs:49.5,3) rectangle (axis cs:49.7,3);
\draw[draw=none, only marks, fill=mycolor4] (axis cs:49.5,3) rectangle (axis cs:49.7,3);
\draw[draw=none, only marks, fill=mycolor1] (axis cs:49.7,0) rectangle (axis cs:49.9,2);
\draw[draw=none, only marks, fill=mycolor2] (axis cs:49.7,2) rectangle (axis cs:49.9,3);
\draw[draw=none, only marks, fill=mycolor3] (axis cs:49.7,3) rectangle (axis cs:49.9,3);
\draw[draw=none, only marks, fill=mycolor4] (axis cs:49.7,3) rectangle (axis cs:49.9,3);
\draw[draw=none, only marks, fill=mycolor1] (axis cs:49.9,0) rectangle (axis cs:50,2);
\draw[draw=none, only marks, fill=mycolor2] (axis cs:49.9,2) rectangle (axis cs:50,3);
\draw[draw=none, only marks, fill=mycolor3] (axis cs:49.9,3) rectangle (axis cs:50,3);
\draw[draw=none, only marks, fill=mycolor4] (axis cs:49.9,3) rectangle (axis cs:50,3);
\draw[draw=none, only marks, fill=mycolor1] (axis cs:50,0) rectangle (axis cs:50.4,2);
\draw[draw=none, only marks, fill=mycolor2] (axis cs:50,2) rectangle (axis cs:50.4,3);
\draw[draw=none, only marks, fill=mycolor3] (axis cs:50,3) rectangle (axis cs:50.4,3);
\draw[draw=none, only marks, fill=mycolor4] (axis cs:50,3) rectangle (axis cs:50.4,3);
\draw[draw=none, only marks, fill=mycolor1] (axis cs:50.4,0) rectangle (axis cs:50.5,2);
\draw[draw=none, only marks, fill=mycolor2] (axis cs:50.4,2) rectangle (axis cs:50.5,3);
\draw[draw=none, only marks, fill=mycolor3] (axis cs:50.4,3) rectangle (axis cs:50.5,3);
\draw[draw=none, only marks, fill=mycolor4] (axis cs:50.4,3) rectangle (axis cs:50.5,3);
\draw[draw=none, only marks, fill=mycolor1] (axis cs:50.5,0) rectangle (axis cs:50.7,2);
\draw[draw=none, only marks, fill=mycolor2] (axis cs:50.5,2) rectangle (axis cs:50.7,3);
\draw[draw=none, only marks, fill=mycolor3] (axis cs:50.5,3) rectangle (axis cs:50.7,3);
\draw[draw=none, only marks, fill=mycolor4] (axis cs:50.5,3) rectangle (axis cs:50.7,3);
\draw[draw=none, only marks, fill=mycolor1] (axis cs:50.7,0) rectangle (axis cs:51,2);
\draw[draw=none, only marks, fill=mycolor2] (axis cs:50.7,2) rectangle (axis cs:51,3);
\draw[draw=none, only marks, fill=mycolor3] (axis cs:50.7,3) rectangle (axis cs:51,3);
\draw[draw=none, only marks, fill=mycolor4] (axis cs:50.7,3) rectangle (axis cs:51,3);
\draw[draw=none, only marks, fill=mycolor1] (axis cs:51,0) rectangle (axis cs:52,2);
\draw[draw=none, only marks, fill=mycolor2] (axis cs:51,2) rectangle (axis cs:52,3);
\draw[draw=none, only marks, fill=mycolor3] (axis cs:51,3) rectangle (axis cs:52,3);
\draw[draw=none, only marks, fill=mycolor4] (axis cs:51,3) rectangle (axis cs:52,3);
\draw[draw=none, only marks, fill=mycolor1] (axis cs:52,0) rectangle (axis cs:52.5,2);
\draw[draw=none, only marks, fill=mycolor2] (axis cs:52,2) rectangle (axis cs:52.5,3);
\draw[draw=none, only marks, fill=mycolor3] (axis cs:52,3) rectangle (axis cs:52.5,3);
\draw[draw=none, only marks, fill=mycolor4] (axis cs:52,3) rectangle (axis cs:52.5,3);
\draw[draw=none, only marks, fill=mycolor1] (axis cs:52.5,0) rectangle (axis cs:53.5,2);
\draw[draw=none, only marks, fill=mycolor2] (axis cs:52.5,2) rectangle (axis cs:53.5,3);
\draw[draw=none, only marks, fill=mycolor3] (axis cs:52.5,3) rectangle (axis cs:53.5,3);
\draw[draw=none, only marks, fill=mycolor4] (axis cs:52.5,3) rectangle (axis cs:53.5,3);
\draw[draw=none, only marks, fill=mycolor1] (axis cs:53.5,0) rectangle (axis cs:54.5,2);
\draw[draw=none, only marks, fill=mycolor2] (axis cs:53.5,2) rectangle (axis cs:54.5,3);
\draw[draw=none, only marks, fill=mycolor3] (axis cs:53.5,3) rectangle (axis cs:54.5,3);
\draw[draw=none, only marks, fill=mycolor4] (axis cs:53.5,3) rectangle (axis cs:54.5,3);
\draw[draw=none, only marks, fill=mycolor1] (axis cs:54.5,0) rectangle (axis cs:54.7,2);
\draw[draw=none, only marks, fill=mycolor2] (axis cs:54.5,2) rectangle (axis cs:54.7,3);
\draw[draw=none, only marks, fill=mycolor3] (axis cs:54.5,3) rectangle (axis cs:54.7,3);
\draw[draw=none, only marks, fill=mycolor4] (axis cs:54.5,3) rectangle (axis cs:54.7,3);
\draw[draw=none, only marks, fill=mycolor1] (axis cs:54.7,0) rectangle (axis cs:54.9,2);
\draw[draw=none, only marks, fill=mycolor2] (axis cs:54.7,2) rectangle (axis cs:54.9,2);
\draw[draw=none, only marks, fill=mycolor3] (axis cs:54.7,2) rectangle (axis cs:54.9,2);
\draw[draw=none, only marks, fill=mycolor4] (axis cs:54.7,2) rectangle (axis cs:54.9,2);
\draw[draw=none, only marks, fill=mycolor1] (axis cs:54.9,0) rectangle (axis cs:55.7,2);
\draw[draw=none, only marks, fill=mycolor2] (axis cs:54.9,2) rectangle (axis cs:55.7,2);
\draw[draw=none, only marks, fill=mycolor3] (axis cs:54.9,2) rectangle (axis cs:55.7,2);
\draw[draw=none, only marks, fill=mycolor4] (axis cs:54.9,2) rectangle (axis cs:55.7,2);
\draw[draw=none, only marks, fill=mycolor1] (axis cs:55.7,0) rectangle (axis cs:56,2);
\draw[draw=none, only marks, fill=mycolor2] (axis cs:55.7,2) rectangle (axis cs:56,2);
\draw[draw=none, only marks, fill=mycolor3] (axis cs:55.7,2) rectangle (axis cs:56,2);
\draw[draw=none, only marks, fill=mycolor4] (axis cs:55.7,2) rectangle (axis cs:56,3);
\draw[draw=none, only marks, fill=mycolor1] (axis cs:56,0) rectangle (axis cs:57,2);
\draw[draw=none, only marks, fill=mycolor2] (axis cs:56,2) rectangle (axis cs:57,2);
\draw[draw=none, only marks, fill=mycolor3] (axis cs:56,2) rectangle (axis cs:57,2);
\draw[draw=none, only marks, fill=mycolor4] (axis cs:56,2) rectangle (axis cs:57,3);
\draw[draw=none, only marks, fill=mycolor1] (axis cs:57,0) rectangle (axis cs:57.9,2);
\draw[draw=none, only marks, fill=mycolor2] (axis cs:57,2) rectangle (axis cs:57.9,2);
\draw[draw=none, only marks, fill=mycolor3] (axis cs:57,2) rectangle (axis cs:57.9,2);
\draw[draw=none, only marks, fill=mycolor4] (axis cs:57,2) rectangle (axis cs:57.9,3);
\draw[draw=none, only marks, fill=mycolor1] (axis cs:57.9,0) rectangle (axis cs:58.5,2);
\draw[draw=none, only marks, fill=mycolor2] (axis cs:57.9,2) rectangle (axis cs:58.5,2);
\draw[draw=none, only marks, fill=mycolor3] (axis cs:57.9,2) rectangle (axis cs:58.5,2);
\draw[draw=none, only marks, fill=mycolor4] (axis cs:57.9,2) rectangle (axis cs:58.5,3);
\draw[draw=none, only marks, fill=mycolor1] (axis cs:58.5,0) rectangle (axis cs:58.7,2);
\draw[draw=none, only marks, fill=mycolor2] (axis cs:58.5,2) rectangle (axis cs:58.7,2);
\draw[draw=none, only marks, fill=mycolor3] (axis cs:58.5,2) rectangle (axis cs:58.7,2);
\draw[draw=none, only marks, fill=mycolor4] (axis cs:58.5,2) rectangle (axis cs:58.7,3);
\draw[draw=none, only marks, fill=mycolor1] (axis cs:58.7,0) rectangle (axis cs:59.4,2);
\draw[draw=none, only marks, fill=mycolor2] (axis cs:58.7,2) rectangle (axis cs:59.4,2);
\draw[draw=none, only marks, fill=mycolor3] (axis cs:58.7,2) rectangle (axis cs:59.4,2);
\draw[draw=none, only marks, fill=mycolor4] (axis cs:58.7,2) rectangle (axis cs:59.4,3);
\draw[draw=none, only marks, fill=mycolor1] (axis cs:59.4,0) rectangle (axis cs:59.5,2);
\draw[draw=none, only marks, fill=mycolor2] (axis cs:59.4,2) rectangle (axis cs:59.5,2);
\draw[draw=none, only marks, fill=mycolor3] (axis cs:59.4,2) rectangle (axis cs:59.5,2);
\draw[draw=none, only marks, fill=mycolor4] (axis cs:59.4,2) rectangle (axis cs:59.5,2);
\draw[draw=none, only marks, fill=mycolor1] (axis cs:59.5,0) rectangle (axis cs:59.7,2);
\draw[draw=none, only marks, fill=mycolor2] (axis cs:59.5,2) rectangle (axis cs:59.7,2);
\draw[draw=none, only marks, fill=mycolor3] (axis cs:59.5,2) rectangle (axis cs:59.7,2);
\draw[draw=none, only marks, fill=mycolor4] (axis cs:59.5,2) rectangle (axis cs:59.7,2);
\draw[draw=none, only marks, fill=mycolor1] (axis cs:59.7,0) rectangle (axis cs:60.7,2);
\draw[draw=none, only marks, fill=mycolor2] (axis cs:59.7,2) rectangle (axis cs:60.7,2);
\draw[draw=none, only marks, fill=mycolor3] (axis cs:59.7,2) rectangle (axis cs:60.7,2);
\draw[draw=none, only marks, fill=mycolor4] (axis cs:59.7,2) rectangle (axis cs:60.7,2);
\draw[draw=none, only marks, fill=mycolor1] (axis cs:60.7,0) rectangle (axis cs:62,2);
\draw[draw=none, only marks, fill=mycolor2] (axis cs:60.7,2) rectangle (axis cs:62,2);
\draw[draw=none, only marks, fill=mycolor3] (axis cs:60.7,2) rectangle (axis cs:62,2);
\draw[draw=none, only marks, fill=mycolor4] (axis cs:60.7,2) rectangle (axis cs:62,4);
\draw[draw=none, only marks, fill=mycolor1] (axis cs:62,0) rectangle (axis cs:62.9,2);
\draw[draw=none, only marks, fill=mycolor2] (axis cs:62,2) rectangle (axis cs:62.9,2);
\draw[draw=none, only marks, fill=mycolor3] (axis cs:62,2) rectangle (axis cs:62.9,2);
\draw[draw=none, only marks, fill=mycolor4] (axis cs:62,2) rectangle (axis cs:62.9,3);
\draw[draw=none, only marks, fill=mycolor1] (axis cs:62.9,0) rectangle (axis cs:63.7,1);
\draw[draw=none, only marks, fill=mycolor2] (axis cs:62.9,1) rectangle (axis cs:63.7,1);
\draw[draw=none, only marks, fill=mycolor3] (axis cs:62.9,1) rectangle (axis cs:63.7,1);
\draw[draw=none, only marks, fill=mycolor4] (axis cs:62.9,1) rectangle (axis cs:63.7,2);
\draw[draw=none, only marks, fill=mycolor1] (axis cs:63.7,0) rectangle (axis cs:64.5,0);
\draw[draw=none, only marks, fill=mycolor2] (axis cs:63.7,0) rectangle (axis cs:64.5,0);
\draw[draw=none, only marks, fill=mycolor3] (axis cs:63.7,0) rectangle (axis cs:64.5,0);
\draw[draw=none, only marks, fill=mycolor4] (axis cs:63.7,0) rectangle (axis cs:64.5,1);
\draw[draw=none, only marks, fill=mycolor1] (axis cs:64.5,0) rectangle (axis cs:68.7,0);
\draw[draw=none, only marks, fill=mycolor2] (axis cs:64.5,0) rectangle (axis cs:68.7,0);
\draw[draw=none, only marks, fill=mycolor3] (axis cs:64.5,0) rectangle (axis cs:68.7,0);
\draw[draw=none, only marks, fill=mycolor4] (axis cs:64.5,0) rectangle (axis cs:68.7,0);
\draw[draw=none, only marks, fill=mycolor1] (axis cs:68.7,0) rectangle (axis cs:70,0);
\draw[draw=none, only marks, fill=mycolor2] (axis cs:68.7,0) rectangle (axis cs:70,0);
\draw[draw=none, only marks, fill=mycolor3] (axis cs:68.7,0) rectangle (axis cs:70,0);
\draw[draw=none, only marks, fill=mycolor4] (axis cs:68.7,0) rectangle (axis cs:70,0);
\draw[draw=none, only marks, fill=mycolor1] (axis cs:70,0) rectangle (axis cs:72.5,0);
\draw[draw=none, only marks, fill=mycolor2] (axis cs:70,0) rectangle (axis cs:72.5,0);
\draw[draw=none, only marks, fill=mycolor3] (axis cs:70,0) rectangle (axis cs:72.5,0);
\draw[draw=none, only marks, fill=mycolor4] (axis cs:70,0) rectangle (axis cs:72.5,0);
\addplot [color=blue, dashdotted, line width=2.0pt, forget plot]
  table[row sep=crcr]{%
0	6\\
73.5	6\\
};
\node[right, align=left]
at (axis cs:3,8){$N_R(v_3)$}; 
\draw[draw=none, only marks, fill=mycolor1] (axis cs:2,9) rectangle (axis cs:14,9.5);
\node[right, align=left]
at (axis cs:17.5,8){$N_{e_2,v_3}$}; 
\draw[draw=none, only marks, fill=mycolor2] (axis cs:16.5,9) rectangle (axis cs:28.5,9.5);
\node[right, align=left]
at (axis cs:32,8){$N_{e_5,v_3}$}; 
\draw[draw=none, only marks, fill=mycolor3] (axis cs:31,9) rectangle (axis cs:43,9.5);
\node[right, align=left]
at (axis cs:46.5,8){$N_{e_6,v_3}$}; 
\draw[draw=none, only marks, fill=mycolor4] (axis cs:45.5,9) rectangle (axis cs:57.5,9.5);
\node[right, align=left]
at (axis cs:61,8) {Failure};
\draw[draw=none, only marks, fill=mycolor5] (axis cs:60,9) rectangle (axis cs:72,9.5);
\node[right, align=left]
at (axis cs:65,5) {$C_{v_3}$};
\end{axis}

\begin{axis}[%
width=7cm,
height=3.8cm,
at={(0in,0in)},
scale only axis,
xmin=0,
xmax=1,
ymin=0,
ymax=1,
axis line style={draw=none},
ticks=none,
]
\end{axis}
\end{tikzpicture}%

%% file: plots/case_e5_v2.tikz
%
%
\definecolor{mycolor1}{rgb}{0.49020,0.49804,0.48627}%
\definecolor{mycolor2}{rgb}{0.98824,0.69020,0.00392}%
\definecolor{mycolor3}{rgb}{0.58824,0.43137,0.74118}%
\begin{tikzpicture}

\begin{axis}[%
width=7cm,
height=3.8cm,
at={(0.541in,0.49in)},
scale only axis,
xmin=0,
xmax=73.5,
xlabel style={font=\color{white!15!black}},
xlabel={$t_c$},
ymin=0,
ymax=3.5,
ylabel style={font=\color{white!15!black}},
ylabel={Number of UAVs at $e_5$},
]
\draw[draw=none, only marks, fill=mycolor1] (axis cs:0,0) rectangle (axis cs:12.6,0);
\draw[draw=none, only marks, fill=mycolor1] (axis cs:12.6,0) rectangle (axis cs:12.8,0);
\draw[draw=none, only marks, fill=mycolor1] (axis cs:12.8,0) rectangle (axis cs:14.7,0);
\draw[draw=none, only marks, fill=mycolor1] (axis cs:14.7,0) rectangle (axis cs:15.5,0);
\draw[draw=none, only marks, fill=mycolor1] (axis cs:15.5,0) rectangle (axis cs:15.8,0);
\draw[draw=none, only marks, fill=mycolor1] (axis cs:15.8,0) rectangle (axis cs:16.8,0);
\draw[draw=none, only marks, fill=mycolor1] (axis cs:16.8,0) rectangle (axis cs:17.5,1);
\draw[draw=none, only marks, fill=mycolor2] (axis cs:16.8,0) rectangle (axis cs:17.5,1);
\draw[draw=none, only marks, fill=mycolor1] (axis cs:17.5,0) rectangle (axis cs:17.6,1);
\draw[draw=none, only marks, fill=mycolor2] (axis cs:17.5,0) rectangle (axis cs:17.6,1);
\draw[draw=none, only marks, fill=mycolor1] (axis cs:17.6,0) rectangle (axis cs:18.5,1);
\draw[draw=none, only marks, fill=mycolor2] (axis cs:17.6,0) rectangle (axis cs:18.5,1);
\draw[draw=none, only marks, fill=mycolor1] (axis cs:18.5,0) rectangle (axis cs:19.5,1);
\draw[draw=none, only marks, fill=mycolor3] (axis cs:18.5,0) rectangle (axis cs:19.5,1);
\draw[draw=none, only marks, fill=mycolor1] (axis cs:19.5,0) rectangle (axis cs:19.7,1);
\draw[draw=none, only marks, fill=mycolor3] (axis cs:19.5,0) rectangle (axis cs:19.7,1);
\draw[draw=none, only marks, fill=mycolor1] (axis cs:19.7,0) rectangle (axis cs:20.8,1);
\draw[draw=none, only marks, fill=mycolor3] (axis cs:19.7,0) rectangle (axis cs:20.8,1);
\draw[draw=none, only marks, fill=mycolor1] (axis cs:20.8,0) rectangle (axis cs:21.8,1);
\draw[draw=none, only marks, fill=mycolor3] (axis cs:20.8,0) rectangle (axis cs:21.8,1);
\draw[draw=none, only marks, fill=mycolor1] (axis cs:21.8,0) rectangle (axis cs:23.3,1);
\draw[draw=none, only marks, fill=mycolor3] (axis cs:21.8,0) rectangle (axis cs:23.3,1);
\draw[draw=none, only marks, fill=mycolor1] (axis cs:23.3,0) rectangle (axis cs:23.5,1);
\draw[draw=none, only marks, fill=mycolor3] (axis cs:23.3,0) rectangle (axis cs:23.5,1);
\draw[draw=none, only marks, fill=mycolor1] (axis cs:23.5,0) rectangle (axis cs:24.5,1);
\draw[draw=none, only marks, fill=mycolor3] (axis cs:23.5,0) rectangle (axis cs:24.5,1);
\draw[draw=none, only marks, fill=mycolor1] (axis cs:24.5,0) rectangle (axis cs:25.8,1);
\draw[draw=none, only marks, fill=mycolor3] (axis cs:24.5,0) rectangle (axis cs:25.8,1);
\draw[draw=none, only marks, fill=mycolor1] (axis cs:25.8,0) rectangle (axis cs:25.9,1);
\draw[draw=none, only marks, fill=mycolor3] (axis cs:25.8,0) rectangle (axis cs:25.9,1);
\draw[draw=none, only marks, fill=mycolor1] (axis cs:25.9,0) rectangle (axis cs:26.3,1);
\draw[draw=none, only marks, fill=mycolor3] (axis cs:25.9,0) rectangle (axis cs:26.3,1);
\draw[draw=none, only marks, fill=mycolor1] (axis cs:26.3,0) rectangle (axis cs:26.8,1);
\draw[draw=none, only marks, fill=mycolor3] (axis cs:26.3,0) rectangle (axis cs:26.8,1);
\draw[draw=none, only marks, fill=mycolor1] (axis cs:26.8,0) rectangle (axis cs:27.3,0);
\draw[draw=none, only marks, fill=mycolor1] (axis cs:27.3,0) rectangle (axis cs:28,1);
\draw[draw=none, only marks, fill=mycolor3] (axis cs:27.3,0) rectangle (axis cs:28,1);
\draw[draw=none, only marks, fill=mycolor1] (axis cs:28,0) rectangle (axis cs:28.9,1);
\draw[draw=none, only marks, fill=mycolor3] (axis cs:28,0) rectangle (axis cs:28.9,1);
\draw[draw=none, only marks, fill=mycolor1] (axis cs:28.9,0) rectangle (axis cs:29.5,1);
\draw[draw=none, only marks, fill=mycolor3] (axis cs:28.9,0) rectangle (axis cs:29.5,1);
\draw[draw=none, only marks, fill=mycolor1] (axis cs:29.5,0) rectangle (axis cs:29.6,1);
\draw[draw=none, only marks, fill=mycolor3] (axis cs:29.5,0) rectangle (axis cs:29.6,1);
\draw[draw=none, only marks, fill=mycolor1] (axis cs:29.6,0) rectangle (axis cs:29.9,1);
\draw[draw=none, only marks, fill=mycolor3] (axis cs:29.6,0) rectangle (axis cs:29.9,1);
\draw[draw=none, only marks, fill=mycolor1] (axis cs:29.9,0) rectangle (axis cs:30,2);
\draw[draw=none, only marks, fill=mycolor3] (axis cs:29.9,0) rectangle (axis cs:30,2);
\draw[draw=none, only marks, fill=mycolor1] (axis cs:30,0) rectangle (axis cs:31,2);
\draw[draw=none, only marks, fill=mycolor3] (axis cs:30,0) rectangle (axis cs:31,2);
\draw[draw=none, only marks, fill=mycolor1] (axis cs:31,0) rectangle (axis cs:31.3,2);
\draw[draw=none, only marks, fill=mycolor3] (axis cs:31,0) rectangle (axis cs:31.3,2);
\draw[draw=none, only marks, fill=mycolor1] (axis cs:31.3,0) rectangle (axis cs:32,2);
\draw[draw=none, only marks, fill=mycolor3] (axis cs:31.3,0) rectangle (axis cs:32,2);
\draw[draw=none, only marks, fill=mycolor1] (axis cs:32,0) rectangle (axis cs:32.3,2);
\draw[draw=none, only marks, fill=mycolor3] (axis cs:32,0) rectangle (axis cs:32.3,2);
\draw[draw=none, only marks, fill=mycolor1] (axis cs:32.3,0) rectangle (axis cs:32.5,2);
\draw[draw=none, only marks, fill=mycolor3] (axis cs:32.3,0) rectangle (axis cs:32.5,2);
\draw[draw=none, only marks, fill=mycolor1] (axis cs:32.5,0) rectangle (axis cs:33.9,2);
\draw[draw=none, only marks, fill=mycolor3] (axis cs:32.5,0) rectangle (axis cs:33.9,2);
\draw[draw=none, only marks, fill=mycolor1] (axis cs:33.9,0) rectangle (axis cs:34.5,2);
\draw[draw=none, only marks, fill=mycolor3] (axis cs:33.9,0) rectangle (axis cs:34.5,2);
\draw[draw=none, only marks, fill=mycolor1] (axis cs:34.5,0) rectangle (axis cs:34.6,2);
\draw[draw=none, only marks, fill=mycolor3] (axis cs:34.5,0) rectangle (axis cs:34.6,2);
\draw[draw=none, only marks, fill=mycolor1] (axis cs:34.6,0) rectangle (axis cs:34.8,2);
\draw[draw=none, only marks, fill=mycolor3] (axis cs:34.6,0) rectangle (axis cs:34.8,2);
\draw[draw=none, only marks, fill=mycolor1] (axis cs:34.8,0) rectangle (axis cs:34.9,2);
\draw[draw=none, only marks, fill=mycolor3] (axis cs:34.8,0) rectangle (axis cs:34.9,2);
\draw[draw=none, only marks, fill=mycolor1] (axis cs:34.9,0) rectangle (axis cs:35.5,2);
\draw[draw=none, only marks, fill=mycolor3] (axis cs:34.9,0) rectangle (axis cs:35.5,2);
\draw[draw=none, only marks, fill=mycolor1] (axis cs:35.5,0) rectangle (axis cs:36,2);
\draw[draw=none, only marks, fill=mycolor3] (axis cs:35.5,0) rectangle (axis cs:36,2);
\draw[draw=none, only marks, fill=mycolor1] (axis cs:36,0) rectangle (axis cs:36.3,2);
\draw[draw=none, only marks, fill=mycolor3] (axis cs:36,0) rectangle (axis cs:36.3,2);
\draw[draw=none, only marks, fill=mycolor1] (axis cs:36.3,0) rectangle (axis cs:36.5,2);
\draw[draw=none, only marks, fill=mycolor3] (axis cs:36.3,0) rectangle (axis cs:36.5,2);
\draw[draw=none, only marks, fill=mycolor1] (axis cs:36.5,0) rectangle (axis cs:37,2);
\draw[draw=none, only marks, fill=mycolor3] (axis cs:36.5,0) rectangle (axis cs:37,2);
\draw[draw=none, only marks, fill=mycolor1] (axis cs:37,0) rectangle (axis cs:37.3,2);
\draw[draw=none, only marks, fill=mycolor3] (axis cs:37,0) rectangle (axis cs:37.3,2);
\draw[draw=none, only marks, fill=mycolor1] (axis cs:37.3,0) rectangle (axis cs:37.4,1);
\draw[draw=none, only marks, fill=mycolor3] (axis cs:37.3,0) rectangle (axis cs:37.4,1);
\draw[draw=none, only marks, fill=mycolor1] (axis cs:37.4,0) rectangle (axis cs:37.5,1);
\draw[draw=none, only marks, fill=mycolor3] (axis cs:37.4,0) rectangle (axis cs:37.5,1);
\draw[draw=none, only marks, fill=mycolor1] (axis cs:37.5,0) rectangle (axis cs:38.5,1);
\draw[draw=none, only marks, fill=mycolor3] (axis cs:37.5,0) rectangle (axis cs:38.5,1);
\draw[draw=none, only marks, fill=mycolor1] (axis cs:38.5,0) rectangle (axis cs:38.9,1);
\draw[draw=none, only marks, fill=mycolor3] (axis cs:38.5,0) rectangle (axis cs:38.9,1);
\draw[draw=none, only marks, fill=mycolor1] (axis cs:38.9,0) rectangle (axis cs:39.4,1);
\draw[draw=none, only marks, fill=mycolor3] (axis cs:38.9,0) rectangle (axis cs:39.4,1);
\draw[draw=none, only marks, fill=mycolor1] (axis cs:39.4,0) rectangle (axis cs:39.5,1);
\draw[draw=none, only marks, fill=mycolor3] (axis cs:39.4,0) rectangle (axis cs:39.5,1);
\draw[draw=none, only marks, fill=mycolor1] (axis cs:39.5,0) rectangle (axis cs:39.9,1);
\draw[draw=none, only marks, fill=mycolor3] (axis cs:39.5,0) rectangle (axis cs:39.9,1);
\draw[draw=none, only marks, fill=mycolor1] (axis cs:39.9,0) rectangle (axis cs:40.4,0);
\draw[draw=none, only marks, fill=mycolor1] (axis cs:40.4,0) rectangle (axis cs:40.5,0);
\draw[draw=none, only marks, fill=mycolor1] (axis cs:40.5,0) rectangle (axis cs:41.3,1);
\draw[draw=none, only marks, fill=mycolor3] (axis cs:40.5,0) rectangle (axis cs:41.3,1);
\draw[draw=none, only marks, fill=mycolor1] (axis cs:41.3,0) rectangle (axis cs:41.4,1);
\draw[draw=none, only marks, fill=mycolor3] (axis cs:41.3,0) rectangle (axis cs:41.4,1);
\draw[draw=none, only marks, fill=mycolor1] (axis cs:41.4,0) rectangle (axis cs:41.5,2);
\draw[draw=none, only marks, fill=mycolor3] (axis cs:41.4,0) rectangle (axis cs:41.5,2);
\draw[draw=none, only marks, fill=mycolor1] (axis cs:41.5,0) rectangle (axis cs:42,2);
\draw[draw=none, only marks, fill=mycolor3] (axis cs:41.5,0) rectangle (axis cs:42,2);
\draw[draw=none, only marks, fill=mycolor1] (axis cs:42,0) rectangle (axis cs:43.5,2);
\draw[draw=none, only marks, fill=mycolor3] (axis cs:42,0) rectangle (axis cs:43.5,2);
\draw[draw=none, only marks, fill=mycolor1] (axis cs:43.5,0) rectangle (axis cs:44.5,2);
\draw[draw=none, only marks, fill=mycolor3] (axis cs:43.5,0) rectangle (axis cs:44.5,2);
\draw[draw=none, only marks, fill=mycolor1] (axis cs:44.5,0) rectangle (axis cs:45.3,2);
\draw[draw=none, only marks, fill=mycolor3] (axis cs:44.5,0) rectangle (axis cs:45.3,2);
\draw[draw=none, only marks, fill=mycolor1] (axis cs:45.3,0) rectangle (axis cs:45.4,2);
\draw[draw=none, only marks, fill=mycolor3] (axis cs:45.3,0) rectangle (axis cs:45.4,2);
\draw[draw=none, only marks, fill=mycolor1] (axis cs:45.4,0) rectangle (axis cs:45.5,2);
\draw[draw=none, only marks, fill=mycolor3] (axis cs:45.4,0) rectangle (axis cs:45.5,2);
\draw[draw=none, only marks, fill=mycolor1] (axis cs:45.5,0) rectangle (axis cs:46.3,2);
\draw[draw=none, only marks, fill=mycolor3] (axis cs:45.5,0) rectangle (axis cs:46.3,2);
\draw[draw=none, only marks, fill=mycolor1] (axis cs:46.3,0) rectangle (axis cs:46.4,2);
\draw[draw=none, only marks, fill=mycolor3] (axis cs:46.3,0) rectangle (axis cs:46.4,2);
\draw[draw=none, only marks, fill=mycolor1] (axis cs:46.4,0) rectangle (axis cs:46.5,2);
\draw[draw=none, only marks, fill=mycolor3] (axis cs:46.4,0) rectangle (axis cs:46.5,2);
\draw[draw=none, only marks, fill=mycolor1] (axis cs:46.5,0) rectangle (axis cs:46.7,2);
\draw[draw=none, only marks, fill=mycolor3] (axis cs:46.5,0) rectangle (axis cs:46.7,2);
\draw[draw=none, only marks, fill=mycolor1] (axis cs:46.7,0) rectangle (axis cs:47.9,2);
\draw[draw=none, only marks, fill=mycolor3] (axis cs:46.7,0) rectangle (axis cs:47.9,2);
\draw[draw=none, only marks, fill=mycolor1] (axis cs:47.9,0) rectangle (axis cs:48,2);
\draw[draw=none, only marks, fill=mycolor3] (axis cs:47.9,0) rectangle (axis cs:48,2);
\draw[draw=none, only marks, fill=mycolor1] (axis cs:48,0) rectangle (axis cs:49.5,2);
\draw[draw=none, only marks, fill=mycolor3] (axis cs:48,0) rectangle (axis cs:49.5,2);
\draw[draw=none, only marks, fill=mycolor1] (axis cs:49.5,0) rectangle (axis cs:49.7,2);
\draw[draw=none, only marks, fill=mycolor3] (axis cs:49.5,0) rectangle (axis cs:49.7,2);
\draw[draw=none, only marks, fill=mycolor1] (axis cs:49.7,0) rectangle (axis cs:49.9,2);
\draw[draw=none, only marks, fill=mycolor3] (axis cs:49.7,0) rectangle (axis cs:49.9,2);
\draw[draw=none, only marks, fill=mycolor1] (axis cs:49.9,0) rectangle (axis cs:50,2);
\draw[draw=none, only marks, fill=mycolor3] (axis cs:49.9,0) rectangle (axis cs:50,2);
\draw[draw=none, only marks, fill=mycolor1] (axis cs:50,0) rectangle (axis cs:50.4,2);
\draw[draw=none, only marks, fill=mycolor3] (axis cs:50,0) rectangle (axis cs:50.4,2);
\draw[draw=none, only marks, fill=mycolor1] (axis cs:50.4,0) rectangle (axis cs:50.5,2);
\draw[draw=none, only marks, fill=mycolor3] (axis cs:50.4,0) rectangle (axis cs:50.5,2);
\draw[draw=none, only marks, fill=mycolor1] (axis cs:50.5,0) rectangle (axis cs:50.7,1);
\draw[draw=none, only marks, fill=mycolor3] (axis cs:50.5,0) rectangle (axis cs:50.7,1);
\draw[draw=none, only marks, fill=mycolor1] (axis cs:50.7,0) rectangle (axis cs:51,2);
\draw[draw=none, only marks, fill=mycolor3] (axis cs:50.7,0) rectangle (axis cs:51,2);
\draw[draw=none, only marks, fill=mycolor1] (axis cs:51,0) rectangle (axis cs:52,2);
\draw[draw=none, only marks, fill=mycolor3] (axis cs:51,0) rectangle (axis cs:52,2);
\draw[draw=none, only marks, fill=mycolor1] (axis cs:52,0) rectangle (axis cs:52.5,2);
\draw[draw=none, only marks, fill=mycolor3] (axis cs:52,0) rectangle (axis cs:52.5,2);
\draw[draw=none, only marks, fill=mycolor1] (axis cs:52.5,0) rectangle (axis cs:53.5,2);
\draw[draw=none, only marks, fill=mycolor3] (axis cs:52.5,0) rectangle (axis cs:53.5,2);
\draw[draw=none, only marks, fill=mycolor1] (axis cs:53.5,0) rectangle (axis cs:54.5,1);
\draw[draw=none, only marks, fill=mycolor3] (axis cs:53.5,0) rectangle (axis cs:54.5,1);
\draw[draw=none, only marks, fill=mycolor1] (axis cs:54.5,0) rectangle (axis cs:54.7,1);
\draw[draw=none, only marks, fill=mycolor3] (axis cs:54.5,0) rectangle (axis cs:54.7,1);
\draw[draw=none, only marks, fill=mycolor1] (axis cs:54.7,0) rectangle (axis cs:54.9,1);
\draw[draw=none, only marks, fill=mycolor3] (axis cs:54.7,0) rectangle (axis cs:54.9,1);
\draw[draw=none, only marks, fill=mycolor1] (axis cs:54.9,0) rectangle (axis cs:55.7,1);
\draw[draw=none, only marks, fill=mycolor3] (axis cs:54.9,0) rectangle (axis cs:55.7,1);
\draw[draw=none, only marks, fill=mycolor1] (axis cs:55.7,0) rectangle (axis cs:56,1);
\draw[draw=none, only marks, fill=mycolor3] (axis cs:55.7,0) rectangle (axis cs:56,1);
\draw[draw=none, only marks, fill=mycolor1] (axis cs:56,0) rectangle (axis cs:57,1);
\draw[draw=none, only marks, fill=mycolor3] (axis cs:56,0) rectangle (axis cs:57,1);
\draw[draw=none, only marks, fill=mycolor1] (axis cs:57,0) rectangle (axis cs:57.9,1);
\draw[draw=none, only marks, fill=mycolor3] (axis cs:57,0) rectangle (axis cs:57.9,1);
\draw[draw=none, only marks, fill=mycolor1] (axis cs:57.9,0) rectangle (axis cs:58.5,1);
\draw[draw=none, only marks, fill=mycolor3] (axis cs:57.9,0) rectangle (axis cs:58.5,1);
\draw[draw=none, only marks, fill=mycolor1] (axis cs:58.5,0) rectangle (axis cs:58.7,1);
\draw[draw=none, only marks, fill=mycolor3] (axis cs:58.5,0) rectangle (axis cs:58.7,1);
\draw[draw=none, only marks, fill=mycolor1] (axis cs:58.7,0) rectangle (axis cs:59.4,1);
\draw[draw=none, only marks, fill=mycolor3] (axis cs:58.7,0) rectangle (axis cs:59.4,1);
\draw[draw=none, only marks, fill=mycolor1] (axis cs:59.4,0) rectangle (axis cs:59.5,1);
\draw[draw=none, only marks, fill=mycolor3] (axis cs:59.4,0) rectangle (axis cs:59.5,1);
\draw[draw=none, only marks, fill=mycolor1] (axis cs:59.5,0) rectangle (axis cs:59.7,1);
\draw[draw=none, only marks, fill=mycolor3] (axis cs:59.5,0) rectangle (axis cs:59.7,1);
\draw[draw=none, only marks, fill=mycolor1] (axis cs:59.7,0) rectangle (axis cs:60.7,1);
\draw[draw=none, only marks, fill=mycolor3] (axis cs:59.7,0) rectangle (axis cs:60.7,1);
\draw[draw=none, only marks, fill=mycolor1] (axis cs:60.7,0) rectangle (axis cs:62,0);
\draw[draw=none, only marks, fill=mycolor1] (axis cs:62,0) rectangle (axis cs:62.9,0);
\draw[draw=none, only marks, fill=mycolor1] (axis cs:62.9,0) rectangle (axis cs:63.7,0);
\draw[draw=none, only marks, fill=mycolor1] (axis cs:63.7,0) rectangle (axis cs:64.5,0);
\draw[draw=none, only marks, fill=mycolor1] (axis cs:64.5,0) rectangle (axis cs:68.7,0);
\draw[draw=none, only marks, fill=mycolor1] (axis cs:68.7,0) rectangle (axis cs:70,0);
\draw[draw=none, only marks, fill=mycolor1] (axis cs:70,0) rectangle (axis cs:72.5,0);
\node[right, align=left]
at (axis cs:16,2.7){$N_{e_5,v_3}$}; 
\draw[draw=none, only marks, fill=mycolor2] (axis cs:15,3) rectangle (axis cs:27,3.2);
\node[right, align=left]
at (axis cs:30.5,2.7){$N_{e_5,v_4}$}; 
\draw[draw=none, only marks, fill=mycolor3] (axis cs:29.5,3) rectangle (axis cs:41.5,3.2);
\node[right, align=left]
at (axis cs:45,2.7) {Failure};
\draw[draw=none, only marks, fill=mycolor1] (axis cs:44,3) rectangle (axis cs:56,3.2);
\end{axis}

\begin{axis}[%
width=7cm,
height=3.8cm,
at={(0in,0in)},
scale only axis,
xmin=0,
xmax=1,
ymin=0,
ymax=1,
axis line style={draw=none},
ticks=none
]
\end{axis}
\end{tikzpicture}%

%% file: plots/case_e6_v2.tikz
%
%
\definecolor{mycolor1}{rgb}{0.49020,0.49804,0.48627}%
\definecolor{mycolor2}{rgb}{0.38431,0.34510,0.76863}%
\definecolor{mycolor3}{rgb}{0.99216,0.27451,0.34902}%
\definecolor{mycolor4}{rgb}{0.53725,0.63529,0.01176}%
\begin{tikzpicture}

\begin{axis}[%
width=7cm,
height=3.8cm,
at={(0.52in,0.49in)},
scale only axis,
xmin=0,
xmax=73.5,
xlabel style={font=\color{white!15!black}},
xlabel={${t}_{c}$},
ymin=0,
ymax=4.5,
ylabel style={font=\color{white!15!black}},
ylabel={Number of UAVs at $e_6$}
]
\draw[draw=none, only marks, fill=mycolor1] (axis cs:0,0) rectangle (axis cs:12.6,0);
\draw[draw=none, only marks, fill=mycolor1] (axis cs:12.6,0) rectangle (axis cs:12.8,0);
\draw[draw=none, only marks, fill=mycolor1] (axis cs:12.8,0) rectangle (axis cs:14.7,0);
\draw[draw=none, only marks, fill=mycolor1] (axis cs:14.7,0) rectangle (axis cs:15.5,0);
\draw[draw=none, only marks, fill=mycolor1] (axis cs:15.5,0) rectangle (axis cs:15.8,0);
\draw[draw=none, only marks, fill=mycolor1] (axis cs:15.8,0) rectangle (axis cs:16.8,0);
\draw[draw=none, only marks, fill=mycolor1] (axis cs:16.8,0) rectangle (axis cs:17.5,0);
\draw[draw=none, only marks, fill=mycolor1] (axis cs:17.5,0) rectangle (axis cs:17.6,0);
\draw[draw=none, only marks, fill=mycolor1] (axis cs:17.6,0) rectangle (axis cs:18.5,0);
\draw[draw=none, only marks, fill=mycolor1] (axis cs:18.5,0) rectangle (axis cs:19.5,1);
\draw[draw=none, only marks, fill=mycolor2] (axis cs:18.5,0) rectangle (axis cs:19.5,1);
\draw[draw=none, only marks, fill=mycolor1] (axis cs:19.5,0) rectangle (axis cs:19.7,1);
\draw[draw=none, only marks, fill=mycolor2] (axis cs:19.5,0) rectangle (axis cs:19.7,1);
\draw[draw=none, only marks, fill=mycolor1] (axis cs:19.7,0) rectangle (axis cs:20.8,1);
\draw[draw=none, only marks, fill=mycolor2] (axis cs:19.7,0) rectangle (axis cs:20.8,1);
\draw[draw=none, only marks, fill=mycolor1] (axis cs:20.8,0) rectangle (axis cs:21.8,1);
\draw[draw=none, only marks, fill=mycolor2] (axis cs:20.8,0) rectangle (axis cs:21.8,1);
\draw[draw=none, only marks, fill=mycolor1] (axis cs:21.8,0) rectangle (axis cs:23.3,1);
\draw[draw=none, only marks, fill=mycolor2] (axis cs:21.8,0) rectangle (axis cs:23.3,1);
\draw[draw=none, only marks, fill=mycolor1] (axis cs:23.3,0) rectangle (axis cs:23.5,1);
\draw[draw=none, only marks, fill=mycolor2] (axis cs:23.3,0) rectangle (axis cs:23.5,1);
\draw[draw=none, only marks, fill=mycolor1] (axis cs:23.5,0) rectangle (axis cs:24.5,1);
\draw[draw=none, only marks, fill=mycolor2] (axis cs:23.5,0) rectangle (axis cs:24.5,1);
\draw[draw=none, only marks, fill=mycolor1] (axis cs:24.5,0) rectangle (axis cs:25.8,1);
\draw[draw=none, only marks, fill=mycolor2] (axis cs:24.5,0) rectangle (axis cs:25.8,1);
\draw[draw=none, only marks, fill=mycolor1] (axis cs:25.8,0) rectangle (axis cs:25.9,1);
\draw[draw=none, only marks, fill=mycolor2] (axis cs:25.8,0) rectangle (axis cs:25.9,1);
\draw[draw=none, only marks, fill=mycolor1] (axis cs:25.9,0) rectangle (axis cs:26.3,1);
\draw[draw=none, only marks, fill=mycolor3] (axis cs:25.9,0) rectangle (axis cs:26.3,1);
\draw[draw=none, only marks, fill=mycolor1] (axis cs:26.3,0) rectangle (axis cs:26.8,1);
\draw[draw=none, only marks, fill=mycolor3] (axis cs:26.3,0) rectangle (axis cs:26.8,1);
\draw[draw=none, only marks, fill=mycolor1] (axis cs:26.8,0) rectangle (axis cs:27.3,1);
\draw[draw=none, only marks, fill=mycolor3] (axis cs:26.8,0) rectangle (axis cs:27.3,1);
\draw[draw=none, only marks, fill=mycolor1] (axis cs:27.3,0) rectangle (axis cs:28,1);
\draw[draw=none, only marks, fill=mycolor3] (axis cs:27.3,0) rectangle (axis cs:28,1);
\draw[draw=none, only marks, fill=mycolor1] (axis cs:28,0) rectangle (axis cs:28.9,1);
\draw[draw=none, only marks, fill=mycolor3] (axis cs:28,0) rectangle (axis cs:28.9,1);
\draw[draw=none, only marks, fill=mycolor1] (axis cs:28.9,0) rectangle (axis cs:29.5,1);
\draw[draw=none, only marks, fill=mycolor3] (axis cs:28.9,0) rectangle (axis cs:29.5,1);
\draw[draw=none, only marks, fill=mycolor1] (axis cs:29.5,0) rectangle (axis cs:29.6,0);
\draw[draw=none, only marks, fill=mycolor1] (axis cs:29.6,0) rectangle (axis cs:29.9,0);
\draw[draw=none, only marks, fill=mycolor1] (axis cs:29.9,0) rectangle (axis cs:30,0);
\draw[draw=none, only marks, fill=mycolor1] (axis cs:30,0) rectangle (axis cs:31,0);
\draw[draw=none, only marks, fill=mycolor1] (axis cs:31,0) rectangle (axis cs:31.3,1);
\draw[draw=none, only marks, fill=mycolor3] (axis cs:31,0) rectangle (axis cs:31.3,1);
\draw[draw=none, only marks, fill=mycolor1] (axis cs:31.3,0) rectangle (axis cs:32,1);
\draw[draw=none, only marks, fill=mycolor2] (axis cs:31.3,0) rectangle (axis cs:32,1);
\draw[draw=none, only marks, fill=mycolor1] (axis cs:32,0) rectangle (axis cs:32.3,1);
\draw[draw=none, only marks, fill=mycolor2] (axis cs:32,0) rectangle (axis cs:32.3,1);
\draw[draw=none, only marks, fill=mycolor1] (axis cs:32.3,0) rectangle (axis cs:32.5,1);
\draw[draw=none, only marks, fill=mycolor2] (axis cs:32.3,0) rectangle (axis cs:32.5,1);
\draw[draw=none, only marks, fill=mycolor1] (axis cs:32.5,0) rectangle (axis cs:33.9,1);
\draw[draw=none, only marks, fill=mycolor2] (axis cs:32.5,0) rectangle (axis cs:33.9,1);
\draw[draw=none, only marks, fill=mycolor1] (axis cs:33.9,0) rectangle (axis cs:34.5,1);
\draw[draw=none, only marks, fill=mycolor2] (axis cs:33.9,0) rectangle (axis cs:34.5,1);
\draw[draw=none, only marks, fill=mycolor1] (axis cs:34.5,0) rectangle (axis cs:34.6,1);
\draw[draw=none, only marks, fill=mycolor2] (axis cs:34.5,0) rectangle (axis cs:34.6,1);
\draw[draw=none, only marks, fill=mycolor1] (axis cs:34.6,0) rectangle (axis cs:34.8,1);
\draw[draw=none, only marks, fill=mycolor2] (axis cs:34.6,0) rectangle (axis cs:34.8,1);
\draw[draw=none, only marks, fill=mycolor1] (axis cs:34.8,0) rectangle (axis cs:34.9,1);
\draw[draw=none, only marks, fill=mycolor2] (axis cs:34.8,0) rectangle (axis cs:34.9,1);
\draw[draw=none, only marks, fill=mycolor1] (axis cs:34.9,0) rectangle (axis cs:35.5,1);
\draw[draw=none, only marks, fill=mycolor2] (axis cs:34.9,0) rectangle (axis cs:35.5,1);
\draw[draw=none, only marks, fill=mycolor1] (axis cs:35.5,0) rectangle (axis cs:36,2);
\draw[draw=none, only marks, fill=mycolor4] (axis cs:35.5,0) rectangle (axis cs:36,1);
\draw[draw=none, only marks, fill=mycolor2] (axis cs:35.5,1) rectangle (axis cs:36,2);
\draw[draw=none, only marks, fill=mycolor1] (axis cs:36,0) rectangle (axis cs:36.3,2);
\draw[draw=none, only marks, fill=mycolor4] (axis cs:36,0) rectangle (axis cs:36.3,1);
\draw[draw=none, only marks, fill=mycolor2] (axis cs:36,1) rectangle (axis cs:36.3,2);
\draw[draw=none, only marks, fill=mycolor1] (axis cs:36.3,0) rectangle (axis cs:36.5,2);
\draw[draw=none, only marks, fill=mycolor4] (axis cs:36.3,0) rectangle (axis cs:36.5,1);
\draw[draw=none, only marks, fill=mycolor2] (axis cs:36.3,1) rectangle (axis cs:36.5,2);
\draw[draw=none, only marks, fill=mycolor1] (axis cs:36.5,0) rectangle (axis cs:37,2);
\draw[draw=none, only marks, fill=mycolor4] (axis cs:36.5,0) rectangle (axis cs:37,1);
\draw[draw=none, only marks, fill=mycolor2] (axis cs:36.5,1) rectangle (axis cs:37,2);
\draw[draw=none, only marks, fill=mycolor1] (axis cs:37,0) rectangle (axis cs:37.3,2);
\draw[draw=none, only marks, fill=mycolor4] (axis cs:37,0) rectangle (axis cs:37.3,1);
\draw[draw=none, only marks, fill=mycolor2] (axis cs:37,1) rectangle (axis cs:37.3,2);
\draw[draw=none, only marks, fill=mycolor1] (axis cs:37.3,0) rectangle (axis cs:37.4,2);
\draw[draw=none, only marks, fill=mycolor4] (axis cs:37.3,0) rectangle (axis cs:37.4,1);
\draw[draw=none, only marks, fill=mycolor2] (axis cs:37.3,1) rectangle (axis cs:37.4,2);
\draw[draw=none, only marks, fill=mycolor1] (axis cs:37.4,0) rectangle (axis cs:37.5,2);
\draw[draw=none, only marks, fill=mycolor3] (axis cs:37.4,0) rectangle (axis cs:37.5,2);
\draw[draw=none, only marks, fill=mycolor1] (axis cs:37.5,0) rectangle (axis cs:38.5,2);
\draw[draw=none, only marks, fill=mycolor3] (axis cs:37.5,0) rectangle (axis cs:38.5,2);
\draw[draw=none, only marks, fill=mycolor1] (axis cs:38.5,0) rectangle (axis cs:38.9,2);
\draw[draw=none, only marks, fill=mycolor3] (axis cs:38.5,0) rectangle (axis cs:38.9,2);
\draw[draw=none, only marks, fill=mycolor1] (axis cs:38.9,0) rectangle (axis cs:39.4,2);
\draw[draw=none, only marks, fill=mycolor3] (axis cs:38.9,0) rectangle (axis cs:39.4,2);
\draw[draw=none, only marks, fill=mycolor1] (axis cs:39.4,0) rectangle (axis cs:39.5,2);
\draw[draw=none, only marks, fill=mycolor3] (axis cs:39.4,0) rectangle (axis cs:39.5,2);
\draw[draw=none, only marks, fill=mycolor1] (axis cs:39.5,0) rectangle (axis cs:39.9,2);
\draw[draw=none, only marks, fill=mycolor3] (axis cs:39.5,0) rectangle (axis cs:39.9,2);
\draw[draw=none, only marks, fill=mycolor1] (axis cs:39.9,0) rectangle (axis cs:40.4,2);
\draw[draw=none, only marks, fill=mycolor3] (axis cs:39.9,0) rectangle (axis cs:40.4,2);
\draw[draw=none, only marks, fill=mycolor1] (axis cs:40.4,0) rectangle (axis cs:40.5,3);
\draw[draw=none, only marks, fill=mycolor3] (axis cs:40.4,0) rectangle (axis cs:40.5,3);
\draw[draw=none, only marks, fill=mycolor1] (axis cs:40.5,0) rectangle (axis cs:41.3,3);
\draw[draw=none, only marks, fill=mycolor3] (axis cs:40.5,0) rectangle (axis cs:41.3,3);
\draw[draw=none, only marks, fill=mycolor1] (axis cs:41.3,0) rectangle (axis cs:41.4,3);
\draw[draw=none, only marks, fill=mycolor3] (axis cs:41.3,0) rectangle (axis cs:41.4,3);
\draw[draw=none, only marks, fill=mycolor1] (axis cs:41.4,0) rectangle (axis cs:41.5,3);
\draw[draw=none, only marks, fill=mycolor3] (axis cs:41.4,0) rectangle (axis cs:41.5,2);
\draw[draw=none, only marks, fill=mycolor2] (axis cs:41.4,2) rectangle (axis cs:41.5,3);
\draw[draw=none, only marks, fill=mycolor1] (axis cs:41.5,0) rectangle (axis cs:42,3);
\draw[draw=none, only marks, fill=mycolor3] (axis cs:41.5,0) rectangle (axis cs:42,3);
\draw[draw=none, only marks, fill=mycolor1] (axis cs:42,0) rectangle (axis cs:43.5,2);
\draw[draw=none, only marks, fill=mycolor3] (axis cs:42,0) rectangle (axis cs:43.5,2);
\draw[draw=none, only marks, fill=mycolor1] (axis cs:43.5,0) rectangle (axis cs:44.5,2);
\draw[draw=none, only marks, fill=mycolor2] (axis cs:43.5,0) rectangle (axis cs:44.5,2);
\draw[draw=none, only marks, fill=mycolor1] (axis cs:44.5,0) rectangle (axis cs:45.3,2);
\draw[draw=none, only marks, fill=mycolor2] (axis cs:44.5,0) rectangle (axis cs:45.3,2);
\draw[draw=none, only marks, fill=mycolor1] (axis cs:45.3,0) rectangle (axis cs:45.4,2);
\draw[draw=none, only marks, fill=mycolor2] (axis cs:45.3,0) rectangle (axis cs:45.4,2);
\draw[draw=none, only marks, fill=mycolor1] (axis cs:45.4,0) rectangle (axis cs:45.5,2);
\draw[draw=none, only marks, fill=mycolor2] (axis cs:45.4,0) rectangle (axis cs:45.5,2);
\draw[draw=none, only marks, fill=mycolor1] (axis cs:45.5,0) rectangle (axis cs:46.3,2);
\draw[draw=none, only marks, fill=mycolor2] (axis cs:45.5,0) rectangle (axis cs:46.3,2);
\draw[draw=none, only marks, fill=mycolor1] (axis cs:46.3,0) rectangle (axis cs:46.4,2);
\draw[draw=none, only marks, fill=mycolor2] (axis cs:46.3,0) rectangle (axis cs:46.4,2);
\draw[draw=none, only marks, fill=mycolor1] (axis cs:46.4,0) rectangle (axis cs:46.5,2);
\draw[draw=none, only marks, fill=mycolor2] (axis cs:46.4,0) rectangle (axis cs:46.5,2);
\draw[draw=none, only marks, fill=mycolor1] (axis cs:46.5,0) rectangle (axis cs:46.7,1);
\draw[draw=none, only marks, fill=mycolor2] (axis cs:46.5,0) rectangle (axis cs:46.7,1);
\draw[draw=none, only marks, fill=mycolor1] (axis cs:46.7,0) rectangle (axis cs:47.9,1);
\draw[draw=none, only marks, fill=mycolor2] (axis cs:46.7,0) rectangle (axis cs:47.9,1);
\draw[draw=none, only marks, fill=mycolor1] (axis cs:47.9,0) rectangle (axis cs:48,1);
\draw[draw=none, only marks, fill=mycolor2] (axis cs:47.9,0) rectangle (axis cs:48,1);
\draw[draw=none, only marks, fill=mycolor1] (axis cs:48,0) rectangle (axis cs:49.5,1);
\draw[draw=none, only marks, fill=mycolor2] (axis cs:48,0) rectangle (axis cs:49.5,1);
\draw[draw=none, only marks, fill=mycolor1] (axis cs:49.5,0) rectangle (axis cs:49.7,1);
\draw[draw=none, only marks, fill=mycolor2] (axis cs:49.5,0) rectangle (axis cs:49.7,1);
\draw[draw=none, only marks, fill=mycolor1] (axis cs:49.7,0) rectangle (axis cs:49.9,1);
\draw[draw=none, only marks, fill=mycolor2] (axis cs:49.7,0) rectangle (axis cs:49.9,1);
\draw[draw=none, only marks, fill=mycolor1] (axis cs:49.9,0) rectangle (axis cs:50,1);
\draw[draw=none, only marks, fill=mycolor2] (axis cs:49.9,0) rectangle (axis cs:50,1);
\draw[draw=none, only marks, fill=mycolor1] (axis cs:50,0) rectangle (axis cs:50.4,1);
\draw[draw=none, only marks, fill=mycolor2] (axis cs:50,0) rectangle (axis cs:50.4,1);
\draw[draw=none, only marks, fill=mycolor1] (axis cs:50.4,0) rectangle (axis cs:50.5,1);
\draw[draw=none, only marks, fill=mycolor2] (axis cs:50.4,0) rectangle (axis cs:50.5,1);
\draw[draw=none, only marks, fill=mycolor1] (axis cs:50.5,0) rectangle (axis cs:50.7,1);
\draw[draw=none, only marks, fill=mycolor2] (axis cs:50.5,0) rectangle (axis cs:50.7,1);
\draw[draw=none, only marks, fill=mycolor1] (axis cs:50.7,0) rectangle (axis cs:51,1);
\draw[draw=none, only marks, fill=mycolor2] (axis cs:50.7,0) rectangle (axis cs:51,1);
\draw[draw=none, only marks, fill=mycolor1] (axis cs:51,0) rectangle (axis cs:52,2);
\draw[draw=none, only marks, fill=mycolor2] (axis cs:51,0) rectangle (axis cs:52,2);
\draw[draw=none, only marks, fill=mycolor1] (axis cs:52,0) rectangle (axis cs:52.5,2);
\draw[draw=none, only marks, fill=mycolor2] (axis cs:52,0) rectangle (axis cs:52.5,2);
\draw[draw=none, only marks, fill=mycolor1] (axis cs:52.5,0) rectangle (axis cs:53.5,2);
\draw[draw=none, only marks, fill=mycolor2] (axis cs:52.5,0) rectangle (axis cs:53.5,2);
\draw[draw=none, only marks, fill=mycolor1] (axis cs:53.5,0) rectangle (axis cs:54.5,2);
\draw[draw=none, only marks, fill=mycolor2] (axis cs:53.5,0) rectangle (axis cs:54.5,2);
\draw[draw=none, only marks, fill=mycolor1] (axis cs:54.5,0) rectangle (axis cs:54.7,2);
\draw[draw=none, only marks, fill=mycolor2] (axis cs:54.5,0) rectangle (axis cs:54.7,2);
\draw[draw=none, only marks, fill=mycolor1] (axis cs:54.7,0) rectangle (axis cs:54.9,2);
\draw[draw=none, only marks, fill=mycolor2] (axis cs:54.7,0) rectangle (axis cs:54.9,2);
\draw[draw=none, only marks, fill=mycolor1] (axis cs:54.9,0) rectangle (axis cs:55.7,2);
\draw[draw=none, only marks, fill=mycolor2] (axis cs:54.9,0) rectangle (axis cs:55.7,2);
\draw[draw=none, only marks, fill=mycolor1] (axis cs:55.7,0) rectangle (axis cs:56,2);
\draw[draw=none, only marks, fill=mycolor4] (axis cs:55.7,0) rectangle (axis cs:56,1);
\draw[draw=none, only marks, fill=mycolor2] (axis cs:55.7,1) rectangle (axis cs:56,2);
\draw[draw=none, only marks, fill=mycolor1] (axis cs:56,0) rectangle (axis cs:57,2);
\draw[draw=none, only marks, fill=mycolor4] (axis cs:56,0) rectangle (axis cs:57,1);
\draw[draw=none, only marks, fill=mycolor2] (axis cs:56,1) rectangle (axis cs:57,2);
\draw[draw=none, only marks, fill=mycolor1] (axis cs:57,0) rectangle (axis cs:57.9,2);
\draw[draw=none, only marks, fill=mycolor4] (axis cs:57,0) rectangle (axis cs:57.9,1);
\draw[draw=none, only marks, fill=mycolor2] (axis cs:57,1) rectangle (axis cs:57.9,2);
\draw[draw=none, only marks, fill=mycolor1] (axis cs:57.9,0) rectangle (axis cs:58.5,2);
\draw[draw=none, only marks, fill=mycolor4] (axis cs:57.9,0) rectangle (axis cs:58.5,1);
\draw[draw=none, only marks, fill=mycolor2] (axis cs:57.9,1) rectangle (axis cs:58.5,2);
\draw[draw=none, only marks, fill=mycolor1] (axis cs:58.5,0) rectangle (axis cs:58.7,2);
\draw[draw=none, only marks, fill=mycolor4] (axis cs:58.5,0) rectangle (axis cs:58.7,1);
\draw[draw=none, only marks, fill=mycolor2] (axis cs:58.5,1) rectangle (axis cs:58.7,2);
\draw[draw=none, only marks, fill=mycolor1] (axis cs:58.7,0) rectangle (axis cs:59.4,2);
\draw[draw=none, only marks, fill=mycolor4] (axis cs:58.7,0) rectangle (axis cs:59.4,1);
\draw[draw=none, only marks, fill=mycolor2] (axis cs:58.7,1) rectangle (axis cs:59.4,2);
\draw[draw=none, only marks, fill=mycolor1] (axis cs:59.4,0) rectangle (axis cs:59.5,2);
\draw[draw=none, only marks, fill=mycolor2] (axis cs:59.4,0) rectangle (axis cs:59.5,2);
\draw[draw=none, only marks, fill=mycolor1] (axis cs:59.5,0) rectangle (axis cs:59.7,2);
\draw[draw=none, only marks, fill=mycolor2] (axis cs:59.5,0) rectangle (axis cs:59.7,2);
\draw[draw=none, only marks, fill=mycolor1] (axis cs:59.7,0) rectangle (axis cs:60.7,2);
\draw[draw=none, only marks, fill=mycolor2] (axis cs:59.7,0) rectangle (axis cs:60.7,2);
\draw[draw=none, only marks, fill=mycolor1] (axis cs:60.7,0) rectangle (axis cs:62,2);
\draw[draw=none, only marks, fill=mycolor4] (axis cs:60.7,0) rectangle (axis cs:62,2);
\draw[draw=none, only marks, fill=mycolor1] (axis cs:62,0) rectangle (axis cs:62.9,1);
\draw[draw=none, only marks, fill=mycolor4] (axis cs:62,0) rectangle (axis cs:62.9,1);
\draw[draw=none, only marks, fill=mycolor1] (axis cs:62.9,0) rectangle (axis cs:63.7,1);
\draw[draw=none, only marks, fill=mycolor4] (axis cs:62.9,0) rectangle (axis cs:63.7,1);
\draw[draw=none, only marks, fill=mycolor1] (axis cs:63.7,0) rectangle (axis cs:64.5,1);
\draw[draw=none, only marks, fill=mycolor4] (axis cs:63.7,0) rectangle (axis cs:64.5,1);
\draw[draw=none, only marks, fill=mycolor1] (axis cs:64.5,0) rectangle (axis cs:68.7,0);
\draw[draw=none, only marks, fill=mycolor1] (axis cs:68.7,0) rectangle (axis cs:70,0);
\draw[draw=none, only marks, fill=mycolor1] (axis cs:70,0) rectangle (axis cs:72.5,0);
\node[right, align=left]
at (axis cs:16,3.7){$N_{e_6,v_3}$}; 
\draw[draw=none, only marks, fill=mycolor4] (axis cs:15,4) rectangle (axis cs:27,4.2);
\node[right, align=left]
at (axis cs:30.5,3.7){$N_{e_6,v_4}$}; 
\draw[draw=none, only marks, fill=mycolor3] (axis cs:29.5,4) rectangle (axis cs:41.5,4.2);
\node[right, align=left]
at (axis cs:45,3.7){$N_{e_6,v_7}$}; 
\draw[draw=none, only marks, fill=mycolor2] (axis cs:44,4) rectangle (axis cs:56,4.2);
\node[right, align=left]
at (axis cs:59.5,3.7) {Failure};
\draw[draw=none, only marks, fill=mycolor1] (axis cs:58.5,4) rectangle (axis cs:70.5,4.2);
\end{axis}

\begin{axis}[%
width=7cm,
height=3.8cm,
at={(0in,0in)},
scale only axis,
xmin=0,
xmax=1,
ymin=0,
ymax=1,
axis line style={draw=none},
ticks=none
]
\end{axis}
\end{tikzpicture}%

%% file: plots/size_time.tikz
%
%
\definecolor{mycolor1}{rgb}{0.00000,0.44700,0.74100}%
\begin{tikzpicture}

\begin{axis}[%
width=7cm,
height=3cm,
at={(0.52in,0.42in)},
scale only axis,
xmin=0,
xmax=1000,
xlabel={Schedule Size},
ymin=0,
ymax=50,
ytick={0,10,...,50},
ylabel style={align=center}, ylabel=Computation Time \\ (seconds)
]
\addplot [color=mycolor1, line width=2.0pt, forget plot]
  table[row sep=crcr]{%
20	1.663871583\\
100	3.98503225\\
200	7.29867375\\
300	9.88211875\\
400	14.735578\\
500	19.392094\\
600	23.92268775\\
700	29.382114875\\
800	36.735090291\\
900	42.738392\\
1000	49.006763084\\
};
\addplot [color=black, line width=2.0pt, only marks, mark=o, mark options={solid, black}, forget plot]
  table[row sep=crcr]{%
20	1.663871583\\
100	3.98503225\\
200	7.29867375\\
300	9.88211875\\
400	14.735578\\
500	19.392094\\
600	23.92268775\\
700	29.382114875\\
800	36.735090291\\
900	42.738392\\
1000	49.006763084\\
};
\end{axis}
\end{tikzpicture}%